\documentclass{ws-rmp}

\usepackage[english]{babel}
\usepackage{graphicx}
\usepackage{textcomp}
\usepackage{amsmath}
\usepackage{amssymb}

\newcommand{\Th}{\operatorname{th}}
\newcommand{\aTh}{\operatorname{ath}}
\newcommand{\ch}{\operatorname{ch}}
\newcommand{\sh}{\operatorname{sh}}

\newcommand{\id}{{\bf 1}}

\newcommand{\Strip}{\mbox{Strip}}

\newcommand{\Hol}{\mbox{Hol}}

\newcommand{\Spec}{\mbox{Spec}}

\newcommand{\supp}{\mbox{supp}}

\newcommand{\Dom}{\mbox{Dom}}

\begin{document}

\title{A UNIFIED MODE DECOMPOSITION METHOD FOR PHYSICAL FIELDS IN HOMOGENEOUS COSMOLOGY}

\author{ZHIRAYR G. AVETISYAN}

\address{Max Planck Institute for Mathematics in the Sciences, Inselstr. 22,\\
04103 Leipzig, Germany\\
\email{jirayrag@gmail.com}}

\address{Institut f\"ur Theoretische Physik, Universit\"at Leipzig, Br\"uderstr. 16,\\
04103 Leipzig, Germany}

\maketitle

\begin{abstract}
The methods of mode decomposition and Fourier analysis of
classical and quantum fields on curved spacetimes previously
available mainly for the scalar field on Friedman-Robertson-Walker
(FRW) spacetimes are extended to arbitrary vector bundle fields on
general spatially homogeneous spacetimes. This is done by
developing a rigorous unified framework which incorporates mode
decomposition, harmonic analysis and Fourier analysis. The limits
of applicability and uniqueness of mode decomposition by
separation of the time variable in the field equation are found.
It is shown how mode decomposition can be naturally extended to
weak solutions of the field equation under some analytical
assumptions. It is further shown that these assumptions can always
be fulfilled if the vector bundle under consideration is analytic.
The propagator of the field equation is explicitly mode
decomposed. A short survey on the geometry of the models
considered in mathematical cosmology is given and it is concluded
that practically all of them can be represented by a semidirect
homogeneous vector bundle. Abstract harmonic analytical Fourier
transform is introduced in semidirect homogeneous spaces and it is
explained how it can be related to the spectral Fourier transform.
The general form of invariant bi-distributions on semidirect
homogeneous spaces is found in the Fourier space which generalizes
earlier results for the homogeneous states of the scalar field on
FRW spacetimes.
\end{abstract}



\section{Introduction}

As long as mankind is not in possession of a successful and
commonly accepted quantum theory of gravitation (and possibly even
thereafter), the quantum field theory on curved spacetimes (QFT in
CST) is an adequate and consistent theoretical framework for
astronomy and cosmology. Vicinities of black holes and the early
epoch of the universe are two prominent physical situations where
gravity is sufficiently strong so that its influence on the
quantum field theory cannot be neglected. At the same time, in
these situations the gravity is sufficiently uniform (i.e., its
local fluctuations are negligible) to be considered classical and
interacting with matter quantum fields mainly macroscopically.
This semiclassical picture is captured by letting quantum fields
propagate on a curved spacetime. The back reaction of the matter
on gravity is described by the semiclassical Einstein equation,
where gravity feels only the expectation values of quantum
entities. This is the domain of quantum field theory on curved
spacetimes, apart from its intrinsic appeal for the beauty and
variety of fundamental problems it poses in mathematical physics.
QFT in CST adopts the more modern algebraic quantum field theory
setup and, as appropriate to mathematical physics, attempts to be
as axiomatic and deductive as possible and mathematically
rigorous. The disadvantage of such an approach is the extreme
difficulty of producing explicit ready-made results which can be
applied on the observational level, and each such result can be
considered as a remarkable success. For instance, it was not until
2010 when the authors of \cite{Degner_Verch_2010} obtained the
first completely rigorous and at the same time explicit
description of cosmological particle creation in states of low
energy. This was done under several assumptions which can or
cannot be considered realistic in cosmology. Namely, the
Klein-Gordon field on FRW spacetimes was chosen, and homogeneity
and sufficient regularity was stipulated a priori for the desired
state of low energy. When one tries to step a bit beyond these
restrictions one immediately faces severe mathematical
difficulties along the entire way from the very setup until the
final expressions. The reason is that the chain of results used in
these constructions has been obtained only under the above
mentioned assumptions. The aim of the current work is the
extension of some of those mathematical methods to a generality
where they can be applied for practically all realistic
cosmological situations. To which extent this program has been
successful will become clear below.

A primary tool for obtaining explicit constructions are geometric
symmetries. After publishing his eminent work on general
relativity, Einstein declared he had no hope to see explicit
solutions of his equation in the near future. It was the rich
symmetry of the FRW spacetimes that allowed Friedmann to find
first explicit solutions and thus to dispel the despair of
Einstein shortly after his publication. This instant can be
considered as the birth of modern mathematical cosmology, which
until today remains one of the main appliers of explicit solutions
in general relativity. One of the merits of geometric symmetries
is the possibility of the separation of variables in the field
equation which helps to obtain explicit solutions. The mode
decomposition of the solutions of the field equation (also
referred to as the Fourier method in PDE, or expansion into
harmonic oscillators in physics) was probably first applied in the
cosmological context by Parker \cite{Parker1969} who performed it
on the flat FRW spacetime. The idea of the method is that one
tries to separate the time variable in the field equation, and
looks for solutions as linear combinations of products $X(\vec
x)T(t)$ where $X$ depends only on the spatial coordinates and $T$
only on time. What Parker discovered is that this is possible on
FRW spacetimes and represents a very handy tool for the analysis
of the dynamics. A thorough analytical investigation of the method
in the cosmological context was conducted in \cite{Fulling1989},
where an abstract functional analytical eigenfunction expansion
was introduced as a methodological background, and precise methods
were suggested for the mode decomposition of regular solutions on
FRW and ultrastatic spacetimes. The theory of the method does not
seem to have been developed any further until nowadays. In
particular, the following questions remain open. What are the
precise limits of applicability of the mode decomposition by means
of separation of the time variable? How many different
decompositions are possible for the same geometrical setup? When
and how can the decomposition be extended to weak (distributional)
solutions?

In the first part of this work satisfactory answers will be given
to these questions and establish a unified framework for the
method. Our geometrical setup will be a finite dimensional vector
bundle over an arbitrary globally hyperbolic Lorentzian manifold
furnished with a pseudo-Riemannian fiber metric and a fiber metric
linear connection. The field equation will then be given by the
field operator $D=\Box^\nabla+m^\star(x)$ where $\Box^\nabla$ is
the connection d'Alambert operator acting on the smooth sections,
and $m^\star(x)$ will be the variable smooth "mass term" (possibly
including a coupling to scalar curvature) to which mild
assumptions will be imposed. This seems to be the most general
setup of a (symmetric) hyperbolic linear field on a curved
spacetime, and covers most practical situations in the
cosmological context.

The results can be briefly described as follows. Precise
geometrical necessary and sufficient conditions are obtained for
the mode decomposition of smooth solutions by time separation to
be realizable. This mode decomposition is given, as perhaps
expected, by the time dependent Fourier transform, and is shown to
be basically the only such mode decomposition possible. The
decomposition is extended to all distributional solutions in a
natural manner given that there exists a choice of modes
fulfilling certain regularity conditions. The conditions become
fully explicit once one has a Paley-Wiener theorem for the spatial
Fourier transform, i.e., a precise description of the Fourier
image of the test functions space. This is the situation in FRW
spaces. Moreover, it is shown that if the bundle is analytic and
the dynamics of the geometry sufficiently rigid (precise
definitions are given) then the conditions are satisfied
regardless of the harmonic analysis involved. It turns out that
that the mode solutions of non scalar fields under certain
circumstances experience infrared instability periods not known
for scalar fields; the author is yet not sure about the physical
essence of this phenomenon. Apart from this, traditionally
infrared integrability issues arise when integrating modes over
the spectrum $\Spec(D_{\Sigma})$ of a Schr\"odinger operator
$D_\Sigma$ if $\overline{\Spec(D_{\Sigma})}$ includes the
eigenvalue $0$ \cite{Fewster:2003ey}. Here this question of
integrability of modes over the spectrum is settled by showing
that the suitably chosen mode solutions remain well under control
even at non positive spectral values. As an example of application
and as a byproduct the explicit formula of the propagator of the
field in the Fourier space is found, which generalizes one
obtained in \cite{Lueders_Roberts_1990}.

Another advantage of geometric symmetries is the possibility to
apply harmonic analysis. This is particularly true for the
cosmological models where a rather rich group of spatial
isometries is imposed. Then the spatial sections of the spacetime
can be considered as homogeneous spaces, and the spatial Fourier
transform can be investigated in much more detail with many
explicit consequences. These properties then can be dragged to the
time dependent Fourier transform and hence to the mode
decomposition. Of particular interest in the quantum field theory
on cosmological spacetimes are the spatially symmetric (invariant)
states, of which the 2-point functions are bi-distributional
solutions of the field equation which are invariant under the full
isometry group. In \cite{Lueders_Roberts_1990} the Fourier image
of the isotropic homogeneous states of the Klein-Gordon field on
FRW spacetime has been obtained, under an additional continuity
requirement which has no clear physical interpretation. In
contrast the most general form of invariant scalar
bi-distributions on $\mathbb{R}^d$ has been obtained in
\cite{Gelfand_Vilenkin1964} using a nice technique. What appears
to be missing is a generalization of these results to sufficiently
many homogeneous spaces so that practically all cosmological
situations are covered. The harmonic analysis of FRW symmetry
groups is well known since long, but strictly speaking isotropy is
not as fundamental in cosmology as homogeneity, and one is also
interested in cosmological models which are only homogeneous
(Bianchi models) or in addition only partially isotropic (LRS
models). The isometry groups of these spaces are described by
Bianchi groups with their quotients and semidirect extensions (in
case of LRS models). Some of these groups are solvable, others are
semisimple, with finite or infinite center. Therefore it is not
easy to establish a unified harmonic analytical approach for all
cases, although one has to admit that Kirillov's theory for the
solvable groups and Helgason's theory for semisimple groups
together would cover the majority of situations. To obtain a
unified theory one can adopt abstract harmonic analysis. This
beautiful branch of mathematics allows to obtain many results in
an admirable generality. However, apart from compact groups, it is
not completely clear how to relate the abstract group Fourier
transform with the eigenfunction expansion of the invariant
Laplace operator. At least there seems to be no unified exposition
of these techniques applied in the cosmological context.

The aim of our second part will be to put together some tools from
harmonic analysis which are adequate in cosmology, and to obtain
useful results with their help. First a short survey on
homogeneous bundles in general, and on semidirect homogeneous
bundles in particular is carried out, and it is shown that they
cover the vast and the most important majority of the realistic
cosmological structures. Next, the abstract harmonic analytical
Fourier transform is introduced with its requisites on semidirect
homogeneous spaces. Although the abstract Fourier transform on
groups and the representation theory in homogeneous spaces are to
be considered as well studied and widely known, the abstract
Fourier transform on homogeneous spaces is not that popular and
deserves a better exposition (at least we were not able to find a
satisfactory one in the literature). Then some properties of the
Fourier transformed distributions are established. Although some
of these results may be known to experts in harmonic analysis, we
were not able to locate them in the for required for our purposes
in the literature. And because these properties are needed to
obtain our results on invariant distributions they are included
here with proofs. Next an attempt is made to unify the abstract
Fourier transform and the usual spatial Fourier transform given by
the eigenfunction expansion of the Laplace operator. Several
remarks are made on this way, which may serve as a guideline to
completely build the desired correspondence once a particular
structure is specified. This has been indeed performed for the
purely homogeneous cosmological models, which will appear in a
subsequent publication. The necessity of such an explicit
correspondence consists in the ability to transfer the results
obtained in the abstract setup to the situation with the usual
Fourier transform, which is far more useful a tool for concrete
calculations. Finally, by a generalization of the above mentioned
technique in \cite{Gelfand_Vilenkin1964} the general form of the
invariant bi-distributions in arbitrary semidirect homogeneous
vector bundle is found without any additional assumptions on their
regularity. It is concluded that the polynomial bound of the
Fourier transformed homogeneous state as found in
\cite{Lueders_Roberts_1990} is a consequence of the imposed
regularity requirements.

\section{Mode Decomposition of Hyperbolic Fields}

\subsection{Linear hyperbolic fields}

It is generally believed that the forces of nature are described
by tensor and spinor fields. A geometrical generalization of those
are the vector bundle fields, i.e., fields as smooth sections of
some vector bundles. In general relativity one works mainly on a
four dimensional Lorentzian smooth manifold $(M,g)$ which is
called a spacetime. We will be concerned with hyperbolic fields
given by a wave equation, hence we put an additional constraint on
the spacetime $(M,g)$ to be globally hyperbolic, so that the
Cauchy problem of the wave equation is well-posed. For simplicity
only linear fields will be discussed here. For the reduction of
the Maxwell and Proca fields to linear hyperbolic fields the
reader is referred to
\cite{Fewster:2003ey},\cite{BarGinouxPfaffle200703}. We summarize
the basic setup of the the linear hyperbolic fields in the
following section.

Let $V$ be an $n$-dimensional vector space. Let
$\mathcal{T}\xrightarrow[]{\pi}M$ be a vector bundle with standard
fiber $V$ and with a pseudo-Riemannian metric $\langle
u,v\rangle_\mathfrak{g}$. Let further
$\mathcal{E}(\mathcal{T})=C^\infty(\mathcal{T})$ and
$\mathcal{D}(\mathcal{T})=C_0^\infty(\mathcal{T})$ be the spaces
of smooth sections and of those with compact support,
correspondingly. Let $\nabla$ be a metric connection on
$\mathcal{T}$ and $\Box^\nabla$ the associated d'Alambert operator
on $\mathcal{E}(\mathcal{T})$. Define the field operator
\index{Field operator} to be the normal hyperbolic operator
$D=\Box^\nabla+m^\star(x)$ acting on $\mathcal{E}(\mathcal{T})$,
where $m^\star\in C^\infty(M)$ is a generalization of the usual
mass term $m^2$ which now can also contain the coupling term $\xi
R$. Note that because differential operators are
support-decreasing,
$D\mathcal{D}(\mathcal{T})\subset\mathcal{D}(\mathcal{T})$. A free
linear hyperbolic field $\phi\in\mathcal{E}(\mathcal{T})$ is a
solution of the field equation $D\phi=0$.

Being a globally hyperbolic spacetime,
$M=\mathcal{I}\times\Sigma$, where
$\mathcal{I}\subseteq\mathbb{R}$ is an interval, and for each
$t\in\mathcal{I}$ the hypersurface $\Sigma_t\sim\Sigma$ is a three
dimensional embedded Riemannian submanifold, which is spacelike
with respect to $g$ and is a Cauchy surface in the sense described
below. Thanks to \cite{Bernal_Sanchez_2005} one can choose a
smooth global time function $t$ and a coordinate atlas such that
$x=(t,\vec x)=(x_0,x_1,x_2,x_3)$ where $t\in\mathcal{I}$ and $\vec
x\in\Sigma$, i.e., $\Sigma_t$ are equal $t$ hypersurfaces. The
restriction of the bundle $\mathcal{T}$ to the submanifold
$\Sigma_t$ will be denoted by $\mathcal{T}_t=\pi^{-1}(\Sigma_t)$.
The spaces of smooth sections will be
$\mathcal{E}(\mathcal{T}_t)=C^\infty(\mathcal{T}_t)$ and
$\mathcal{D}(\mathcal{T}_t)=C_0^\infty(\mathcal{T}_t)$. If
$i_t:\mathcal{T}_t\to\mathcal{T}$ is the identical embedding, then
its pullback $i_t^*$ is the restriction map for objects on
$\mathcal{T}$ to $\mathcal{T}_t$. In particular
$i_t^*:\mathcal{E}(\mathcal{T})\to\mathcal{E}(\mathcal{T}_t)$ and
$i_t^*:\mathcal{D}(\mathcal{T})\to\mathcal{D}(\mathcal{T}_t)$ are
linear surjective maps. The embedding $\pi\circ
i_t\circ\pi^{-1}:M\to\Sigma$ gives rise to a natural embedding
$i_t:TM\to T\Sigma$ and of all tensor bundles (using the same
symbols $i_t$, $i^*_t$ for different restrictions in the spirit of
polymorphism should not lead to a confusion). The Riemannian
metric $h$ on $T\Sigma$ will be $h=-i^*_t(g)$, with minus sign
here because of the signature convention $(+,-,-,-)$. The
restriction $i^*_t(\nabla)=\nabla_{i^*_t(.)}=\nabla^t$ is a metric
connection on $\mathcal{T}_t$. The associated Laplace operator
$\Delta_t=\Delta^{\nabla^t}$ is an elliptic operator on
$\mathcal{E}(\mathcal{T}_t)$ (so that $-\Delta_t$ is a positive
operator). The restriction of the field operator $D$ to
$\mathcal{E}(\mathcal{T}_t)$ will be denoted by
$D_{\Sigma_t}=-\Delta_t+m^\star(x)$.

An existence and uniqueness theorem
\cite{BarGinouxPfaffle200703},\cite{Gunther1988},\cite{Dimock1980}
for wave operators tells that the Cauchy problem is well posed:
there exists a bijective linear map
$$
\mathcal{E}(\mathcal{T}_t)\oplus
\mathcal{E}(\mathcal{T}_t)\ni(f_0,f_1)\to
\jmath(f_0,f_1)\in\left\{f\in\mathcal{E}(\mathcal{T})\mbox{:
}Df=0\right\}
$$
such that $f_0=i_t^*(f)$ and $f_1=i_t^*(\nabla_tf)$, where
$\nabla_t=\nabla_{\frac{\partial}{\partial t}}$. Furthermore,
there exist unique Green's operators
$E^\pm:\mathcal{D}(\mathcal{T})\to\mathcal{E}(\mathcal{T})$
satisfying $DE^\pm=E^\pm D=id_{\mathcal{D}(\mathcal{T})}$ and
$supp\{G^\pm f\}\subset J^\pm(supp\{f\})$ for all
$f\in\mathcal{D}(\mathcal{T})$. Here $J^\pm(N)$ with a subset
$N\subset M$ denotes the causal future/past of $N$. Define by
$E=E^+-E^-$ the {\it propagator}\index{Propagator} of $D$, which
satisfies $DE=ED=0$. Now
$Sol(\mathcal{T})=\jmath(\mathcal{E}(\mathcal{T}_t)\oplus
\mathcal{E}(\mathcal{T})_t)$ and
$Sol_0(\mathcal{T})=\jmath(\mathcal{D}(\mathcal{T}_t)\oplus
\mathcal{D}(\mathcal{T}_t))$ will denote correspondingly the
spaces of all smooth solutions, and of those satisfying
$supp\{f\}\cap\Sigma_t$ compact for all $t\in\mathcal{I}$,
respectively. Then $E\mathcal{D}(\mathcal{T})\subset
Sol_0(\mathcal{T})$. There is a symplectic form on
$Sol_0(\mathcal{T})$:
$$
\sigma(u,v)=\int_{\Sigma_t}d\mu_h\left[\langle
i_t^*(u),i_t^*(\nabla_tv)\rangle_g-\langle
i_t^*(\nabla_tu),i_t^*(v)\rangle_g\right]\mbox{, }\forall u,v\in
Sol_0(\mathcal{T})\mbox{, }\forall t\in\mathcal{I},
$$
where $h=-i_t^*(g)$ is the induced Riemannian metric on
$\Sigma_t$. That this is conserved (analogous to a Wronskian in
ODE) can be seen by considering the Green's identity for $u,v\in
Sol_0(\mathcal{T})$ on the regular cylindric region
$U=(t_1;t_2)\times\Sigma\subset M$ for any $t_1\neq t_2$,
$$
0=\int_Ud\mu_g\left[\langle u,Dv\rangle_g-\langle
Du,v\rangle_g\right]=
$$
$$
=\int_{\partial U}d\mu_h\left[\langle
i_t^*(u),i_t^*(\nabla_tv)\rangle_g-\langle
i_t^*(\nabla_tu),i_t^*(v)\rangle_g\right]\mbox{, }\forall u,v\in
Sol_0(\mathcal{T}).
$$
This identity also helps us along with Green's operators to find
the explicit form of the map $\jmath$. Given any $v\in
Sol_0(\mathcal{T})$, $f\in\mathcal{D}(\mathcal{T})$, we apply it
two times; once for the pair $v,u=E^+(f)$ on the region
$U^+=(-\inf\{\mathcal{I}\};t)$ and once for the pair $v,u=E^-(f)$
on the region $U^+=(t;\inf\{\mathcal{I}\})$. Summing up the
resulting two identities and using the support properties of
$E^\pm$ we arrive at
\begin{eqnarray}
\int_Md\mu_g\langle
v,f\rangle_g=\sigma(v,E(f)).\label{CauchyExplicit}
\end{eqnarray}
We see that the functional
$\sigma(v,E(.)):\mathcal{D}(\mathcal{T})\to\mathbb{C}$ actually is
given by a smooth integral kernel, which equals $v$. Thus we can
write symbolically
$$
\jmath(f_0,f_1)[y]=\int_{\Sigma_t}d\mu_h\left[\langle
f_0,\nabla_tE(y)\rangle_g-\langle f_1,E(y)\rangle_g\right]\mbox{,
}\forall f_0,f_1\in\mathcal{D}(\mathcal{T}_t)\mbox{, }y\in
M\mbox{, }t\in\mathcal{I}.
$$
For full details of this last computation the reader is referred
to \cite{DIMOCK_1992}, where the argument is given for 1-forms,
but is readily applicable to our more general case.

\begin{proposition}\label{ESurj}
The operator $E:\mathcal{D}(\mathcal{T})\to Sol_0(\mathcal{T})$ is
surjective.
\end{proposition}

\begin{proof}
Let $v\in Sol_0$, and let $K_v=\supp
v\cap\left([0,1]\times\Sigma\right)$ be the compact region of its
support between times 0 and 1. Let further $\chi\in\mathcal{E}(M)$
be a smooth function which equals 1 for $t<0$ and 0 for $t>1$.
Denote $v^-=-v\chi$ and $v^+=v(1-\chi)$, then $v=v^+-v^-$. Let
$f_v=Dv^+$, then $\supp f_v\subset K_v$ is compact, hence
$f_v\in\mathcal{D}(\mathcal{T})$. The equation $f_v=Dv^+$ has a
unique solution with past compact support, and it is given by
$v^+=E^+f_v$. Now $Dv^-=-Dv+Dv^+=f_v$, and similarly $v^-=E^-f_v$.
Then $v=E^+f_v-E^-f_v=Ef_v$. The arbitrariness of $\chi$ reflects
the non-injectivity of $E$.
\end{proof}

\subsection{Spectral mode decomposition}

Henceforth we will use nomenclature introduced in the Appendix A
without special notice. Consider the operators
$D:\mathcal{D}(\mathcal{T})\to\mathcal{D}(\mathcal{T})$ and
$D_\Sigma:\mathcal{D}(\mathcal{T}_t)\to\mathcal{D}(\mathcal{T}_t)$.
If $m^\star(x)\in\mathbb{R}$ everywhere on $M$, then by the virtue
of Green's identity $D$ and $D_{\Sigma_t}$ are formally
self-adjoint with respect to the inner products $(,)_M$ and
$(,)_{\Sigma_t}$. We will not need the self-adjointness of $D$ in
the current work. The constructions below will pertain mainly to
$D_{\Sigma_t}$. The conditions on $m^\star(x)$ for $D_{\Sigma_t}$
to have a self-adjoint extension can be found in
\cite{CyconFroeseKirschSimon200801}. We moreover require that the
operator $D_{\Sigma_t}$ be lower semi-bounded. In practice will be
mainly interested in cosmological models, where
$m^\star(x)=m^\star(t)$ is a function of time only, so that no
problems arise. Below we assume self-adjoint extensions for both
$D$ and $D_{\Sigma_t}$, but for $D$ this is only symbolic and
targets simply at coherent notations.

Let $D$ and $D_{\Sigma_t}$ be extended to self-adjoint operators
on $L^2(\mathcal{T})$ and $L^2(\mathcal{T}_t)$ respectively. In
the rigged Hilbert spaces
\cite{Gelfand_Vilenkin1964},\cite{Maurin1972},\cite{Maurin1968}
$\mathcal{D}(\mathcal{T})\subset
L^2(\mathcal{T})\subset\mathcal{D}(\mathcal{T})'$ and
$\mathcal{D}(\mathcal{T}_t)\subset
L^2(\mathcal{T}_t)\subset\mathcal{D}(\mathcal{T}_t)'$ operators
$D$ and $D_{\Sigma_t}$ possess complete systems of eigenfunctions
$\{u_\rho\}$ and $\{\zeta_\lambda\}$ satisfying
$$
Du_\rho=\rho u_\rho\mbox{,
}u_\rho\in\mathcal{D}(\mathcal{T})'\mbox{, }\rho\in\mathbb{R},
$$
$$
D_{\Sigma_t}\zeta_\lambda=\lambda\zeta_\lambda\mbox{,
}\zeta_\lambda\in\mathcal{D}(\mathcal{T}_t)'\mbox{,
}\lambda\in\mathbb{R}.
$$
Denote by $\mathcal{D}(\mathcal{T})_\rho'$ and
$\mathcal{D}(\mathcal{T}_t)_\lambda'$ the linear spaces of
eigenfunctions corresponding to $\rho$ and $\lambda$,
respectively. Furthermore, there exists an isomorphism
\begin{eqnarray}
L^2(\mathcal{T}_t)=\int_\mathbb{R}^\oplus
d\nu(\lambda)H(\lambda),\label{SpectralDecomp}
\end{eqnarray}
where
$$
D_{\Sigma_t}|_{H(\lambda)}=\lambda,
$$
and $d\nu(\lambda)$ is a positive measure. Each $H(\lambda)$ is
continuously embedded in $\mathcal{D}(\mathcal{T}_t)_\lambda'$.

The eigenfunction expansion\index{Eigenfunction expansion} of $D$
will be the map
$$
\mathcal{D}(\mathcal{T})\ni f\to\tilde
f_\rho\in(\mathcal{D}(\mathcal{T})_\rho')'^*\mbox{,
}\rho\in\mathbb{R},
$$
($X'^*$ denotes the space of continuous antilinear functionals on
the space $X$) where $\tilde f_\rho$ is defined by
$$
\tilde f_\rho(u_\rho)=\bar u_\rho(f)\mbox{, }\forall
f\in\mathcal{D}(\mathcal{T})\mbox{,
}u_\rho\in\mathcal{D}(\mathcal{T})_\rho'.
$$
(Here we defer a little from Gelfand's notations who puts $\tilde
f_\rho(u_\rho)=u_\rho(f)$.)The expansion of $D_{\Sigma_t}$ on
$\mathcal{D}(\mathcal{T}_t)$ is constructed similarly. Note that
$D_{\Sigma_t}$ is an elliptic operator, hence
$\mathcal{D}(\mathcal{T}_t)_\lambda'\subset\mathcal{E}(\mathcal{T}_t)$.

If each $\mathcal{D}(\mathcal{T}_t)_\lambda'$ is finite
dimensional (eigenvalue $\lambda$ has a finite multiplicity
$N_\lambda$), then
$$
H(\lambda)=\mathcal{D}(\mathcal{T}_t)_\lambda'\mbox{, }\dim
H(\lambda)=N_\lambda.
$$
Choose $\{\zeta_{\lambda,i}\}_{i=1}^{N_\lambda}$ be a an
orthonormal basis in $\mathcal{D}(\mathcal{T}_t)_\lambda'$
(orthonormality understood in $H(\lambda)$). Then
$(\mathcal{D}(\mathcal{T}_t)_\lambda')'^*\sim\mathbb{C}^{N_\lambda}$
by the bijective linear map
$$
\tilde f(\zeta_\lambda)=\tilde
f\left(\sum_{i=1}^{N_\lambda}c_i\zeta_{\lambda,i}\right)\to\{\tilde
f_i=\tilde f(\zeta_{\lambda,i})\}_{i=1}^{N_\lambda}\mbox{,
}\forall \tilde f\in (\mathcal{D}(\mathcal{T}_t)_\lambda')'^*,
$$
where each $\tilde f_i\in\mathbb{C}$. In particular, if $\tilde
f_\lambda$ is the mode expansion of
$f\in\mathcal{D}(\mathcal{T}_t)$, then the map
\begin{eqnarray}
\mathcal{D}(\mathcal{T}_t)\ni f\to\tilde f_\lambda\to\{\tilde
f_{\lambda,i}\}\in\int_\mathbb{R}^\oplus
d\nu(\lambda)\mathbb{C}^{N_\lambda}\label{PreFourier}
\end{eqnarray}
will serve as a Fourier transform on $\mathcal{D}(\mathcal{T}_t)$.
Define
$$
Spec\{D_{\Sigma_t}\}=\supp\{d\nu\},
$$
and
$$
\tilde\Sigma=\{(\lambda,i)\mbox{: }\lambda\in
Spec\{D_{\Sigma_t}\}\mbox{, }i=1,...,N_\lambda\}.
$$
Define the spectral measure on $\tilde\Sigma$ as
$$
d\mu(\alpha)=d\nu(\lambda)\times d\sharp(i),
$$
where $d\sharp$ is the counting measure. The map
(Eq.\ref{PreFourier}) can be reformulated as
$$
\tilde f(\alpha)=\mathcal{F}[f](\alpha)\mbox{,
}f\in\mathcal{D}(\mathcal{T}_t).
$$
Then the formula (Eq.\ref{SpectralDecomp}) arises a Plancherel
formula
$$
(f,h)_{\Sigma_t}=\int_{\tilde\Sigma}d\mu(\alpha)\bar{\tilde{f}}(\alpha)\tilde
g(\alpha),
$$
and a Peter-Weyl (or Fourier inversion) formula
\begin{eqnarray} f(x)=\int_{\tilde\Sigma}d\mu(\alpha)\tilde
f(\alpha)\zeta_\alpha(x),\label{PeterWeyl}
\end{eqnarray}
which holds in the $L^2$-sense so far. In our cases of interest
this convergence will be in the compact topology.

However, if $\mathcal{D}(\mathcal{T}_t)_\lambda'$ is infinite
dimensional, more delicate tools are needed to obtain a Fourier
transform with desired properties. Such tools naturally include an
investigation of symmetries of the underlying geometrical
structure, and this is the subject of the harmonic analysis. We
will often use the formal structure (Eq.\ref{PeterWeyl}) without
mentioning a concrete realization, assuming that this is possible.
For the cases of our interest we will indeed find a realization by
means of adapted Fourier transform.

In the theory of Fourier transform, and in particular in the
Euclidean case, the Paley-Wiener theorems describe the functional
analytical structure of the image
$\mathcal{F}[\mathcal{D}(\mathcal{T}_t)]$ of the test function
space under the action of the Fourier transform. This description
is very useful when analyzing the situation in the Fourier space.
Unfortunately there is no (at least known to us) general
Paley-Wiener argument valid for any Fourier transform arisen in
this manner, and the proofs of the existing ones are rather
structure-specific. In applications we would like, however, to
obtain results which are valid in a large variety of cases, and
therefore we will introduce a notion of 'conventional' Fourier
transform which consists of a number of assumptions pertaining to
the analytical properties of a given Fourier transform. Some of
our later results will be valid under the assumption that the
eigenvalue expansion of the operator $\Delta_t$ has at least some
of the properties of a conventional Fourier transform. One says
that a good definition is an assumption of a theorem. In this
sense the following is not a good definition as we will not manage
to use all properties in this work. However it seems feasible that
these properties will become useful for several applications in
quantum field theory.

\begin{definition}
A Fourier transform $\mathcal{F}$ given by the eigenfunction
expansion against a complete system
$\{\zeta_\alpha\}_{\alpha\in\tilde\Sigma}$ will be called
conventional if
\begin{romanlist}[iii]
\item The Fourier space (or momentum space) $\tilde\Sigma$ is a
manifold consisting of $n=\dim V$ components,
$\tilde\Sigma=\bigcup_{i=1}^n\tilde\Sigma^i$, and each component
$\tilde\Sigma^i$ is either a discreet set or an (not necessarily
connected) analytical manifold

\item The eigenvalue $\lambda(\alpha)$ is an analytic function on
$\tilde\Sigma$

\item The range $\mathcal{F}[\mathcal{D}(\mathcal{T}_t)]$ is a
subspace of the space of analytic functions $\tilde f(\alpha)$ on
$\tilde\Sigma$ with rapid decay in $\lambda$

\item There is an involution $\alpha\to-\alpha$ on $\tilde\Sigma$
such that $\zeta_{-\alpha}=\bar\zeta_\alpha$.
\end{romanlist}
\end{definition}

Note that it follows $\lambda(-\alpha)=\lambda(\alpha)$. In later
sections we will give harmonic analytical justifications for such
a 'conjecture'. This conjecture is anticipated, in particular, for
all cosmological models. Moreover, being true for FRW spaces, it
can be proven also for Bianchi I-VII spaces (to appear in a future
publication).

Further we will be mainly interested in the space of weak
solutions of the field equation, $\mathcal{D}(\mathcal{T})_0'$,
and will try to find a convenient characterization of it. In
particular we will be looking for a complete system of solutions
$\{u_\alpha\}$ spanning $\mathcal{D}(\mathcal{T})_0'$ and being in
addition well handled (i.e., smooth, explicit etc.). One means of
doing this is to look at a subspace of
$\mathcal{D}(\mathcal{T})_0'$ which consist of solutions
$f(x)=a(t)b(\vec x)$, $a\in C^\infty(\mathcal{I})$,
$b\in\mathcal{E}(\mathcal{T}_t)$. Then under fortunate
circumstances the field equation breaks apart into two lower
dimensional elliptical eigenproblems, which are much easier to
deal with. Which are those circumstances and whether such
solutions span $\mathcal{D}(\mathcal{T})_0'$, and related
questions, are the matter of the problem of variable separation.
In the next sections we will find out in which cases this is
possible and how to perform it.

\subsection{Separation of variables}

As discussed above, we would like to span the space
$\mathcal{D}(\mathcal{T})_0'$ of weak solutions of the field
equation by a family of easily computable smooth solutions
$\{u_\alpha\}$. In this section we will see when and how one can
perform this for the smooth solutions $Sol_0(\mathcal{T})$. The
necessary requisites for this will be predominantly geometric
requirements. In the next section we will show that under
additional functional analytical assumptions the procedure can be
extended to $\mathcal{D}(\mathcal{T})_0'$ in a natural way.

\begin{definition} Let $S$ be a subspace of $\mathcal{E}(\mathcal{T})$ with closure $\bar S\supseteq S$,
$\mathfrak{M}$ a measure space with measure $d\mathfrak{m}$. An
$\mathfrak{M}$-measurable family
$\{u_\alpha\}_{\alpha\in\mathfrak{M}}$ of elements
$u_\alpha\in\bar S$ will be called a {\it complete} or {\it
spanning system} for $S$ if for any $v\in S$ there exists a unique
(modulo null-supported functions) $\mathfrak{M}$-measurable
function $a^v:\mathfrak{M}\to\mathbb{R}$
($a^v:\mathfrak{M}\to\mathbb{C}$) such that
$$
v=\int_\mathfrak{M}d\mathfrak{m}(\alpha)a^v(\alpha)u_\alpha.
$$
\end{definition}
For the details on integration of nuclear Frech\'et space-valued
functions see \cite{Thomas1975} and references therein. We will
always take $\mathfrak{M}$ to be {\it minimal}, i.e., there exists
no subset $A\subset\mathfrak{M}$ with $\mathfrak{m}(A)>0$ such
that $a^v(A)=0$ for all $v\in S$. If the uniqueness requirement is
relaxed, then $\{u_\alpha\}_{\alpha\in\mathfrak{M}}$ will be
called a {\it redundant complete system} for $S$. Note that from
the uniqueness property it follows, that for
$d\mathfrak{m}$-almost all $\alpha\in\mathfrak{M}$, there exists
no $\alpha\neq\beta\in\mathfrak{M}$ with
$u_\alpha+p(\alpha)u_\beta=0$, $p(\alpha)\neq0$ a number. In other
words, almost all $u_\alpha$ are pairwise independent.

$\mathcal{E}(\mathcal{T})$ is a closed topological vector space
with the topology of compact convergence, and $Sol(\mathcal{T})$
and $Sol_0(\mathcal{T})$ are linear subspaces. A spanning system
$\{u_\alpha\}_{\alpha\in\mathfrak{M}}$ of $Sol_0(\mathcal{T})$ of
the form $u_\alpha=T_\alpha X_\alpha$, where $T_\alpha\neq\bar
T_\alpha\in C^\infty(\mathcal{I})$ ($T_\alpha$ and $\bar T_\alpha$
linearly independent) and $X_\alpha\in\mathcal{E}(\mathcal{T}_t)$,
such that $DT_\alpha X_\alpha=D\bar T_\alpha X_\alpha=0$, will be
called a {\it complete (time-)variable separated system of
solutions} (or shorter, {\it separating system}).\index{Separating
system}

We will assume that a Fourier transform $\mathcal{F}$ on
$\mathcal{D}(\mathcal{T}_t)$ is specified by means of the spectral
decomposition of $D_{\Sigma_t}$ as described in the previous
section. The system of eigenfunctions $\{\zeta^t_\alpha\}$ of
$\Delta_t$, with the Fourier space $\tilde\Sigma_t$ and the
spectral measure $d\mu(\alpha)$ on it, provide a spanning system
for $\mathcal{D}(\mathcal{T}_t)$ by means of the Fourier inversion
(or Peter-Weyl) formula. Below we will come across the question of
a spectral theory of formally non-self-adjoint, i.e., asymmetric
differential operators of type $a(x)D_{\Sigma_t}$. As a rule, the
eigenfunction problems of asymmetric (aside from unitary)
operators are ill-posed, and eigenfunctions do not comprise a
complete system, but there are rare exceptions. At this point we
have to admit the non-exhaustiveness of our treatment, as we do
not analyze this possibility. We will loosely rule out the
possibility of such operators to have a well-posed eigenfunction
problem.

A small remark will be useful later in the section.
\begin{remark}\label{SepSysSpanRemark}
If $\{T_\alpha X_\alpha\}_{\alpha\in\mathfrak{M}}$ is a separating
system for $Sol_0(\mathcal{T})$ with compact topology, then for
each $t\in\mathcal{I}$, the family $\{T_\alpha(t)
X_\alpha\}_{\alpha\in\mathfrak{M}}$ is a redundant complete system
for $\mathcal{D}(\mathcal{T}_t)$. In particular, for each $\vec
x\in\Sigma_t$, the family $\{X_\alpha(\vec
x)\}_{\alpha\in\mathfrak{M}}$ contains a (possibly redundant)
basis of $V$.
\end{remark}
The assertions are relatively obvious in the view of the fact,
that the restriction maps
$i^*_t,i^*_t\circ\nabla_t:Sol_0(\mathcal{T})\to\mathcal{D}(\mathcal{T}_t)$
are surjective, and hence a spanning system for
$Sol_0(\mathcal{T})$ must give a redundant complete system for the
Cauchy data
$\mathcal{D}(\mathcal{T}_t)\oplus\mathcal{D}(\mathcal{T}_t)$ on
$\Sigma_t$.

\begin{remark}\label{2LinIndSolUnique}
Let two equations $\ddot T(t)+F(t)\dot T(t)+G(t)T(t)=0$ and $\ddot
T(t)+H(t)\dot T(t)+J(t)T(t)=0$ have two common linearly
independent solutions $T(t)$ and $S(t)$. Then by Liouville formula
the Wronski determinant $\det W[T,S](t)$ evolves by
$$
\det W[T,S](t)=\det W[T,S](0)e^{-\int_0^td\tau F(\tau)}=\det
W[T,S](0)e^{-\int_0^td\tau H(\tau)},
$$
hence $F=H$ and thereby also $G=J$.
\end{remark}

\begin{proposition}\label{PropVarSep} The solution space
$Sol_0(\mathcal{T})$ admits a separating system if and only if
there exists a covering of $\mathcal{T}$ by local trivializations
such that the following local conditions are satisfied everywhere
(metric $g$ is time-separated):
\begin{romanlist}
\item $g_{00}=g_{00}(t)$, the metric component $g_{00}$ depends
only on time

\item the expression $\sum_{i,j=1}^3g^{ij}(x)\frac{\partial
g_{ij}}{\partial t}(x)$ is a function of time only

\item the connection 1-form $\Gamma$ and Christoffel symbols
${\bf\Gamma}^k_{ij}$ satisfy
$$
\sum_{i=1}^3g^{ij}[\Gamma_0,\Gamma_i]=0\mbox{, }\forall j>0,
$$
$$
\sum_{i,j=1}^3g^{ij}\left[\Gamma_0,\frac{\partial\Gamma_j}{\partial
x^i}+\Gamma_i\Gamma_j-\sum_{k=0}^3{\bf\Gamma}^k_{ij}\Gamma_k\right]=0,
$$
$$
\Gamma_0=\Gamma_0(t)\mbox{ is a function of time only}
$$

\item the eigenfunction problem of $D_{\Sigma_t}$ on different
$\Sigma_t$ can be adjusted, so that all $\tilde\Sigma_t$ are
isomorphic and the eigenfunctions $\zeta^t_\alpha=\zeta_\alpha$
are time-independent.
\end{romanlist}
\end{proposition}

\begin{proof} Throughout the section we will work exclusively locally,
i.e., in a local trivialization
$\pi^{-1}(U)\xrightarrow[]{\Psi}U\times V$, $U\subset M$. Thus we
identify the sections in a bundle having a typical fiber
$\mathfrak{F}$ with functions in $C^\infty(U;\mathfrak{F})$. We
will not keep the flag $U$ in this section but will always
understand objects as restricted to $U$.

The d'Alambert operator $\Box^\nabla$ on
$\mathcal{E}(\mathcal{T})$ has the following local expression in
terms of the connection form coefficients $\Gamma_i$ and
Christoffel symbols ${\bf\Gamma}^k_{ij}$,
\begin{eqnarray}
\Box^\nabla=\sum_{i,j=0}^3g^{ij}\left[\frac{\partial^2}{\partial
x^i\partial x^j}+2\Gamma_i\frac{\partial}{\partial
x^j}-\sum_{k=0}^3{\bf\Gamma}^k_{ij}\frac{\partial}{\partial
x^k}+\frac{\partial\Gamma_i}{\partial
x^j}+\Gamma_i\Gamma_j-\sum_{k=0}^3{\bf\Gamma}^k_{ij}\Gamma_k\right],\label{dAlambLocal}
\end{eqnarray}
and the field operator $D$ locally looks like
$$
D=\sum_{i,j=0}^3g^{ij}\frac{\partial^2}{\partial x^i\partial
x^j}+\sum_{i=0}^3A^i\frac{\partial}{\partial x^i}+B+m^\star,
$$
where $A^i,B\in C^\infty(U,End(V))$. To achieve a time separation
we need to choose a coordinate atlas such that everywhere
$g^{0i}=0$ for $i>0$. Then the operator $D$ locally breaks apart
into two differential operators, $D=D_t+D_{\Sigma_t}$, where
$$
D_t=g^{00}\frac{\partial^2}{\partial
t^2}+A^0\frac{\partial}{\partial t}+B^0,
$$
and
$$
D_{\Sigma_t}=\sum_{i,j=1}^3g^{ij}\frac{\partial}{\partial
x^i}\frac{\partial}{\partial
x^j}+\sum_{i=1}^3A^i\frac{\partial}{\partial
x^i}+B^3+m^\star=-\Delta_t+m^\star
$$
is the restricted field operator defined earlier. $B^0,B^3\in
C^\infty(U,End(V))$ are to be seen explicitly from
(Eq.\ref{dAlambLocal}).

$\Rightarrow${\it Necessity:} Let $\{T_\alpha X_\alpha\}$ be the
separating system system. Then
\begin{eqnarray}
DT_\alpha(t)X_\alpha(\vec
x)=(D_t+D_{\Sigma_t})T_\alpha(t)X_\alpha(\vec x)=\ddot
T_\alpha(t)g^{00}(x)X_\alpha(\vec x)+\nonumber\\
+\dot T_\alpha(t)A^0(x)X_\alpha(\vec
x)+T_\alpha(t)\left[B^0(x)+D_{\Sigma_t}\right]X_\alpha(\vec
x)=0.\label{PreModeEq}
\end{eqnarray}
That the metric signature is definite it follows that $g^{00}(x)$
never vanishes. We find a family of second order linear
homogeneous differential equations
$$
\ddot T_\alpha(t)g^{00}(x)X_\alpha^i(\vec x)+\dot
T_\alpha(t)\left(A^0(x)X_\alpha(\vec
x)\right)^i+T_\alpha(t)\left(\left[B^0(x)+D_{\Sigma_t}\right]X_\alpha(\vec
x)\right)^i=0
$$
parameterized by the spatial coordinates $\vec x\in\Sigma$ and
fiber indices $i=1,...,n$. By definition we similarly have $D\bar
T_\alpha(t)X_\alpha(\vec x)=0$. This means that all these
equations share at least two linearly independent solutions
$T_\alpha$ and $\bar T_\alpha$. If for some $\vec x$ and $i$,
$X_\alpha^i(\vec x)=0$, then the existence of two linearly
independent solutions for the resulting first order equation means
that
$$
\left(A^0(x)X_\alpha(\vec
x)\right)^i=\left(\left[B^0(x)+D_{\Sigma_t}\right]X_\alpha(\vec
x)\right)^i=0.
$$
Otherwise, by Remark \ref{2LinIndSolUnique} we find that there
exist functions $F_\alpha,G_\alpha\in C^\infty(\mathcal{I})$ such
that
$$
\left(A^0(x)X_\alpha(\vec
x)\right)^i=g^{00}(x)F_\alpha(t)X_\alpha^i(\vec x)\mbox{,
}\left(\left[B^0(x)+D_{\Sigma_t}\right]X_\alpha(\vec
x)\right)^i=g^{00}(x)G_\alpha(t)X_\alpha^i(\vec x).
$$
In both cases we establish that
\begin{eqnarray}
g_{00}(x)A^0(x)X_\alpha(\vec x)=F_\alpha(t)X_\alpha(\vec
x)\label{FDef}
\end{eqnarray}
and
\begin{eqnarray}
g_{00}(x)\left[B^0(x)+D_{\Sigma_t}\right]X_\alpha(\vec
x)=G_\alpha(t)X_\alpha(\vec x).\label{GDef}
\end{eqnarray}
 Thus for each $t\in\mathcal{I}$,
$X_\alpha$-s must be nothing else but the joint eigenfunctions of
the operators $g_{00}(x)A^0(x)$ and
$g_{00}(x)\left[B^0(x)+D_{\Sigma_t}\right]$ corresponding to
eigenvalues $F_\alpha(t)$ and $G_\alpha(t)$, respectively. The
operator $g_{00}(x)A^0(x)$ is simply a matrix, and at each point
$x\in M$ has at most $n$ independent eigenvectors. By Remark
\ref{SepSysSpanRemark}, $X_\alpha(\vec x)$-s span $V$, and thereby
$\{X_\alpha\}_{\alpha\in\mathfrak{M}}$ contains bases of all
eigenspaces of $g_{00}(x)A^0(x)$. From (Eq.\ref{dAlambLocal}) we
find
\begin{eqnarray}
g_{00}A^0=2\Gamma_0-g_{00}\sum_{i,j=0}^3g^{ij}{\bf\Gamma}^0_{ij},\label{A0Simpl}
\end{eqnarray}
and
$$
g_{00}B^0=\frac{\partial}{\partial
t}\Gamma_0+\Gamma^2_0-\sum_{k=1}^3{\bf\Gamma}^k_{00}\Gamma_k-g_{00}\sum_{i,j=0}^3g^{ij}{\bf\Gamma}^0_{ij}\Gamma_0.
$$
Now turn to the eigenfunction problem (Eq.\ref{GDef}). As
discussed above, for this problem to be well-posed it is necessary
that the differential operator
$g_{00}(x)\left[B^0(x)+D_{\Sigma_t}\right]$ is at least formally
self-adjoint. But this is possible only if $g_{00}(x)=g_{00}(t)$,
thus we have obtained the condition (i). Let us switch to an
atlas, where the time function $t$ is redefined such that
$g_{00}(t)=1$ (this step is not crucial, but only for
convenience). It follows, that
$$
{\bf\Gamma}^k_{00}=0\mbox{, }\forall k>0,
$$
so we obtain
\begin{eqnarray}
A^0=2\Gamma_0-\sum_{i,j=1}^3g^{ij}{\bf\Gamma}^0_{ij},\label{A0Def}\\
B^0=\frac{\partial}{\partial
t}\Gamma_0+\Gamma^2_0-\sum_{i,j=1}^3g^{ij}{\bf\Gamma}^0_{ij}\Gamma_0.\label{B0Def}
\end{eqnarray}
Combining (Eq.\ref{FDef}) and (Eq.\ref{A0Def}) we see that
$\{X_\alpha\}$-s are the eigenvectors of $\Gamma_0$, and these
eigenvectors are independent of $t$. Hence they are also the
eigenvectors of $\frac{\partial}{\partial t}\Gamma_0$, and thus by
(Eq.\ref{B0Def}) $A^0$ and $B^0$ are simultaneously
triangularizable,
$$
B^0X_\alpha(\vec x)=H_\alpha(x)X_\alpha(\vec x),
$$
for some $H_\alpha\in C^\infty(M)$. We note that
$$
{\bf\Gamma}^0_{ij}=-\frac{1}{2}\frac{\partial g_{ij}}{\partial t},
$$
and denote
$$
P(x)=-\sum_{i,j=1}^3g^{ij}(x){\bf\Gamma}^0_{ij}(x)=\frac{1}{2}\sum_{i,j=1}^3g^{ij}(x)\frac{\partial
g_{ij}}{\partial t}(x).
$$
Now (Eq.\ref{FDef}) and (Eq.\ref{GDef}) tell us, that for each
$t\in\mathcal{I}$ the operators $A^0$ and $D_{\Sigma_t}+B^0$ have
a common system of eigenfunctions spanning
$\mathcal{D}(\mathcal{T}_t)$, and therefore must commute,
$$
\left[A^0,D_{\Sigma_t}+B^0\right]u=\left[A^0,D_{\Sigma_t}\right]u=0\mbox{,
}\forall u\in\mathcal{D}(\mathcal{T}_t).
$$
This requires
$$
A^0(x)=2\Gamma_0(x)+P(x)=A^0(t),
$$
and
$$
\sum_{i=1}^3g^{ij}[\Gamma_0,\Gamma_i]=0\mbox{, }\forall j>0,
$$
$$
\sum_{i,j=1}^3g^{ij}\left[\Gamma_0,\frac{\partial\Gamma_j}{\partial
x^i}+\Gamma_i\Gamma_j-\sum_{k=0}^3{\bf\Gamma}^k_{ij}\Gamma_k\right]=0,
$$
exactly as the statement. Similarly, that operators $B^0$ and
$B^0+D_{\Sigma_t}$ have the same eigenfunctions implies, that
$[B^0,D_{\Sigma_t}]=0$, which on its turn requires
$B^0(x)=B^0(t)$, and thereby $P(x)=P(t)$ and
$\Gamma_0(x)=\Gamma_0(t)$. Thus we have proven parts (ii) and
(iii) of the statement. It follows further, that
$H_\alpha(x)=H_\alpha(t)$, and thus the eigenfunction problem
(Eq.\ref{GDef}) becomes
$$
D_{\Sigma_t} X_\alpha(\vec
x)=(G_\alpha(t)-H_\alpha(t))X_\alpha(\vec x).
$$
This is exactly the eigenfunction problem of $D_{\Sigma_t}$,
whence we conclude, that necessarily
$$
\{X_\alpha\}_{\alpha\in\mathfrak{M}}\subset\{\zeta^t_\lambda\}_{\lambda\in\mathbb{R}}.
$$
Therefore
$$
G_\alpha(t)=H_\alpha(t)+\lambda_\alpha(t),
$$
where
$$
\lambda_\alpha(t)=\{\lambda\in\mathbb{R}\mbox{:
}X_\alpha\in\mathcal{D}(\mathcal{T}_t)_\lambda'\}.
$$
Now (Eq.\ref{PreModeEq}) becomes
\begin{eqnarray}
\ddot T_\alpha(t)+F_\alpha(t)\dot
T_\alpha(t)+G_\alpha(t)T_\alpha(t)=0,\label{ModEq}
\end{eqnarray}
which is the mode equation for the mode $T_\alpha$. We have two
spanning systems for $\mathcal{D}(\mathcal{T}_t)$:
$\{X_\alpha\}_{\alpha\in\mathfrak{M}}$ and
$\{\zeta^t_\alpha\}_{\alpha\in\tilde\Sigma_t}$, and hence in each
eigenspace $\mathcal{D}(\mathcal{T}_t)_\lambda'$ we can choose a
basis from $\{X_\alpha\}_{\alpha\in\mathfrak{M}}$. Thus a complete
eigenfunction system can be chosen among
$\{X_\alpha\}_{\alpha\in\mathfrak{M}}$, proving the (iv) statement
of the proposition. We are complete with the necessity.

$\Leftarrow${\it Sufficiency: } Suppose all the points of the
statement are satisfied. Then, as we have seen above, by (iii)
$A^0$ and $B^0$ are functions of $t$ having the same eigenvectors,
and moreover, commute with $\Delta_t$. It follows that the actions
of $A^0$ and $B^0$ preserve $\mathcal{D}(\mathcal{T}_t)_\lambda'$,
and thus by a Gramm-Schmidt operation the representatives
$\zeta_\alpha$ can be chosen such that they are eigenfunctions of
$A^0$ and $B^0$. Thus each $\tilde\Sigma_\lambda$, and thereby the
entire $\tilde\Sigma$, decomposes into $n$ components
corresponding to the eigendirections of $A^0$,
$$
\tilde\Sigma=\bigcup_{i=1}^n\tilde\Sigma^i.
$$
For spatially homogeneous spacetimes discussed in later sections
we will give a more conceptual justification of such a subdivision
in terms of the representation theory.

Let for each $\alpha\in\tilde\Sigma$ choose a mode solution
$T_\alpha$ of (Eq.\ref{ModEq}) arbitrarily (strictly speaking, not
completely arbitrarily, but such that $T_\alpha$ and $\bar
T_\alpha$ are linearly independent) and consider the union of two
systems
$$
\{u,v\}_{\alpha\in\tilde\Sigma}\doteq\{u_\alpha\}_{\alpha\in\tilde\Sigma}\cup\{v_\alpha\}_{\alpha\in\tilde\Sigma}\mbox{,
}u_\alpha=T_\alpha\zeta_\alpha\mbox{, }v_\alpha=\bar
T_\alpha\zeta_\alpha.
$$
Choose any $\phi\in Sol_0(\mathcal{T})$. Then for each
$t\in\mathcal{I}$ the restriction $i_t^*(\phi)[\vec x]=\phi(t,\vec
x)\in\mathcal{D}(\mathcal{T}_t)$ can be Fourier expanded as
\begin{eqnarray}
\phi(t,\vec
x)=\int_{\tilde\Sigma}d\mu(\alpha)\hat\phi(\alpha;t)\zeta_\alpha(\vec
x)\label{TimeDepFourier}
\end{eqnarray}
with the integral converging in $L^2(\tilde\Sigma,\mu)$. Hence we
can differentiate under the integral,
$$
D\phi(t,\vec
x)=\int_{\tilde\Sigma}d\mu(\alpha)D\left[\hat\phi(\alpha;t)\zeta_\alpha(\vec
x)\right]=\int_{\tilde\Sigma}d\mu(\alpha)\left[\ddot{\hat{\phi}}(\alpha;t)+\right.
$$
$$
\left.+F_\alpha(t)\dot{\hat{\phi}}(\alpha;t)+G_\alpha(t)\hat\phi(\alpha;t)\right]\zeta_\alpha(\vec
x)=0,
$$
where for convenience we again reparameterized $t$ to get
$g_{00}=1$. Thus $\hat\phi(\alpha;t)$ is a solution of the mode
equation. All solutions of the ordinary second order equation
(Eq.\ref{ModEq}) are smooth and comprise a two complex dimensional
space,
$$
\hat\phi(\alpha;t)=a^\phi_\alpha T_\alpha(t)+b^\phi_\alpha\bar
T_\alpha(t)\mbox{, }a^\phi_\alpha,b^\phi_\alpha\in\mathbb{C}.
$$
Inserting this into (Eq.\ref{TimeDepFourier}) we finally arrive at
$$
\phi(t,\vec x)=\int_{\tilde\Sigma}d\mu(\alpha)\left[a^\phi_\alpha
T_\alpha(t)\zeta_\alpha(\vec x)+b^\phi_\alpha\bar
T_\alpha(t)\zeta_\alpha(\vec x)\right],
$$
which exactly means, that $\{u,v\}_{\alpha\in\tilde\Sigma}$ is a
separating system for $Sol_0(\mathcal{T})$. (For compatibility
with the definition one can concatenate $u_\alpha$ and $v_\alpha$
to a single function on the disjoint union
$\tilde\Sigma\sqcup\tilde\Sigma$.)
\end{proof}

The assertion of this proposition can be interpreted as follows.
If a mode decomposition in a reasonable fashion exists for
$Sol_0(\mathcal{T})$ then it is basically the mode decomposition
given by the time dependent Fourier transform which we will define
a few paragraphs later.\index{Mode decomposition}

As a supplement to the proposition we make a few remarks. Let
$\hat g_{ij}=g(\partial_i,\partial_j)$ and $\hat
h_{ij}=h(\partial_i,\partial_j)$ be the matrices of the metrics
$g$ and $h$, correspondingly, in a local chart, and $\sigma_k(\hat
h)$ the eigenvalues of the symmetric matrix $\hat h$.
\begin{remark}\label{RemarkCondii}
The condition (ii) of Proposition \ref{PropVarSep} is equivalent
to
$$
\det\hat g(x)=-g_{00}(t)\det\hat h(x)=-g_{00}(t)\sigma_1(\hat
h)\sigma_2(\hat h)\sigma_3(\hat
h)=-g_{00}(t)e^{2\int_0^tdt'P(t')}\det\hat h_0(\vec x),
$$
where $\det\hat h_0(\vec x)\in C^\infty(\Sigma)$ is a positive
smooth function (the notation will become clear later).
\end{remark}
The assertion follows from the combination of condition (ii) with
the Laplace's formula,
$$
\frac{\partial}{\partial t}\det\hat g=\det\hat g\cdot Tr[\hat
g^{-1}\frac{\partial}{\partial t}\hat g].
$$

Because $\nabla$ is a metric connection, the restrictions of the
previous proposition imply restrictions on the fiber metric
$\langle,\rangle_\mathfrak{g}$.  In case of a tensor bundle of
rank $(m,n)$ with Levi-Civita connection, coefficients $\Gamma_i$
are expressed in Christoffel symbols and the fiber metric is
induced by the spacetime metric, thus the restrictions fall onto
the spacetime $(M,g)$.

\begin{corollary}\label{CorMetricSep} Let a local moving frame be chosen,
such that the metric $\langle,\rangle_\mathfrak{g}$ is represented
by the matrix $\hat{\mathfrak{g}}$. Conditions (iii) of
Proposition \ref{PropVarSep} imply the following restrictions on
$\hat{\mathfrak{g}}$:
$$
\hat{\mathfrak{g}}(x)=\hat{\mathfrak{B}}^T(t)\hat{\mathfrak{g}}^0(\vec
x)\hat{\mathfrak{B}}(t),
$$
where $\hat{\mathfrak{g}}^0$ and $\hat{\mathfrak{B}}$ are matrix
valued smooth functions. In particular, for the a tensor bundle of
rank $(m,n)$ to allow for seprataion it is necessary that the
spacetime metric be represented by a matrix
$$
\hat g=1\oplus\left(-\hat h_0(\vec x)\hat B(t)\right),
$$
where $\hat h_0$ and $\hat B$ are matrix valued smooth functions.
\end{corollary}

\begin{proof} Locally the conservation of the metric
$\nabla\langle,\rangle_\mathfrak{g}=0$ can be written as
$$
\frac{\partial}{\partial
x^i}\hat{\mathfrak{g}}-\Gamma^T_i\hat{\mathfrak{g}}-\hat{\mathfrak{g}}\Gamma_i=0,
$$
where $\Gamma_i$ are the matrices of the connection coefficients
in the chosen frame. In particular, for $i=0$ we have
$$
\frac{\partial}{\partial
t}\hat{\mathfrak{g}}(x)-\Gamma^T_0(t)\hat{\mathfrak{g}}(x)-\hat{\mathfrak{g}}(x)\Gamma_0(t)=0,
$$
where $\Gamma_0=\Gamma_0(t)$ was used. The solutions of this
equation are of the form
$$
\hat{\mathfrak{g}}(x)=\hat{\mathfrak{B}}^T(t)\hat{\mathfrak{g}}^0(\vec
x)\hat{\mathfrak{B}}(t),
$$
where
\begin{eqnarray}
\hat{\mathfrak{B}}(t)=e^{\int_0^tdt'\Gamma_0(t')},\label{BDef}
\end{eqnarray}
and $\hat{\mathfrak{g}}^0(\vec x)$ is a smooth symmetric matrix
field on $\Sigma$.

Now if we identify the tensor space $(T_pM)^m_n$ with an $4^{n+m}$
dimensional vector space $V$ using a suitable bases, then each
matrix $\Gamma_i$ will be a $4^{n+m-1}\times 4^{n+m-1}$ matrix of
blocks, with blocks being the Christoffel symbols ${\bf\Gamma}_i$
for contravariant indices and $-{\bf\Gamma}^T_i$ for covariant
indices. $\Gamma_0=\Gamma_0(t)$ means
${\bf\Gamma}_0={\bf\Gamma}_0(t)$. With our time-separated metric
we have
$$
\hat g=1\oplus-\hat h.
$$
One can find
$$
{\bf\Gamma}_0=0\oplus\left(\frac{1}{2}\hat
h^{-1}\frac{\partial\hat h}{\partial
t}\right)={\bf\Gamma}_0(t)=0\oplus\hat A(t)
$$
for some smooth $3\times3$ matrix $\hat A(t)$. The solution is
$$
\hat h(x)=\hat h_0(\vec x)e^{2\int_0^tdt'\hat A(t')}=\hat h_0(\vec
x)\hat B(t),
$$
for smooth symmetric commuting matrix fields $\hat h_0(\vec x)$
and $\hat B(t)$.
\end{proof}

Now the notation $\det\hat h_0$ of Remark \ref{RemarkCondii}
becomes clear, and we see that
$$
\det\hat B(t)=e^{2\int_0^tdt'P(t')}
$$
for a tensor bundle. Note that for the scalar field conditions
(iii) are trivially satisfied and do not restrict the spacetime.

\begin{remark}\label{dmu_h}
For the volume form measure $d\mu_h$ on $\Sigma_t$ we have locally
$$
d\mu_h(\vec x)=\sqrt{\det\hat{h}(t,\vec x)}dx^1dx^2dx^3.
$$
By Remark \ref{RemarkCondii} we have
$$
\det\hat h(t,\vec x)=e^{2\int_0^tdt'P(t')}\det\hat h_0(\vec x),
$$
hence
$$
d\mu_h(\vec x)=e^{\int_0^tdt'P(t')}d\mu_{h_0}(\vec x),
$$
where
$$
d\mu_{h_0}(\vec x)=\sqrt{\det\hat{h_0}(\vec x)}dx^1dx^2dx^3.
$$
\end{remark}

Henceforth by stating that a mode decomposition of
$Sol_0(\mathcal{T})$ exists we will mean that the assumptions of
the Proposition \ref{PropVarSep} are satisfied and the
corresponding covering is chosen. We are ready to formulate
precisely the time dependent Fourier transform. Note that although
$\zeta_\alpha$ are $t$-independent, the spatial metric $h$ and the
fiber metric $\langle,\rangle_\mathfrak{g}$ depend on $t$, and
$\zeta_\alpha$ are not orthonormal with respect to the measure
$d\mu_h$ for all $t$ simultaneously. At this point we appoint once
and forever to normalize $\zeta_\alpha$ such that they are
orthonormal at $t=0$. Or equivalently, they are orthonormal with
respect to the measure $d\mu_{h_0}$ of Remark \ref{dmu_h} and the
fiber metric $\mathfrak{g}^0$ of Corollary \ref{CorMetricSep}.

\begin{definition}
For $f\in\mathcal{D}(\mathcal{T})$ we define the time dependent
Fourier transform $\tilde
f(t,\alpha)=\mathcal{F}[f(t,.)](\alpha)\in
C_0^\infty\left(\mathcal{I},\tilde{\mathcal{D}}(\tilde\Sigma)\right)$
by
$$
\mathcal{F}[f(t,.)](\alpha)=\int_{\Sigma_t}d\mu_{h_0}(\vec
x)\langle\bar\zeta_\alpha(\vec x),f(t,\vec
x)\rangle_{\mathfrak{g}^0}.
$$
\end{definition}

Here we note another important corollary, which will be useful
later. It will give the time dependent Plancherel formula.

\begin{corollary}\label{TDPlancherel}
Suppose the assumptions of Proposition \ref{PropVarSep} are
satisfied, and the corresponding covering is chosen. Then for all
$f\in\mathcal{D}(\mathcal{T})$
$$
(\zeta_\alpha,f(t,.))_{\Sigma_t}=I_\alpha(t)\mathcal{F}[f(t,.)](\alpha)
$$
and the time dependent Plancherel formula for the time-dependent
Fourier transform is given by
$$
(f(t,.),h(t,.))_{\Sigma_t}=\int_{\tilde\Sigma}d\mu(\alpha)I_\alpha(t)\overline{\mathcal{F}[f(t,.)]}(\alpha)\mathcal{F}[h(t,.)](\alpha),
$$
where
$$
I_\alpha(t)=e^{\int_0^tdt'F_\alpha(t')}.
$$
\end{corollary}

\begin{proof}
By Remark \ref{dmu_h}
$$
(f(t,.),h(t,.))_{\Sigma_t}=\int_{\Sigma_t}d\mu_h(\vec
x)(f(t,.),h(t,.))_\mathfrak{g}=
$$
$$
=e^{\int_0^tdt'P(t')}\int_{\Sigma_t}d\mu_{h_0}(\vec
x)(f(t,.),h(t,.))_\mathfrak{g}.
$$
At the same time by Corollary \ref{CorMetricSep} we have
\begin{eqnarray}
(f(t,.),h(t,.))_\mathfrak{g}=(\hat{\mathfrak{B}}(t)f(t,.),\hat{\mathfrak{B}}(t)h(t,.))_{\mathfrak{g}^0}.\label{FiberProdt}
\end{eqnarray}
Because we have normalized $\zeta_\alpha$ with respect to
$d\mu_{h_0}$ and $\mathfrak{g}^0$, the conventional Plancherel
formula holds for them,
$$
\int_{\Sigma_t}d\mu_{h_0}(\vec
x)(f(t,.),h(t,.))_{\mathfrak{g}^0}=\int_{\tilde\Sigma}d\mu(\alpha)\overline{\mathcal{F}[f(t,.)](\alpha)}\mathcal{F}[h(t,.)](\alpha).
$$
Combining these three formulas we find
$$
(f(t,.),h(t,.))_{\Sigma_t}=e^{\int_0^tdt'P(t')}\int_{\tilde\Sigma}d\mu(\alpha)\overline{\mathcal{F}[\hat{\mathfrak{B}}(t)f(t,.)](\alpha)}\mathcal{F}[\hat{\mathfrak{B}}(t)f(t,.)](\alpha).
$$
Meanwhile
$$
(\zeta_\alpha,f(t,.))_{\Sigma_t}=e^{\int_0^tdt'P(t')}\int_{\Sigma_t}d\mu_{h_0}(\vec
x)(\hat{\mathfrak{B}}(t)\zeta_\alpha,\hat{\mathfrak{B}}(t)h(t,.))_{\mathfrak{g}^0}.
$$
By definition
$$
\hat{\mathfrak{B}}(t)\zeta_\alpha=e^{\int_0^tdt'\Gamma_0(t')}\zeta_\alpha=e^{\frac{1}{2}\int_0^tdt'[A^0(t')-P(t')]}\zeta_\alpha=e^{\frac{1}{2}\int_0^tdt'[F_\alpha(t')-P(t')]}\zeta_\alpha,
$$
whence
$$
(\zeta_\alpha,f(t,.))_{\Sigma_t}=e^{\frac{1}{2}\int_0^tdt'[F_\alpha(t')+P(t')]}\mathcal{F}[\hat{\mathfrak{B}}(t)f(t,.)](\alpha).
$$
Finally
$$
\hat{\mathfrak{B}}(t)f(t,\vec
x)=\int_{\tilde\Sigma}d\mu(\alpha)\mathcal{F}[f(t,.)](\alpha)\hat{\mathfrak{B}}(t)\zeta_\alpha(\vec
x)=\int_{\tilde\Sigma}d\mu(\alpha)\mathcal{F}[f(t,.)](\alpha)e^{\frac{1}{2}\int_0^tdt'[F_\alpha(t')-P(t')]}\zeta_\alpha,
$$
thus
$$
\mathcal{F}[\hat{\mathfrak{B}}(t)f(t,.)](\alpha)=e^{\frac{1}{2}\int_0^tdt'[F_\alpha(t')-P(t')]}\mathcal{F}[f(t,.)](\alpha).
$$
The assertions now easily follow.
\end{proof}

At last we compute the spectra of operators $A^0$ and $B^0$ for
the tensor bundle to find the functions $F_\alpha$ and $H_\alpha$.
In view of Corollary \ref{CorMetricSep} the function $P(t)$
becomes
$$
P(t)=\frac{1}{2}Tr\left[\hat B^{-1}(t)\frac{\partial\hat
B}{\partial t}(t)\right].
$$
Then
$$
Spec\{A^0\}=2Spec\{\Gamma_0\}+P(t),
$$
$$
Spec\{B^0\}=\left\{\dot \sigma(t)+\sigma^2(t)+\sigma(t)P(t)\mbox{:
}\sigma\in Spec\{\Gamma_0\}\right\}.
$$
As a useful example we calculate these spectra for the scalar and
1-form fields on uniformly expanding (e.g., FRW) manifolds,
$$
ds^2=dt^2-a^2(t)d\sigma^2(\vec x).
$$
Here the matrix $\hat B(t)=a^2(t)1$, and hence
$$
\hat A(t)=\frac{\partial}{\partial t}\ln a(t)1=H(t)1,
$$
and
$$
P(t)=3H(t)\mbox{, }H(t)=\frac{\dot a(t)}{a(t)}.
$$
For scalar case $n=m=0$ and we have
$$
Spec\{\Gamma_0\}=\{0\},
$$
thus
$$
Spec\{A^0\}=\{3H(t)\}\mbox{, }Spec\{B^0\}=\{0\},
$$
as well known. For the 1-form case, $m=0$, $n=1$, we have
$$
Spec\{\Gamma_0\}=\{0, -H(t)\},
$$
and thereby
$$
Spec\{A^0\}=\{3H(t),H(t)\}\mbox{, }Spec\{B^0\}=\{0,-\dot
H(t)-2H^2(t)\},
$$
where the first members are similar to the scalar case and
represent the scalar modes, but second ones represent the
transversal and longitudinal modes.

As we have seen, for the separation it is necessary that the
evolution of the metric be represented by linear transformations.
If the connection also satisfies such a condition in a suitable
sense, than the operator $D_{\Sigma_t}$ is essentially the same at
every $t$ up to some scale factors. (Maybe the condition (iii) of
the main proposition already implies such a restriction on the
connection, but we are not sure yet.) This will be the case for
all our bundles of interest, and it will provide analytical
advantages. To summarize what we expect precisely we give the
following definitions.

\begin{definition}
We will say that the operator $D_{\Sigma_t}$ has a {\bf strictly
uniform} spectrum over time if there exists a lower semi-bounded
function $\omega(\alpha)$ on $\tilde\Sigma$, a positive smooth
function $C(t)>0$ and a smooth function $\tilde m^\star(t)$ such
that $\lambda_\alpha(t)=\omega(\alpha)C(t)+\tilde m^\star(t)$, or
equivalently, the expression
$$
\frac{d}{dt}\ln|\lambda_\alpha(t)-\tilde m^\star(t)|
$$
does not depend on $\alpha$.
\end{definition}\index{Strictly uniform spectrum}

This is a rather strong condition. It basically requires that the
eigenspaces of $D_{\Sigma_t}$ coincide for different $t$ up to an
overall shift, and that eigenvalues be linearly proportional. Such
a property would be very comfortable, but it does not hold for
some models of our interest. In particular, it does not hold for
the Bianchi I model with distortions. Hence we will derive some of
our results under a milder restriction which holds at least in all
cosmological situations where the spectral theory is explicit
enough so that it can be checked. (The explicit spectral theory of
the scalar field on Bianchi I-VII spacetimes will appear in a
subsequent publication.)

\begin{definition}
We will say that the operator $D_{\Sigma_t}$ has a {\bf loosely
uniform} spectrum over time if
$$
\left|\frac{d}{dt}\ln|\lambda_\alpha(t)-\tilde
m^\star(t)|\right|\le C_\mathcal{R}\mbox{, }\forall
t\in\mathcal{R}, \alpha\in\tilde\Sigma,
$$
for any compact interval $\mathcal{R}\subset\mathcal{I}$, and for
some $0<C_\mathcal{R}\in\mathbb{R}$ and a smooth function $\tilde
m^\star(t)$.
\end{definition}\index{Loosely uniform spectrum}

If the Fourier transform is conventional, then it will be natural
to require that $\omega$ be an analytic function on
$\tilde\Sigma$.

\subsection{Some properties of the mode solutions}

In this section we investigate the equation (Eq.\ref{ModEq}) and
obtain some useful properties of the mode solutions $T_\alpha$.
The mode equation is
$$
\ddot T_\alpha(t)+F_\alpha(t)\dot
T_\alpha(t)+G_\alpha(t)T_\alpha(t)=0,
$$
where
$$
G_\alpha(t)=H_\alpha(t)+\lambda_\alpha(t),
$$
and $\lambda_\alpha(t)$ are the eigenvalues of the operator
$D_{\Sigma_t}=-\Delta_t+m^\star(x)$. Note that $G_\alpha$ may
become null or negative for some rates of expansion. This
corresponds to the so-called positive back-reaction in a linear
system and results in exponential solutions. This is an
interesting phenomenon appearing in non scalar fields (for scalar
fields $H_\alpha=0$), and its significance is not yet completely
clear to us. To understand it one could, for instance, track its
influence on the energy-momentum tensor etc. It is not obvious
that this is really a physical infrared instability, because it
may occur for the co-vector field but not for the vector
counterpart, for instance. It is also worth mentioning, that for
the co-vector (1-form) field the introduction of a conformal
coupling precisely cancels this instability. It seems plausible
that for each field there is a choice of the coupling constant
which compensates this bad infrared behavior. We say infrared,
because $\lambda_\alpha(t)$ attains arbitrarily large positive
values at any $t$, thus on an unbounded subbundle of
$R\times\tilde\Sigma$, $G_\alpha$ is positive.

We will make this more explicit under the assumption, that
$D_{\Sigma_t}$ has a strictly uniform spectrum. Then the function
$C(t)$ is uniformly bounded from below, and the functions
$H_\alpha$ and $\tilde m^\star$ are uniformly bounded from above
on any compact interval $\mathcal{R}$. On the other hand
$\omega\to+\infty$, hence is suffices to choose $\omega$ large
enough to make $G_\alpha=H_\alpha+C\omega+\tilde m^\star>0$.

Fix a component $\tilde\Sigma^i$ and write $H=H_\alpha$,
$F=F_\alpha$ and $I=I_\alpha$ for all $\alpha\in\tilde\Sigma^i$.
Define a new variable
$$
s(t)=\int_0^td\tau e^{-\int_0^\tau d\tau'F(\tau')}=\int_0^td\tau
I^{-1}(\tau),
$$
which is in a smooth monotone bijective correspondence with $t$.
The inverse function will be denoted by $t(s)$. Regarding all the
acting functions of $t$ as functions of $s$ we obtain
\begin{eqnarray}
\ddot T_\alpha(s)+\Lambda_\alpha(s)T_\alpha(s)=0,\label{ModEqs}
\end{eqnarray}
where
$$
\Lambda_\alpha(s)=\left[G_\alpha(t)e^{2\int_0^td\tau
F(\tau)}\right]_{t=t(s)}=G_\alpha(s)I^2(s).
$$
This is a time dependent harmonic oscillator equation, to which
the results in the appendix apply.

\begin{remark}
Note that the Wronski determinant of two solutions $Q,R$
$$
\det W[Q,R](s)=Q(s)\dot R(s)-\dot Q(s)R(s)=\mbox{const}
$$
in variable $t$ becomes
$$
\det W[Q,R](t)=\frac{dt}{ds}\left(Q(t)\dot R(t)-\dot
Q(t)R(t)\right)=I(t)\left(Q(t)\dot R(t)-\dot
Q(t)R(t)\right)=\mbox{const}.
$$
\end{remark}

Applying Corollary \ref{TEstAbstract} to (Eq.\ref{ModEqs}) for
different $\alpha$ we find estimates which in principle depend on
$\alpha$ in a complicated way. But under the assumption of loose
uniformity on $D_{\Sigma_t}$ we will be able to invoke more
comfortable expressions.

\begin{proposition}\label{TEstLooseUni}
Suppose $D_{\Sigma_t}$ has a loosely uniform spectrum over time.
Then for a family of arbitrary solutions $T_\alpha$ of
(Eq.\ref{ModEqs}) the following estimate holds
$$
|T_\alpha(s)|\le
R_\mathcal{R}|T_\alpha(0)|+\frac{S_\mathcal{R}}{\max\{1,\sqrt{U_\mathcal{R}+T_\mathcal{R}\lambda_\alpha(0)}\}}|\dot
T_\alpha(0)|\mbox{, }\forall s\in\mathcal{R},
$$
with $0<R_\mathcal{R},S_\mathcal{R},T_\mathcal{R}\in\mathbb{R}$
and $U_\mathcal{R}\in\mathbb{R}$, for any compact interval
$\mathcal{R}$.
\end{proposition}

\begin{proof}
Fix a compact interval $\mathcal{R}$ and for each
$\alpha\in\tilde\Sigma^i$ apply Corollary \ref{TEstAbstract} with
$\Lambda_\alpha(s)=I^2(s)H(s)+I^2(s)\lambda_\alpha(s)$. Because
$\Lambda_\alpha$ is real, we get $A_\mathcal{R}(\alpha)=0$. As
$\lambda_\alpha(s)$ is lower semi-bounded we have
$$
p_\mathcal{R}\doteq\inf_{\tilde\Sigma}\inf_\mathcal{R}\lambda_\alpha>-\infty.
$$
Denote $m_\mathcal{R}=\inf_\mathcal{R}\{I^2H\}$ and
$n_\mathcal{R}=\inf_\mathcal{R}\{I^2\}>0$. Then
$$
c_\mathcal{R}(\alpha)\ge
m_\mathcal{R}+n_\mathcal{R}\inf_\mathcal{R}\lambda_\alpha\ge
m_\mathcal{R}+n_\mathcal{R}p_\mathcal{R}.
$$
It follows that
$\kappa(\alpha)\le\sqrt{1+|m_\mathcal{R}+n_\mathcal{R}p_\mathcal{R}|}$
and
$e_\mathcal{R}(\alpha)\ge1+\max\{0,m_\mathcal{R}+n_\mathcal{R}\inf_\mathcal{R}\lambda_\alpha\}$.
Denote $M_\mathcal{R}=\sup_\mathcal{R}\{|I^2H|\}\ge0$ and
$N_\mathcal{R}=\sup_\mathcal{R}\{I^2\}>0$. We find next
$$
D_\mathcal{R}(\alpha)\le1+|m_\mathcal{R}+n_\mathcal{R}p_\mathcal{R}|+M_\mathcal{R}+N_\mathcal{R}\left|\sup_\mathcal{R}\lambda_\alpha\right|.
$$
Now we observe that by loose uniformity
$$
\left|\ln\frac{|\lambda_\alpha(s)-\tilde
m^\star(s)|}{|\lambda_\alpha(s')-\tilde
m^\star(s')|}\right|=\left|\int_s^{s'}d\sigma\partial_s\ln|\lambda_\alpha(\sigma)-\tilde
m^\star(\sigma)|\right|\le|\mathcal{R}|\sqrt{N_\mathcal{R}}C_\mathcal{R},
$$
hence
\begin{eqnarray}
\sup_\mathcal{R}|\lambda_\alpha-\tilde
m^\star|\le|\lambda_\alpha(0)-\tilde
m^\star(0)|e^{|\mathcal{R}|\sqrt{N_\mathcal{R}}C_\mathcal{R}}\nonumber\\
\inf_\mathcal{R}|\lambda_\alpha-\tilde
m^\star|\ge|\lambda_\alpha(0)-\tilde
m^\star(0)|e^{-|\mathcal{R}|\sqrt{N_\mathcal{R}}C_\mathcal{R}}.\label{lambda_msupinf}
\end{eqnarray}
Note that whenever $\lambda_\alpha(0)-\tilde m^\star(0)>0$ then it
follows by continuity that $\lambda_\alpha(s)-\tilde m^\star(s)>0$
for all $s\in\mathcal{R}$. Denote
$$
\lambda_{min}=\tilde m^\star(0)-\min\{0,\inf_\mathcal{R}\tilde
m^\star\cdot
e^{|\mathcal{R}|\sqrt{N_\mathcal{R}}C_\mathcal{R}}\}-\min\{0,\frac{m_\mathcal{R}}{n_\mathcal{R}}e^{|\mathcal{R}|\sqrt{N_\mathcal{R}}C_\mathcal{R}}\}.
$$
Then from $\lambda_\alpha(0)>\lambda_{min}$ it follows
$\lambda_\alpha(0)-\tilde m^\star(0)>0$,
$m_\mathcal{R}+n_\mathcal{R}\inf_\mathcal{R}\lambda_\alpha>0$ and
$$
\inf_\mathcal{R}\lambda_\alpha\ge\inf_\mathcal{R}\tilde
m^\star+(\lambda_\alpha(0)-\tilde
m^\star(0))e^{-|\mathcal{R}|\sqrt{N_\mathcal{R}}C_\mathcal{R}}>0.
$$
Now we have that
$$
e_\mathcal{R}(\alpha)\ge1+\chi[\lambda_\alpha(0)>\lambda_{min}]\left(m_\mathcal{R}+n_\mathcal{R}(\inf_\mathcal{R}\tilde
m^\star+(\lambda_\alpha(0)-\tilde
m^\star(0))e^{-|\mathcal{R}|\sqrt{N_\mathcal{R}}C_\mathcal{R}})\right),
$$
where the characteristic function $\chi$ plays here the role of
the condition checking. From (Eq.\ref{lambda_msupinf}) we find
$$
\sup_\mathcal{R}|\lambda_\alpha|\le\sup_\mathcal{R}\tilde
m^\star+|\lambda_\alpha(0)-\tilde
m^\star(0)|e^{|\mathcal{R}|\sqrt{N_\mathcal{R}}C_\mathcal{R}},
$$
whence
$$
D_\mathcal{R}(\alpha)\le1+|m_\mathcal{R}+n_\mathcal{R}p_\mathcal{R}|+M_\mathcal{R}+N_\mathcal{R}(\sup_\mathcal{R}\tilde
m^\star+|\lambda_\alpha(0)-\tilde
m^\star(0)|e^{|\mathcal{R}|\sqrt{N_\mathcal{R}}C_\mathcal{R}}).
$$
Thus we establish that for $\lambda_\alpha(0)\le\lambda_{min}$
$$
\frac{D_\mathcal{R}(\alpha)}{e_\mathcal{R}(\alpha)}\le1+|m_\mathcal{R}+n_\mathcal{R}p_\mathcal{R}|+M_\mathcal{R}+N_\mathcal{R}(\sup_\mathcal{R}\tilde
m^\star+(|\lambda_{min}|+|\tilde
m^\star(0)|)e^{|\mathcal{R}|\sqrt{N_\mathcal{R}}C_\mathcal{R}}),
$$
and for $\lambda_\alpha(0)>\lambda_{min}$
$$
\frac{D_\mathcal{R}(\alpha)}{e_\mathcal{R}(\alpha)}\le
e^{2|\mathcal{R}|\sqrt{N_\mathcal{R}}C_\mathcal{R}}N_\mathcal{R}\left(\frac{1}{n_\mathcal{R}}+\frac{1+|m_\mathcal{R}+n_\mathcal{R}p_\mathcal{R}|+M_\mathcal{R}+N_\mathcal{R}|\sup_\mathcal{R}\tilde
m^\star|}{N_\mathcal{R}e^{|\mathcal{R}|\sqrt{N_\mathcal{R}}C_\mathcal{R}}}+\right.
$$
$$
\left.+\frac{1+|m_\mathcal{R}|+n_\mathcal{R}|\inf_\mathcal{R}\tilde
m^\star|}{n_\mathcal{R}e^{-|\mathcal{R}|\sqrt{N_\mathcal{R}}C_\mathcal{R}}}\right).
$$
Finally
$$
\frac{d}{ds}\ln(\kappa^2+\Lambda_\alpha)=\frac{\frac{d}{ds}(I^2(H+\tilde
m^\star))+\frac{d}{ds}(I^2)(\lambda_\alpha-\tilde
m^\star)+I^2(\lambda_\alpha-\tilde
m^\star)\frac{d}{ds}\ln|\lambda_\alpha-\tilde
m^\star|}{\kappa^2+I^2H+I^2\lambda_\alpha}.
$$
Denote $P_\mathcal{R}=\sup_\mathcal{R}|\partial_s(I^2(H-\tilde
m^\star))|\ge0$ and
$Q_\mathcal{R}=\sup_\mathcal{R}|\partial_s(I^2)|\ge 0$. Again
using the loose uniformity, for $\lambda_\alpha-\tilde
m^\star\le1$
$$
\left|\frac{d}{ds}\ln(\kappa^2+\Lambda_\alpha)\right|\le
P_\mathcal{R}+Q_\mathcal{R}+(N_\mathcal{R})^\frac{3}{2}C_\mathcal{R},
$$
and else
$$
\left|\frac{d}{ds}\ln(\kappa^2+\Lambda_\alpha)\right|\le\frac{P_\mathcal{R}+Q_\mathcal{R}+(N_\mathcal{R})^\frac{3}{2}C_\mathcal{R}}{n_\mathcal{R}}.
$$
Summarizing this all we find that by Corollary \ref{TEstAbstract}
there exist numbers
$0<R_\mathcal{R},S_\mathcal{R},T_\mathcal{R}\in\mathbb{R}$ and
$U_\mathcal{R}\in\mathbb{R}$ such that for a family of arbitrary
solutions $T_\alpha$ we have
$$
|T_\alpha(s)|\le
R_\mathcal{R}|T_\alpha(0)|+\frac{S_\mathcal{R}}{\max\{1,\sqrt{U_\mathcal{R}+T_\mathcal{R}\lambda_\alpha(0)}\}}|\dot
T_\alpha(0)|,
$$
what was to be proven.
\end{proof}

The result can be strengthened under additional assumptions. These
are perhaps too restrictive, but they appear to be sufficient for
some important applications. Let $\mathbb{H}_a=\{z\in\mathbb{C}:
|\Im z|<a\}$.

\begin{proposition}\label{TEstHol}
Suppose the bundle $\mathcal{T}$ is analytic, so that all
functions figuring in (Eq.\ref{ModEqs}) are real analytic
functions of $s$. Suppose further that $D_{\Sigma_t}$ has a
strictly uniform spectrum. Choose the initial data to be
$T_\alpha(0)=p(\omega(\alpha))$ and $\dot
T_\alpha(0)=q(\omega(\alpha))$, where $p(\omega)$, $q(\omega)$ are
holomorphic functions on $\mathbb{H}_a$ for some $a>0$. Then for
each $s$, $T_\alpha(s)=r_s(\omega(\alpha))$, where $r_s(\omega)$
is holomorphic in $\omega$ on $\mathbb{H}_a$ and real analytic in
$s$, and for any compact interval $\mathcal{R}$ it holds
$$
|r_s(\omega)|\le
R_\mathcal{R}|p(\omega)|+\frac{S_\mathcal{R}}{\max\{1,\sqrt{U_\mathcal{R}+T_\mathcal{R}\Re\omega}\}}|q(\omega)|\mbox{,
}\forall s\in\mathcal{R},
$$
with $0<R_\mathcal{R},S_\mathcal{R},T_\mathcal{R}\in\mathbb{R}$
and $U_\mathcal{R}\in\mathbb{R}$.
\end{proposition}

\begin{proof}
By strict uniformity we have
$\lambda_\alpha(t)=\omega(\alpha)C(t)+\tilde m^\star(t)$, and if
the initial data depend only on $\omega$, then the solutions will
also be such. Therefore for convenience we write
\begin{eqnarray}
\ddot T_\omega(s)+I^2(s)(H(s)+\omega C(s)+\tilde
m^\star(s))T_\omega(s)=0\label{ModEqHol}
\end{eqnarray}
with $T_\omega(0)=p(\omega)$ and $\dot T_\omega(0)=q(\omega)$.
From the theory of power series it is clear that any real analytic
function on $s(\mathcal{I})$ can be extended to a holomorphic
function in some open neighborhood $\delta(s(\mathcal{I}))$ of
$s(\mathcal{I})$. Consider (Eq.\ref{ModEqHol}) as a complex
differential equation, then for any $\omega\in\mathbb{H}_a$, by
Satz 4.1 and Satz 4.2 of \cite{Herold1975} there exist
neighborhoods $\delta(0)$ of 0 and $\delta(\omega)$ of $\omega$
such that $T_\omega(s)$ is holomorphic in
$\delta(0)\times\delta(\omega)$. At the same time by Satz 5.3 of
\cite{Herold1975}, for any $\omega\in\mathbb{H}_a$ the solution
$T_\omega$ can be analytically continued to the whole of
$\delta(s(\mathcal{I}))$. Thus $T_\omega(s)$ is holomorphic in
$\delta(\mathbb{R})\times\mathbb{H}_a$. Now restrict back to the
real axis and fix the interval $\mathcal{R}$. The reasoning of the
previous proposition can be repeated literally except that now
$A_\mathcal{R}(\alpha)$ is not zero but equals
$A_\mathcal{R}(\omega)=|\Im\omega|\sup_\mathcal{R}\{I^2C\}<a\sup_\mathcal{R}\{I^2C\}$.
This results in a similar formula as in Proposition
\ref{TEstLooseUni} with perhaps different coefficients, and that
proves our assertion.
\end{proof}

We have an immediate corollary.

\begin{corollary}\label{TPolBound}
Under the assumptions of Proposition \ref{TEstHol}, if
$p,q\in\mathcal{A}(\mathbb{H}_a)$ then for each $s\in\mathcal{R}$,
$r_s\in\mathcal{A}(\mathbb{H}_a)$.
\end{corollary}

\subsection{Mode decomposition of weak solutions}

The aim of this section will be to extend the mode decomposition
of $Sol_0(\mathcal{T})$ obtained previously to entire
$\mathcal{D}(\mathcal{T})_0'$. Here we assume all the conditions
of Proposition \ref{PropVarSep} are satisfied, and we have chosen
the system $\{u,v\}_\alpha$ with $u_\alpha=T_\alpha\zeta_\alpha$
and $v_\alpha=\bar T_\alpha\zeta_\alpha$, which span
$Sol_0(\mathcal{T})$. For convenience we will also assume at least
the part (iv) of the definition of the conventional Fourier
transform to hold.

Unfortunately we do not have a precise analytical description of
the Fourier transformed test function space
$\tilde{\mathcal{D}}(\tilde\Sigma)$ even under the assumptions of
the conventional Fourier transform, as it was, for instance, in
the Euclidean space by Paley-Wiener theorem. In particular we need
to know for which modes $T_\alpha$ it holds
\begin{eqnarray}
T_\alpha(t)\tilde
f(\alpha)\in\tilde{\mathcal{D}}(\tilde\Sigma)\mbox{, }\forall
\tilde f(\alpha)\in\tilde{\mathcal{D}}(\tilde\Sigma)\mbox{,
}t\in\mathcal{I}.\label{TMultCond}
\end{eqnarray}
At least we are able to find a sufficient condition under
additional assumptions.

\begin{proposition}\label{TMultPropAnal}
Suppose the bundle $\mathcal{T}$ is analytic and $D_{\Sigma_t}$
has a strictly uniform spectrum. For each
$\alpha\in\tilde\Sigma^i$ set $T_\alpha(0)=p^i(\omega(\alpha))$
and $\dot T_\alpha(0)=q^i(\omega(\alpha))$, where
$p^i,q^i\in\mathcal{A}[\mathbb{H}_0]$. Then (Eq.\ref{TMultCond})
holds.
\end{proposition}

\begin{proof}
Choose the interval $\mathcal{R}$ such that it contains both 0 and
$t$. First we note that by Corollary \ref{TPolBound} for
$\alpha\in\tilde\Sigma^i$ we have
$T_\alpha(t)=r^i_t(\omega(\alpha))$ with
$r^i_t\in\mathcal{A}[\mathbb{H}_0]$. Denote
$F_t^i(\lambda)=r^i_t(\frac{\lambda}{C(t)})\in\mathcal{A}[\mathbb{H}_0]$.
Obviously for any $f\in\mathcal{D}(\mathcal{T}_t)$,
$$
F_t^i(\lambda_\alpha(t))\tilde
f(\alpha)=\widetilde{\left[F_t^i(D_{\Sigma_t})f\right]}(\alpha),
$$
where $F_t^i(D_{\Sigma_t})$ is defined by functional calculus.
Then by Proposition \ref{FuncCalcProp}
$$
F_t^i(\lambda_\alpha(t))\tilde
f(\alpha)\in\tilde{\mathcal{D}}(\tilde\Sigma).
$$

Let $\{U_n\}$ be a covering by local trivializations of
$\mathcal{T}_t$, and let $\{\imath_n\}$ be a subordinate partition
of unity. The support of $f$ is covered by $N_f$ (finite)
trivializing neighborhoods, and we write $f=\sum_n\imath_nf=\sum_n
f_n$. It follows $\tilde f=\sum_n\tilde f_n$ and
$T_\alpha(t)\tilde f(\alpha)=\sum_n T_\alpha(t)\tilde
f_n(\alpha)$. Consider $f_n$ as a section in the trivial bundle
$\pi^{-1}(U_n)$. As we have seen already (and as we will see even
more evidently for homogeneous spacetimes in the next chapter)
each component $\tilde\Sigma^i$ supports the Fourier transform of
one fiber component in some local frame. Thus we can write
$f_n=\sum_i f^i_n$, where $f^i_n\in\mathcal{D}(U_n)$ and $\tilde
f^i_n$ is supported in $\tilde\Sigma^i$. We get
$$
T_\alpha(t)\tilde f(\alpha)=\sum_n\sum_i
F_t^i(\lambda_\alpha(t))\tilde
f^i_n(\alpha)=\sum_n\sum_i\widetilde{\left[F_t^i(D_{\Sigma_t})f^i_n\right]}(\alpha)\in\tilde{\mathcal{D}}(\tilde\Sigma),
$$
which completes the proof.
\end{proof}

\begin{remark}\label{MultI_alpha}
An argument involving local trivializations as in the proof of
Proposition \ref{TMultPropAnal} will show that the multiplication
of $\mathcal{F}[f(t,.)]$ by $I_\alpha(t)$ amounts to
multiplication of each fiber component by a number, hence
$I_\alpha(t)\mathcal{F}[f(t,.)]\in\tilde{\mathcal{D}}(\tilde\Sigma)$
for all $t\in\mathcal{I}$.
\end{remark}

Two useful facts about the time dependent Fourier transform can be
given by the following

\begin{proposition}\label{SmoothFourCompInt}
Let $\tilde f(t,\alpha)\in
C_0^\infty\left(\mathcal{I},\tilde{\mathcal{D}}(\tilde\Sigma)\right)$.
then
\begin{romanlist}
\item $f(t,\vec x)=\mathcal{F}^{-1}[\tilde
f(t,\alpha)]\in\mathcal{D}(\mathcal{T})$

\item $\int_\mathcal{I}dt\tilde
f(t,\alpha)\in\tilde{\mathcal{D}}(\tilde\Sigma)$.
\end{romanlist}
\end{proposition}

\begin{proof}
Let
$$
f(t,\vec x)=\mathcal{F}^{-1}[\tilde
f(t,\alpha)]=\int_{\tilde\Sigma}d\mu(\alpha)\tilde
f(t,\alpha)\zeta_\alpha(\vec x).
$$
For each $t\in\mathcal{I}$ we have $\tilde
f(t,\alpha)\in\tilde{\mathcal{D}}(\tilde\Sigma)$ and therefore
$f(t,\vec x)\in\mathcal{D}(\mathcal{T}_t)$. If the compact
interval $A\subset\mathcal{I}$ is such that $\forall t\notin A$,
$\tilde f(t,\alpha)=0$, then obviously $\forall t\notin A$,
$f(t,\vec x)=0$. Because the integration converges in
$L^2(\tilde\Sigma,\mu)$, differentiation can be interchanged with
the integral, thus $f(t,\vec x)$ is smooth in $t$. The part (i) is
proven.

Now write
$$
\tilde f(t,\alpha)=\int_{\Sigma_t}d\mu_{h_0}(\vec
x)\langle\bar\zeta_\alpha(\vec x),f(t,\vec
x)\rangle_{\mathfrak{g}^0},
$$
and
$$
\int_\mathcal{I}dt\tilde
f(t,\alpha)=\int_\mathcal{I}dt\int_{\Sigma_t}d\mu_{h_0}(\vec
x)\langle\bar\zeta_\alpha(\vec x),f(t,\vec
x)\rangle_{\mathfrak{g}^0}=
$$
$$
=\int_{\Sigma_t}d\mu_{h_0}(\vec x)\langle\bar\zeta_\alpha(\vec
x),\int_\mathcal{I}dtf(t,\vec
x)\rangle_{\mathfrak{g}^0}=\mathcal{F}[\int_\mathcal{I}dtf(t,\vec
x)],
$$
where Fubini's theorem was used with the justification that both
integrals run over compact supports \cite{Dieudonne1976}. For the
part (ii) it remains to show that $\int_\mathcal{I}dtf(t,\vec
x)\in\mathcal{D}(\mathcal{T}_t)$. But this is again clear because
the integral runs over a compact support.
\end{proof}

Next we want to show that the Cauchy problem can be well-posed in
the distributional sense. We will do it by generalizing
(Eq.\ref{CauchyExplicit}) to distributional solutions.

\begin{proposition}\label{CauchyProbDistrib}
For any $u_0,u_1\in\mathcal{D}(\mathcal{T}_t)'$ there exists a
unique $\jmath(u_0,u_1)=u\in\mathcal{D}(\mathcal{T})_0'$ such that
$$
u(f)=u_0(i_t^*(\nabla_tE[f]))-u_1(i_t^*(E[f]))\mbox{, }\forall
f\in\mathcal{D}(\mathcal{T}).
$$
\end{proposition}

\begin{proof}
By Proposition \ref{ESurj} we know that $E$ is surjective, so we
denote the bijective part of $E$ to be
$E_\updownarrow:\mathcal{D}(\mathcal{T})/\ker E\to
Sol_0(\mathcal{T})$. For surjectivity of $\jmath$ it suffices to
set
$$
u_0(v_1)=u(E_\updownarrow^{-1}[\jmath(0,v_1)])\mbox{,
}u_1(v_0)=-u(E_\updownarrow^{-1}[\jmath(v_0,0)])\mbox{, }\forall
v_0,v_1\in\mathcal{D}(\mathcal{T}_t).
$$
Indeed,
$$
u(f)=u(E_\updownarrow^{-1}[E_\updownarrow[f]])=u(E_\updownarrow^{-1}[\jmath(i_t^*(E_\updownarrow[f]),i_t^*(\nabla_tE_\updownarrow[f]))])=
$$
$$
=u(E_\updownarrow^{-1}[\jmath(i_t*(E_\updownarrow[f]),0)])+u(E_\updownarrow^{-1}[\jmath(0,i_t^*(\nabla_tE_\updownarrow[f]))])=u_0(i_t^*(\nabla_tE[f]))-u_1(i_t^*(E[f])).
$$
For injectivity of $\jmath$ let
$u_0,u_1\in\mathcal{D}(\mathcal{T}_t)'$ be given. Define $u$ as in
the statement. Then obviously $u(Df)=0$ because $EDf=0$ for any
$f\in\mathcal{D}(\mathcal{T})$, hence
$u\in\mathcal{D}(\mathcal{T}_t)_0'$. Now suppose the same formula
holds also for different $u_0',u_1'\in\mathcal{D}(\mathcal{T}_t)'$
with the same $u$. Then we have
$$
0=(u_0-u_0')(i_t^*(\nabla_tE[f]))-(u_1-u_1')(i_t^*(E[f]))\mbox{,
}\forall f\in\mathcal{D}(\mathcal{T}).
$$
Evaluating on $f=E_\updownarrow^{-1}(\jmath(v_0,0))$ and
$g=E_\updownarrow^{-1}(\jmath(0,v_1))$ for arbitrary
$v_0,v_1\in\mathcal{D}(\mathcal{T}_t)$ we find $u_0=u_0'$ and
$u_1=u_1'$.
\end{proof}

Now we come to the main assertion. Let the modes $T_\alpha$ be
chosen such that (Eq.\ref{TMultCond}) holds.

\begin{proposition}\label{WeakSolDecomp} Under the assumptions
made, there exist closed topological subspaces
$\tilde{\mathcal{D}}^u(\tilde\Sigma),\tilde{\mathcal{D}}^v(\tilde\Sigma)\subset\tilde{\mathcal{D}}(\tilde\Sigma)$,
such that for any $\psi\in\mathcal{D}(\mathcal{T})_0'$ there are
unique distributions
$a^{\psi}\in\tilde{\mathcal{D}}^u(\tilde\Sigma)'$,
$b^{\psi}\in\tilde{\mathcal{D}}^v(\tilde\Sigma)'$ with
$$
\psi(f)=a^{\psi}(u_\alpha(f))+b^{\psi}(v_\alpha(f))\mbox{,
}\forall f\in\mathcal{D}(\mathcal{T}).
$$
\end{proposition}

\begin{proof}
Considered as distributions, the functions $u_\alpha$ act as
$$
u_\alpha(f)=\langle
u_\alpha,f\rangle_M=\int_\mathcal{I}dtT_\alpha(t)\langle\zeta_\alpha,f(t,.)\rangle_{\Sigma_t}=\int_\mathcal{I}dtT_\alpha(t)(\zeta_{-\alpha},\check\Gamma
f(t,.))_{\Sigma_t}=
$$
$$
=\int_\mathcal{I}dtT_\alpha(t)I_\alpha(t)\mathcal{F}[\check\Gamma
f(t,.)](-\alpha)\mbox{, }\forall f\in\mathcal{D}(\mathcal{T}).
$$
(Remember that $g_{00}=1$.) The action of $v_\alpha$ is similar.
By assumption (Eq.\ref{TMultCond}) and Remark \ref{MultI_alpha} we
find
$$
T_\alpha(t)I_\alpha(t)\mathcal{F}[\check\Gamma
f(t,.)](-\alpha)\in\tilde{\mathcal{D}}(\tilde\Sigma).
$$
Then by Proposition \ref{SmoothFourCompInt} we get
$u_\alpha(f)\in\tilde{\mathcal{D}}(\tilde\Sigma)$ (similarly for
$v_\alpha$).

In general, the maps $f\to u_\alpha(f)$ and $f\to v_\alpha(f)$
need not be surjective. Therefore we define
$$
\tilde{\mathcal{D}}^u(\tilde\Sigma)=u_\alpha(\mathcal{D}(\mathcal{T})).
$$
By continuity of the map $f\to u_\alpha(f)$ (which is easy to
establish), $\tilde{\mathcal{D}}^u(\tilde\Sigma)$ is a closed
subspace of $\tilde{\mathcal{D}}(\tilde\Sigma)$. Similarly we
define $\tilde{\mathcal{D}}^v(\tilde\Sigma)$.

Recall the mode expansion for arbitrary $\phi\in
Sol_0(\mathcal{T})$,
$$
\phi(t,\vec x)=\int_{\tilde\Sigma}d\mu(\alpha)\left[a^\phi_\alpha
u_\alpha(x)+b^\phi_\alpha v_\alpha(x)\right].
$$
Thus $Sol_0(\mathcal{T})$ can be written as a direct sum of linear
subspaces, $Sol_0(\mathcal{T})=Sol_0^u(\mathcal{T})\oplus
Sol_0^v(\mathcal{T})$, with
$$
Sol_0^u(\mathcal{T})=\{\phi\in Sol_0(\mathcal{T})\mbox{:
}b^\phi=0\}\mbox{, }Sol_0^v(\mathcal{T})=\{\phi\in
Sol_0(\mathcal{T})\mbox{: }a^\phi=0\},
$$
and we will write $\phi=\phi^u+\phi^v$. Regarding as a
distribution in $\mathcal{D}(\mathcal{T})'$, $\phi^u$ and $\phi^v$
act as
$$
\phi^u(f)=\int_{\tilde\Sigma}d\nu(\alpha)a^\phi_\alpha
u_\alpha(f)\mbox{,
}\phi^v(f)=\int_{\tilde\Sigma}d\nu(\alpha)b^\phi_\alpha
v_\alpha(f).
$$
The functions $a^\phi_\alpha,b^\phi_\alpha$ can be regarded as
distributions $a^\phi\in\tilde{\mathcal{D}}^u(\tilde\Sigma)'$,
$b^\phi\in\tilde{\mathcal{D}}^v(\tilde\Sigma)'$, and we can write
\begin{eqnarray}
\phi(f)=\phi^u(f)+\phi^v(f)=a^\phi(u_\alpha(f))+b^\phi(v_\alpha(f)).\label{Sol0ModDecomp}
\end{eqnarray}
Now let $\varphi\in Sol(\mathcal{T})$ be a solution, which does
not necessarily have $supp\{\varphi\}\cap\Sigma_t$ compact. Its
Cauchy data are
$$
(i^*_t(\varphi),i^*_t(\nabla_t\varphi))=(\varphi_0,\varphi_1)\in\mathcal{E}(\mathcal{T}_t)\oplus\mathcal{E}(\mathcal{T}_t).
$$
Choosing a countable (compactly finite) partition of unity on
$\Sigma$ we can write
$$
(\varphi_0,\varphi_1)=\sum_{i=1}^\infty(\phi^i_0,\phi^i_1)\mbox{,
}(\phi^i_0,\phi^i_1)\in\mathcal{D}(\mathcal{T}_t)\oplus\mathcal{D}(\mathcal{T}_t),
$$
where the sum involves finite items on any compact region
$U\in\Sigma$. Now for each $i$ we have
$$
\phi^i=i^{-1}_t(\phi^i_0,\phi^i_1)\in Sol_0(\mathcal{T}),
$$
thus
$$
\phi^i=\phi^{i,u}+\phi^{i,v}=i^{-1}_t(\phi^{i,u}_0,\phi^{i,u}_1)+i^{-1}_t(\phi^{i,v}_0,\phi^{i,v}_1)\mbox{,
}\phi^{i,u}\in Sol^u_0(\mathcal{T})\mbox{, }\phi^{i,v}\in
Sol^v_0(\mathcal{T}).
$$
Set
$$
\varphi^u=\sum_{i=1}^\infty
i^{-1}_t(\phi^{i,u}_0,\phi^{i,u}_1)=\sum_{i=1}^\infty\phi^{i,u},
$$
and
$$
\varphi^v=\sum_{i=1}^\infty
i^{-1}_t(\phi^{i,v}_0,\phi^{i,v}_1)=\sum_{i=1}^\infty\phi^{i,v},
$$
where the sums converge in compact topology. (This can be seen as
follows. The intersection of the causal cone of any compact region
with a Cauchy surface is a compact surface, and therefore only
finite summands survive.) But we have
$$
\phi^{i,u}(f)=a^{\phi^i}(u_\alpha(f))\mbox{,
}\phi^{i,v}(f)=b^{\phi^i}(v_\alpha(f))
$$
for some distributions $a^{\phi^i}$ and $b^{\phi^i}$. Thus we
obtain
$$
\varphi^u(f)=\sum_{i=1}^\infty\phi^{i,u}(f)=\sum_{i=1}^\infty
a^{\phi^i}(u_\alpha(f))
$$
and
$$
\varphi^v(f)=\sum_{i=1}^\infty\phi^{i,v}(f)=\sum_{i=1}^\infty
b^{\phi^i}(v_\alpha(f)).
$$
This convergence defines distributions
$$
a^\varphi=\sum_{i=1}^\infty
a^{\phi^i}\in\tilde{\mathcal{D}}^u(\tilde\Sigma)'\mbox{,
}b^\varphi=\sum_{i=1}^\infty
b^{\phi^i}\in\tilde{\mathcal{D}}^v(\tilde\Sigma)',
$$
such that
$$
\varphi^u(f)=a^\varphi(u_\alpha(f))\mbox{,
}\varphi^v(f)=b^\varphi(v_\alpha(f))\mbox{,
}\varphi=\varphi^u+\varphi^v\mbox{, }\forall
f\in\mathcal{D}(\mathcal{T}),
$$
and thus $Sol(\mathcal{T})=Sol^u(\mathcal{T})\oplus
Sol^v(\mathcal{T})$, where
$$
Sol^u(\mathcal{T})=\{\varphi\in Sol(\mathcal{T})\mbox{:
}b^\varphi=0\}\mbox{, }Sol^v(\mathcal{T})=\{\varphi\in
Sol(\mathcal{T})\mbox{: }a^\varphi=0\}.
$$

Now let $\psi\in\mathcal{D}(\mathcal{T})_0'$ be a weak solution,
and $\{\chi_m\}$ a usual mollifier on $\Sigma_t$. Define the
mollifications $\heartsuit_m\psi\in Sol(\mathcal{T})$ by
$$
\heartsuit_m\psi=\jmath(\chi_m\psi_0,\chi_m\psi_1),
$$
where $\psi=\jmath(\psi_0,\psi_1)$ by Proposition
\ref{CauchyProbDistrib}. Then it is easy to see that
$\heartsuit_m\psi\to\psi$ in $\mathcal{D}(\mathcal{T})'$. That
$\heartsuit_m\psi\in Sol(\mathcal{T})$ it follows
$$
\heartsuit_m\psi=(\heartsuit_m\psi)^u+(\heartsuit_m\psi)^v\mbox{,
}(\heartsuit_m\psi)^\bullet\in Sol^\bullet(\mathcal{T}).
$$
The disjointness $Sol^u(\mathcal{T})\cap Sol^v(\mathcal{T})=0$
implies that $(\heartsuit_m\psi)^u\to\psi^u$ and
$(\heartsuit_m\psi)^v\to\psi^v$ with some distributions
$\psi^u\in\overline{Sol^u(\mathcal{T})}$,
$\psi^v\in\overline{Sol^v(\mathcal{T})}$, such that
$\psi=\psi^u+\psi^v$. We denote
$$
\psi^u(f)=\lim_{m\to\infty}(\heartsuit_m\psi)^u(f)=\lim_{m\to\infty}a^{\psi_m}(u_\alpha(f))\doteq
a^{\psi}(u_\alpha(f)),
$$
$$
\psi^v(f)=\lim_{m\to\infty}(\heartsuit_m\psi)^v(f)=\lim_{m\to\infty}b^{\psi_m}(v_\alpha(f))\doteq
b^{\psi}(v_\alpha(f)),
$$
for some distributions $a^{\psi}$ and $b^{\psi}$. Finally we
arrive at
$$
\psi(f)=a^{\psi}(u_\alpha(f))+b^{\psi}(v_\alpha(f))\mbox{,
}\forall f\in\mathcal{D}(\mathcal{T}).
$$
The map
$$
\mathcal{D}(\mathcal{T})_0'\ni\psi\to(a^{\psi},b^{\psi})\in\tilde{\mathcal{D}}^u(\tilde\Sigma)'\oplus\tilde{\mathcal{D}}^v(\tilde\Sigma)'
$$
is a bijection by construction.
\end{proof}

\subsection{The propagator}

In this section we will find the explicit form of the propagator
$E$ in terms of the mode decomposition. Of course, Green's
functions can be calculated using the techniques of inverse
operators. But our approach will be more concordant to the spirit
of this work and will at the same time demonstrate the usefulness
of the mode decomposition in general.

To use the mode decomposition for weak solutions we assume that at
least the condition (iv) of the conventional Fourier transform
holds, and that the assumptions of Proposition \ref{WeakSolDecomp}
are satisfied. Choose mode solutions to be such that
$T_\alpha(0)=T_{-\alpha}(0)$ and $\dot T_\alpha(0)=\dot
T_{-\alpha}(0)$. Then because $\alpha\to-\alpha$ preserves both
$\lambda_\alpha(t)$ and the component $\tilde\Sigma^i$, we have
the same mode equations for $T_\alpha$ and $T_{-\alpha}$, hence
everywhere $T_\alpha(t)=T_{-\alpha}(t)$.

The function
$$
\det W[T_\alpha,\bar T_\alpha](t)=I_\alpha(t)\left[\dot
T_\alpha(t)\bar
T_\alpha(t)-T_\alpha(t)\dot{\bar{T}}_\alpha(t)\right]\in
C^\infty(\mathcal{I},i\cdot \mathbb{R})
$$
is the Wronskian of two independent solutions $T_\alpha$ and $\bar
T_\alpha$ and is therefore an imaginary constant. For convenience
we appoint once and forever to consider only the modes normalized
by
\begin{eqnarray}
\dot T_\alpha(t)\bar
T_\alpha(t)-T_\alpha(t)\dot{\bar{T}}_\alpha(t)=i\cdot
I^{-1}_\alpha(t).\label{ModeTNorm}
\end{eqnarray}
It can be seen that this condition is consistent with our previous
assumptions for the modes $T_\alpha$.\index{Mode normalization}

We remark that the Krein space involution $\check\Gamma$ commutes
with the connection components $\Gamma_i$. Indeed, by definition
$\check\Gamma=P^+-P^-$, where $P^\pm$ are the projections onto the
subspaces of positive/negative definiteness of the metric
$\langle,\rangle_g$. Let $\{e_i\}$ be a pseudo-orthonormal moving
frame, i.e., $\langle e_i,e_i\rangle_g=\pm1$. The value of each
$\langle e_i,e_i\rangle_g$ is preserved under $\nabla$, and
therefore $e_i$ remains in the same eigenspace of $\check\Gamma$,
although in our main frame $e_i$ experiences gradient,
$$
\nabla e_i=\sum_{j=1}^4\sum_{k=1}^n\Gamma^{k}_{ji}dx^j\otimes e_k.
$$
Hence $\check\Gamma$ commutes with all $\Gamma_i$. We have that
$$
\langle u,v\rangle_{\Sigma_t}=(\check\Gamma\bar
u,v)_{\Sigma_t}=\int_{\tilde\Sigma}d\mu(\alpha)s(\alpha)I_\alpha(t)\tilde
u(-\alpha)\tilde v(\alpha)\mbox{, }\forall
u\in\mathcal{E}(\mathcal{T}_t)\mbox{,
}v\in\mathcal{D}(\mathcal{T}_t),
$$
where $s(\alpha)$ is the Fourier image of the Krein involution
$\check\Gamma$, which due to the remark above satisfies
$s(\tilde\Sigma^i)=\{+1,-1\}$, i.e., is constant on each component
$\tilde\Sigma^i$. We have used the fact that
$\bar{\tilde{\bar{u}}}(\alpha)=\tilde u(-\alpha)$ which follows
from the condition (iv) of the conventional Fourier transform.

Now the propagator is the unique operator
$E:\mathcal{D}(\mathcal{T})\to Sol_0(\mathcal{T})$ which satisfies
\begin{eqnarray}
v(f)=\langle v(t;.),\dot E[f](t;.)\rangle_{\Sigma_t}-\langle\dot
v(t;.),E[f](t;.)\rangle_{\Sigma_t},\nonumber\\
\forall v\in Sol_0(\mathcal{T})\mbox{,
}f\in\mathcal{D}(\mathcal{T})\mbox{,
}t\in\mathcal{I}.\label{PropMainEq}
\end{eqnarray}
As $v\in Sol_0(\mathcal{T})$ we can write
\begin{eqnarray}
v(x)=\int_{\tilde\Sigma}d\mu(\alpha)a^v(\alpha)u_\alpha(x)+\int_{\tilde\Sigma}d\mu(\alpha)b^v(\alpha)v_\alpha(x),\label{vModDec}
\end{eqnarray}
and
$$
\widetilde{v(t;.)}(\alpha)=a^v(\alpha)T_\alpha(t)+b^v(\alpha)\bar
T_\alpha(t),
$$
$$
\widetilde{\dot{v}(t;.)}(\alpha)=a^v(\alpha)\dot
T_\alpha(t)+b^v(\alpha)\dot{\bar{T}}_\alpha(t).
$$
Similarly for $E[f]\in Sol_0(\mathcal{T})$,
$$
E[f](x)=\int_{\tilde\Sigma}d\mu(\alpha)a^E[f](\alpha)u_\alpha(x)+\int_{\tilde\Sigma}d\mu(\alpha)b^E[f](\alpha)v_\alpha(x),
$$
$$
\widetilde{E[f](t;.)}(\alpha)=a^E[f](\alpha)T_\alpha(t)+b^E[f](\alpha)\bar
T_\alpha(t),
$$
and
$$
\widetilde{\dot{E}[f](t;.)}(\alpha)=a^E[f](\alpha)\dot
T_\alpha(t)+b^E[f](\alpha)\dot{\bar{T}}_\alpha(t),
$$
with some distribution fields $a^E[f](\alpha)$ and
$b^E[f](\alpha)$. Using all this we compute
$$
\langle v(t;.),\dot E[f](t;.)\rangle_{\Sigma_t}-\langle\dot
v(t;.),E[f](t;.)\rangle_{\Sigma_t}=
$$
$$
=-\int_{\tilde\Sigma}d\mu(\alpha)s(\alpha)I_\alpha(t)\left[\dot
T_\alpha(t)\bar
T_\alpha(t)-T_\alpha(t)\dot{\bar{T}}_\alpha(t)\right]\left[a^v(-\alpha)b^E[f](\alpha)-b^v(-\alpha)a^E[f](\alpha)\right]=
$$
by normalization (Eq.\ref{ModeTNorm})
\begin{eqnarray}
=-i\int_{\tilde\Sigma}d\mu(\alpha)s(\alpha)I_\alpha(t)\left[a^v(-\alpha)b^E[f](\alpha)-b^v(-\alpha)a^E[f](\alpha)\right].\label{ECompIntRes1}
\end{eqnarray}
On the other hand, we know that $ED=0$, thus $a^E[f](\alpha)$ and
$b^E[f](\alpha)$ are weak solutions of the field equation and can
be mode decomposed as
\begin{eqnarray}
a^E[f](\alpha)=a^1_\alpha(u_\beta(f))+a^2_\alpha(v_\beta(f)),\nonumber\\
b^E[f](\alpha)=b^1_\alpha(u_\beta(f))+b^2_\alpha(v_\beta(f)).\label{ECompIntRes2}
\end{eqnarray}
By (Eq.\ref{vModDec}) we have
\begin{eqnarray}
v(f)=\int_{\tilde\Sigma}d\mu(\alpha)a^v(\alpha)u_\alpha(f)+\int_{\tilde\Sigma}d\mu(\alpha)b^v(\alpha)v_\alpha(f).\label{ECompIntRes3}
\end{eqnarray}
Inserting (Eq.\ref{ECompIntRes1}), (Eq.\ref{ECompIntRes2}) and
(Eq.\ref{ECompIntRes3}) into (Eq.\ref{PropMainEq}) we obtain
$$
a^1_\alpha=b^2_\alpha=0\mbox{, }a^2_\alpha=-b^1_\alpha=i\cdot
s(\alpha)\delta(\beta-\alpha).
$$
And our final formula is
$$
E[f](x)=i\int_{\tilde\Sigma}d\mu(\alpha)s(\alpha)\left[v_{-\alpha}(f)u_\alpha(x)-u_{-\alpha}(f)v_\alpha(x)\right],
$$
which is in full accord with the result obtained by
\cite{Lueders_Roberts_1990} for scalar fields on FRW spacetimes.

\section{Aspects of Harmonic Analysis in Homogeneous Spacetimes}

\subsection{Spatially homogeneous cosmological models}

The main goal of the current work is to refurbish the mathematical
framework of quantum field theory on classical cosmological
spacetimes, in general, and to advance towards a satisfactory
rigorous description of cosmological particle creation in states
of low energy for hyperbolic fields, in particular. The latter
would be an extension of results  obtained in
\cite{Degner_Verch_2010} for the Klein-Gordon field on specific
FRW models to more general situations. Thus although some results
were and will be obtained under abstract general assumptions, our
attention is concentrated at the geometrical setup of most common
cosmological models. Supported by observations of the universe at
large scale, cosmology considers mainly spatially homogeneous, or
in addition also isotropic, spacetimes. A condensed account of
cosmological arguments and their geometrical implications can be
found, for instance, in \cite{pittphilsci1507}. The essence of
these geometrical restrictions is mathematically expressed by
imposing the existence of a sufficiently rich system of symmetries
(more precisely, a group of spatial isometries) on the spacetime.
Extensive treatments of all possible isometry groups and  related
questions can be found in \cite{pittphilsci1507}, \cite{Petrov59},
\cite{StephaniKramerMacCallumHoenselaersHerlt200305}. An
introduction to the generalities of harmonic analysis on vector
bundles is given in \cite{Camporesi:1990wm}. In this section we
will try to deductively introduce our geometrical setup with the
help of the information in the above mentioned references.

{\bf Foliation by equal time Cauchy hypersurfaces.} Recall that we
are working with a four dimensional globally hyperbolic Lorentzian
manifold $(M,g)$ on which a global smooth time function and an
atlas can be chosen following \cite{Bernal_Sanchez_2005} such that
$M$ is foliated by three dimensional spacelike equal-time smooth
Cauchy hypersurfaces and
$$
ds^2=g_{00}dt^2-d\sigma^2,
$$
where $d\sigma^2$ is the line element on any of those Cauchy
surfaces being Riemannian submanifolds.

{\bf The structure group.} Any vector bundle $\mathcal{T}$ can be
considered as associated to its frame bundle
$\mathcal{P}_\mathcal{T}$ with structure group $GL(n)$. If we want
the fiberwise transformations to respect the fiber metric, then we
have to restrict the principal bundle to the orthogonal frame
bundle. All fibers $V_p$ with their respective non-degenerate
pseudo-Riemannian structures $\mathfrak{g}_p$ are isomorphic, and
their generalized orthogonal groups $O(\mathfrak{g}_p)$ (i.e.,
groups of invertible endomorphisms of $V_p$ preserving
$\mathfrak{g}_p$) are isomorphic to the generalized Lorentz group
$O(\pm_\mathfrak{g})$, where $\pm_\mathfrak{g}$ in this context
will be understood as the signature of $\mathfrak{g}$. But the
same vector bundle $\mathcal{T}$ can be associated also to another
principal bundle (which we again denote by
$\mathcal{P}_\mathcal{T}$) with structure group $H$ (say, for
field theoretical reasons). Then we have a representation $r$ of
$H$ on $V$. If $r$ also respects the metric, then $r(H)\in
O(\pm_\mathfrak{g})$, so $H$ is homomorphic to
$O(\pm_\mathfrak{g})$. For instance, $H=SO^+(\pm_\mathfrak{g})$
(tensor bundle) or $H=Spin^+(\pm_\mathfrak{g})$ (spinor bundle).

{\bf Isometries.} Let us start with reminding some definitions. An
{\it isometry} of the spacetime $(M,g)$ is a diffeomorphism
$\psi:M\to M$ such that $\psi^*g=g$ holds on $M$, where $\psi^*$
is the pullback of $\psi$. If $\psi':M\to M$ is another isometry,
then obviously such is also their superposition $\psi\circ\psi'$.
With the superposition as product, isometries thus constitute an
abstract group, which we will denote $\mbox{\bf Iso}(M)$. If
$\mathcal{T}\to M$ is a the vector bundle over $M$ as defined
previously, then an {\it isometry} of the vector bundle
$\mathcal{T}$ is a morphism $\Psi:\mathcal{T}\to\mathcal{T}$
covering an isometry of the base, $\pi\circ\Psi\circ\pi^{-1}\in
\mbox{\bf Iso}(M)$, such that $\Psi^*\mathfrak{g}=\mathfrak{g}$
and $\Psi^*\nabla=\nabla$ (or more precisely $\Psi^*D=D$ when a
normal hyperbolic field operator $D$ is specified), where $\Psi^*$
denotes pullback maps, $\mathfrak{g}$ is the pseudo-Riemannian
fiber metric, and $\nabla$ is the metric connection. Again via
superposition, the isometries of the bundle $\mathcal{T}$ comprise
an abstract group $\mbox{\bf Iso}(\mathcal{T})$.

The map $\mbox{\bf
Iso}(\mathcal{T})\ni\Psi\to\pi\circ\Psi\circ\pi^{-1}\in\mbox{\bf
Iso}(M)$ gives a homomorphism of $\mbox{\bf Iso}(\mathcal{T})$
into $\mbox{\bf Iso}(M)$. The image of this homomorphosm is a
subgroup of $\mbox{\bf Iso}(M)$ and will be denoted by $\mbox{\bf
Iso}^\mathcal{T}(M)\subset\mbox{\bf Iso}(M)$, and its kernel is a
normal subgroup of $\mbox{\bf Iso}(\mathcal{T})$. This kernel
$\mbox{\bf Iso}(\mathcal{T})/\mbox{\bf Iso}^\mathcal{T}(M)$
consists of isometries of the bundle $\mathcal{T}$ covering the
identity map of $M$. These are precisely the smooth sections in
the principle bundle
$\mathcal{P}_\mathcal{T}\xrightarrow{loc}M\times H$ of
$\mathcal{T}$, i.e., $\mbox{\bf Iso}(\mathcal{T})/\mbox{\bf
Iso}^\mathcal{T}(M)=C^\infty(\mathcal{P}_\mathcal{T})$. The group
multiplication is given by the pointwise multiplication of
sections.

{\bf Homogeneous bundle structure.} If the sections in the bundle
$\mathcal{T}$ are going to represent physical fields, than one
should have a concrete picture of how they transform under the
diffeomorphisms of the spacetime $M$. In case of the tensor bundle
this picture is automatically encoded in the pullback map. An
abstract vector bundle does not have such a structure by itself.
Thus a physical field theory has to specify a homomorphism
$\rho:\mbox{Diff}(M)\to C^\infty(\mathcal{P}_\mathcal{T})$. For
the tangent bundle $\rho(\psi)=d\psi$, $\psi\in\mbox{Diff}(M)$.
When considering arbitrary diffeomorphisms, then the structure
group should be $GL(n)$ rather than a smaller $H$. But if we
restrict $\rho$ to $\rho:\mbox{\bf Iso}^\mathcal{T}(M)\to
C^\infty(\mathcal{P}_\mathcal{T})$, then $H$ can be chosen. For
brevity denote $G=\mbox{\bf Iso}^\mathcal{T}(M)$. We have the
injection
$$
G\ni g\to g\times\rho(g)\in\mbox{\bf Iso}(\mathcal{T}),
$$
which gives sense to the left action of $G$ on $\mathcal{T}$ by
isometries.

The abstract group of isometries of a pseudo-Riemannian manifold
of dimension $m$ is given the compact open topology, in which it
becomes a Lie group of dimension at most $n(n+1)/2$
\cite{Helgason197901}. It can be further shown, that the compact
open topology in this case is equivalent to the pointwise
convergence topology of isometries. Thus we automatically obtain a
Lie group structure on $\mbox{\bf Iso}(M)$. Then
$G\subset\mbox{\bf Iso}(M)$ is a topological subgroup defined by
$$
G=\{\psi\in\mbox{\bf Iso}(M)\mbox{:
}\left(\psi\times\rho(\psi)\right)^*\mathfrak{g}=\mathfrak{g}\mbox{,
}\left(\psi\times\rho(\psi)\right)^*D=D\}.
$$
If $\rho$ is a continuous homomorphism, then all the operations in
the equations
$$
\left(\psi\times\rho(\psi)\right)^*\mathfrak{g}=\mathfrak{g}\mbox{,
}\left(\psi\times\rho(\psi)\right)^*D=D
$$
are continuous, and therefore the subspace $G$ of $\mbox{\bf
Iso}(M)$ defined by this equation is a closed topological
subspace. But then by Cartan's theorem $G$ is actually a Lie
subgroup, as it is a closed topological subgroup of the Lie group
$\mbox{\bf Iso}(M)$. Thus we have the structure of a
$G$-homogeneous vector bundle $\mathcal{T}$.

{\bf Spatially homogeneous bundle.} The bundle $\mathcal{T}$ will
be called {\it spatially homogeneous} if the orbits of $\mbox{\bf
Iso}^\mathcal{T}(M)$ are three dimensional smooth spacelike
hypersurfaces which foliate $M$. (Maybe it is worth mentioning
here that everywhere in this work we consider only connected
spacetimes $M$.) By Theorem 8.16 in
\cite{StephaniKramerMacCallumHoenselaersHerlt200305} there exists
a parametrization of these orbits by the affine parameter of the
family of normal geodesics, such that the metric takes the form
$$
ds^2=dt^2-d\sigma^2.
$$
On the other hand, our original foliation by equal time Cauchy
surfaces due to Theorem 1.1 in \cite{Bernal_Sanchez_2005} also
yielded such a metric form. We assume that the time function can
be chosen such that equal time Cauchy surfaces are the orbits of
$\mbox{\bf Iso}^\mathcal{T}(M)$ (probably this can be shown to be
true in general). We note that due to the transitive action of $G$
on $\Sigma_t$ for every $t$, it holds $G\subset\mbox{\bf
Iso}^{\mathcal{T}_t}(\Sigma_t)$. We did not write $G=\mbox{\bf
Iso}^{\mathcal{T}_t}(\Sigma_t)$ because it is possible that for
some $t\neq t'\in\mathcal{I}$, $\mbox{\bf
Iso}^{\mathcal{T}_t}(\Sigma_t)\neq\mbox{\bf
Iso}^{\mathcal{T}_{t'}}(\Sigma_{t'})$, i.e., for some time
instances the time slice may be more symmetric than usual. We will
concentrate on $G$, which is the maximal guaranteed amount of
symmetry which is present at any time. Thus we see that
$\mathcal{T}_t$ also has the structure of a $G$-homogeneous vector
bundle.

Consider the principle bundle $\mathcal{P}_{\mathcal{T}_t}$ of
$\mathcal{T}_t$, which is a subbundle of
$\mathcal{P}_\mathcal{T}$. The smooth left action of $G$ on
$\mathcal{T}_t$ gives a smooth left action of $G$ on
$\mathcal{P}_{\mathcal{T}_t}$ as well. This action allows one to
construct a global smooth section in
$\mathcal{P}_{\mathcal{T}_t}$, whence it follows that the bundle
$\mathcal{T}_t$ is trivial. Because
$M\sim\Sigma_t\times\mathcal{I}$, the whole bundle $\mathcal{T}$
is also trivial. Thus spatially homogeneous vector bundles over
$M$ are necessarily trivial.

The requirement that the field operator $D$ is $G$-invariant
implies that the function $m^\star(x)$ is in fact a function of
time only.

{\bf Homogeneous space structure.} Now let $\mbox{\bf
StabIso}^\mathcal{T}(p)\subset G$ be the stabilizer of $G$ at some
fixed point $p\in M$. Then $\mbox{\bf StabIso}^\mathcal{T}(p)$ is
a closed Lie subgroup by  Cartan's theorem. That for all $p\in M$,
the groups $\mbox{\bf StabIso}^\mathcal{T}(p)$ are isomorphic,
then we denote them all by $\mbox{\bf StabIso}^\mathcal{T}(M)$. In
this case the orbits $\Sigma_t$ of $G$ are diffeomorphic to the
homogeneous space $G/\mbox{\bf
StabIso}^\mathcal{T}(M)\doteq\Sigma$. Denote $O=\mbox{\bf
StabIso}^\mathcal{T}(M)^+$, the identity component. Then
$\Gamma=\mbox{\bf StabIso}^\mathcal{T}(M)/O$ is a discrete normal
subgroup of $G$. If the homogeneous space $\Sigma$ is itself a Lie
subgroup of $G$, then it acts on each $\Sigma_t$ simply
transitively.

{\bf The four dimensional reality.} As already mentioned, the
isometry group $\mbox{\bf Iso}(M)$ of the $n=4$ dimensional
spacetime $M$ is a Lie group of dimension at most $n(n+1)/2=10$.
Thus in principle one can construct all real Lie algebras
$\mathcal{G}$ of dimension up to 10, their corresponding connected
simply connected Lie groups $G$, then all discrete normal
subgroups $\Gamma$ of such $G$ etc., thereby exhausting all
possible isometry groups of $M$. This heavy task have been done by
Petrov et al \cite{Petrov59} and others
\cite{StephaniKramerMacCallumHoenselaersHerlt200305}, and all the
possibilities are listed in tables. It turned out that only the
Minkowski space has isometry group of maximal dimension 10, which
is the Poincar\'e group. Among all possibilities we are interested
in those whose orbits are $\Sigma_t$. Thus the dimension of $G$ is
at least three. There are three possibilities of six dimensional
such isometry groups, which correspond to FRW spacetimes. A number
of possibilities are available with four dimensional groups, which
correspond to the LRS spacetimes. And finally there are nine
classes of three dimensional real Lie groups $Bi(N)$ (called
Bianchi groups), which together with their factors $Bi(N)/\Gamma$
by discrete subgroups $\Gamma$ represent the isometry groups of
the spatially homogeneous spacetimes. It turned out further, that
in all these cases besides one (the so called Kantowski-Sachs
model) the isometry group is the semidirect product
$G=\Sigma\rtimes O$. In this case we will call $\mathcal{T}_t$ a
{\it semidirect homogeneous vector bundle}. In particular, for six
dimensional FRW groups, four dimensional LRS groups and three
dimensional Bianchi groups $O=SO(3)$, $SO(2)$ and $\{1\}$,
respectively. The normal subgroups $\Sigma$ are nothing else than
$Bi(N)/\Gamma$.

In next sections we will work on the semidirect homogeneous vector
bundles. After establishing the necessary mathematical framework,
we will obtain results concerning the structure of $G$-invariant
homogeneous bi-distributions.

\subsection{On harmonic analysis in semidirect homogeneous vector
bundles}

In this section we will collect information on harmonic analysis
in $G$-homogeneous vector bundles $\mathcal{T}\to G/O$ where
$G=\Sigma\rtimes O$ which will be useful later in the work. This
does not pretend to be self-contained or systematic; quite the
contrary, we will introduce mainly what we were not able to find
in the literature. Otherwise references will be provided.

{\bf Semidirect homogeneous vector bundles.} Let $G=\Sigma\rtimes
O$, where $O$ is a compact connected type I Lie subgroup, and
$\Sigma$ a connected normal type I Lie subgroup. Moreover, we
demand that the modular function of $\Sigma$ has a non-trivial
kernel, so that the representation theories of both $\Sigma$ and
$G$ are well under control by Theorem 7.50 of
\cite{Folland199502}. We note that this is the case for all
Bianchi groups which are in fact the only candidates for $\Sigma$
in our context. Let $\Sigma=G/O$ have a Riemannian structure $h$
which is invariant under the left action of $G$. Let further
$\mathcal{T}\to \Sigma$ be an $n$-dimensional (real or complex)
vector bundle with standard fiber $V$ and a pseudo-Riemannian
fiber metric $\mathfrak{g}$. Let there be a smooth left action of
$G$ on $\mathcal{T}$ covering the left multiplication of $G$ on
the base, such that the fiber metric is invariant under that
action. Then we will call $\mathcal{T}$ a {\it semidirect
$G$-homogeneous vector bundle}. If we choose an orthonormal frame
$\{X_i\}$ of $T^*_1\Sigma$ (or $\{Y_i\}$ of $\mathcal{T}|_1$), and
drag it throughout $\Sigma$ using the transitive left action of
$G$, we will obtain $G$-invariant global smooth frame $\{X_i\}$ in
$T^*\Sigma$ (similarly, $\{Y_i\}$ in $\mathcal{T}$). Thus both
$T^*\Sigma$ and $\mathcal{T}$ are trivial bundles. Associated to
the Riemannian structure $h$ there is a Laplace operator $\Delta$
acting on sections $f\in C^\infty(\mathcal{T})$.

{\bf The regular and quasi-regular representations for the line
bundle.}\index{Quasi-regular representation} Suppose $\mathcal{T}$
from above is a line bundle, $n=1$. The left regular
representation $L_g$ of $G$ on $C(G)$ acts as
$$
L_gf(x)=f(g^{-1}x)\mbox{, }\forall g,x\in G.
$$
Because the Riemannian structure is $G$-invariant, the metric
measure $dx$ is a left Haar measure on $G$, and hence $L_g$ is a
unitary representation on $L^2(G)$.

Now any point $x\in G$ can be uniquely written as $x=x_\Sigma
x_O$, where $x_\Sigma\in \Sigma$ and $x_O\in O$. Let $dx_\Sigma$
be the metric driven left $G$-invariant measure on $\Sigma$, and
$dx_O$ the Lebesgue measure on $O$ normalized to $|O|=1$. Then
$dx=dx_\Sigma dx_O$ gives a left Haar measure on $G$. Functions
$f$ on $\Sigma$ are identified with their right $O$-invariant
extensions to $G$, i.e., $f(xo)=f(x)=f(xO)$, for any $x\in G$,
$o\in O$. Thus $C(\Sigma)\in C(G)$ (similarly $L^2(\Sigma)\in
L^2(G), etc.$) and we may consider the restriction $U_g$ of the
left regular representation $L_g$ on $C(G)$ to $C(\Sigma)$. Its
action will be given by
$$
U_gf(x_\Sigma O)=f(g^{-1}x_\Sigma O)\mbox{, }\forall x_\Sigma\in
\Sigma\mbox{, }g\in G.
$$
The representation $U_g$ of $G$ is the left quasi-regular
representation, and it is nothing else but the induced
representation $\mbox{Ind}_O^G1$. Note that for $O=\{1\}$ we
simply have $G=\Sigma$ and $L_g=U_g$.

Neither $L_g$ nor $U_g$ need to be irreducible. The central
decomposition of $L_g$ is
$$
L_g=\int^\oplus_{\hat G} d\nu(\pi)L_g(\pi),
$$
where $\nu(\pi)$ is the Plancherel measure and
$L_g(\pi)=\pi\otimes1$ is the primary representation composed of
$mult(\pi,L_g)=\dim\bar\pi\in[1,\infty]$ copies of $\pi$
\cite{Folland199502}. The central decomposition of $U_g$ will be
$$
U_g=\int^\oplus_{\hat G_\Sigma} d\mu(\pi)U_g(\pi),
$$
where $\hat G_\Sigma\subset\hat G$, $d\mu$ is the spectral measure
of $U_g$ and for $\mu$-almost all $\pi$, $U_g(\pi)$ is a multiple
of $\pi$ (multiplicities $mult(\pi,U_g)$ and the measure
$d\mu(\pi)$ need to be determined). The corresponding Hilbert
space decompositions are
$$
L^2(G)=\int_{\hat G}^\oplus
d\nu(\pi)\mathcal{H}_\pi\otimes\mathcal{H}_{\bar\pi}
$$
and
$$
L^2(\Sigma)=\int_{\hat G_\Sigma}^\oplus d\mu(\pi)\mathcal{H}(\pi),
$$
where
$\mathcal{H}(\pi)=\mathcal{H}_\pi\otimes\mathbb{C}^{mult(\pi,U_g)}\subset\mathcal{H}_\pi\otimes\mathcal{H}_{\bar\pi}$.
Here $\mathbb{C}^{mult(\pi,U_g)}$ symbolizes some Hilbert space of
dimension $mult(\pi,U_g)$ which is finite or infinite.

In the following we will deal with $U_g$ keeping in mind that in
case $G=\Sigma$ everything reduces to $L_g$.

{\bf The operator $\Pi_\pi$.} Consider for any $\pi\in\hat G$ the
bounded operator
$$
\Pi_\pi=\int_Odo\pi(o).
$$
Then $\Pi_\pi$ is self adjoint,
$$
\Pi_\pi^*=\int_Odo\pi(o)^*=\int_Odo\pi(o^{-1})=\Pi_\pi.
$$
Moreover, because $O$ is unimodular, we have
$$
\pi(o)\Pi_\pi=\pi(o)\int_Odo'\pi(o')=\int_Od(oo')\pi(oo')=\Pi_\pi=\Pi_\pi\pi(o)\mbox{,
}\forall o\in O,
$$
and hence $\Pi_\pi$ is a projection,
$$
\Pi_\pi^2=\int_Odo\pi(o)\Pi_\pi=\int_Odo\Pi_\pi=\Pi_\pi.
$$
$\Pi_\pi$ is a projection onto an invariant subspace of $\pi|_O$.
Recall the operator $D_\pi$ of \cite{Folland199502} which
satisfied $D_\pi\pi(x)=\Delta^{\frac{1}{2}}(x)\pi(x)D_\pi$, for
all $x\in G$. In particular, we find that
$D_\pi\pi(o)=\pi(o)D_\pi$ for all $o\in O$, and consequently,
$D_\pi\Pi_\pi=\Pi_\pi D_\pi$.

{\bf The Fourier transform in $G/O$.}\index{Fourier transform,
harmonic analytical} The Fourier transform in $\Sigma=G/O$
associated to $U_g$ is naturally the restriction of that on $G$
associated to $L_g$; for $\mu$-almost all $\pi\in\hat G_\Sigma$
$$
\hat f(\pi)=\pi(f)D_\pi\in\mathcal{H}(\pi).
$$
For any $f\in C_0(\Sigma)$ and $\mu$-almost all $\pi\in\hat
G_\Sigma$ we have
\begin{eqnarray}
\pi(f)=\int_\Sigma dx_\Sigma\int_Odx_Of(x_\Sigma
O)\pi(x_\Sigma)\pi(x_O)=\int_\Sigma dx_\Sigma f(x_\Sigma
O)\pi(x_\Sigma)\Pi_\pi.\label{Pi_f_Def}
\end{eqnarray}
As usual we have $\pi(U_g f)=\pi(L_gf)=\pi(g)\pi(f)$ for $g\in G$,
$f\in C_0(\Sigma)$. The convolution $f\ast h$ has the property
that if $f\in C_0(G)$ and $h\in C_0(\Sigma)$ then $f\ast h\in
C_0(\Sigma)$. Moreover, it satisfies $\pi(f\ast h)=\pi(f)\pi(h)$.

{\bf The case of arbitrary $\mathcal{T}$.} Let now $\dim V=n\ge
1$. The left quasi-regular representation of $G$ on
$C^\infty(\mathcal{T})$ acts as
$$
U^{\mathcal{T}}_gf(x)=g^{-1}f(g^{-1}x)\mbox{, }\forall f\in
C^\infty(\mathcal{T}).
$$
Recall the $G$-invariant orthonormal frame $\{Y_i\}_{i=1}^n$ in
$\mathcal{T}$ and write any $f\in C^\infty(\mathcal{T})$ as
$f=\sum f^iY_i$. Using that $U^{\mathcal{T}}_gY_i=Y_i$ we find
$$
U^{\mathcal{T}}_gf(x)=\sum_{i=1}^n U_gf^i\times Y_i,
$$
where $U_g$ is the left quasi-regular representation of $G$ on
$C^\infty(\Sigma)$. Thus $U^{\mathcal{T}}_g=\oplus_nU_g$, and the
harmonic analysis of $U^{\mathcal{T}}_g$ is the same as that of
$U_g$ except that each primary subrepresentation of
$U^{\mathcal{T}}_g$ is the $n$-fold copy of the corresponding
primary subrepresentation of $U_g$. Making the identification
$C^\infty_0(\mathcal{T})\ni f\to \{f^i\}\in\oplus_n
C^\infty_0(\Sigma)$ we find the Fourier transform of $f\in
C^\infty_0(\mathcal{T})$ to be
$$
\hat f(\pi)=\oplus_{i=1}^n\hat f^i(\pi),
$$
or to say in words, a matrix with $n$ times more columns than that
of a scalar function. The inverse Fourier transform will be
$$
f(x)=\sum_{i=1}^n\int_{\hat
G_\Sigma}d\mu(\pi)Tr\left[D_\pi\Pi_\pi\pi^*(x)\hat
f^i(\pi)\right]\times Y_i(x).
$$

\section{On the Fourier transform of distributions}

Here we will collect miscellaneous facts about distributions and
their Fourier transform, which we did not meet in the literature.
We continue working with the semidirect homogeneous vector bundle
$\mathcal{T}$ with notations established earlier.

Let $\hat{\mathcal{D}}^{\mathcal{T}}(\hat G_\Sigma)$ be the image
of $D(\mathcal{T})=C^\infty_0(\mathcal{T})$ under the harmonic
analytical Fourier transform $f(x_\Sigma)\to\hat f(\pi)$. As we
have already seen, $\hat f(\pi)=\oplus_n\hat f^i(\pi)$, hence
$\hat{\mathcal{D}}^{\mathcal{T}}(\hat
G_\Sigma)=\oplus_n\hat{\mathcal{D}}(\Sigma)$, where
$\hat{\mathcal{D}}(\Sigma)$ is the image under the Fourier
transform of $C^\infty_0(\Sigma)$.
$\hat{\mathcal{D}}^{\mathcal{T}}(\hat G_\Sigma)$ inherits the
topology of $D(\mathcal{T})$ via the Fourier transform, and one
can consider the Fourier transform of distributions
$D(\mathcal{T})'\ni u\to\hat
u\in\hat{\mathcal{D}}^{\mathcal{T}}(\hat G_\Sigma)'$ given by
$\hat u(\hat f)=u(f)$.

The Fourier transform has the remarkable property that it
interchanges the local and global behaviors. Namely, the local
irregularities of a function $f$ are reflected in the decay
properties of $\hat f(\pi)$ at large $\pi$, and conversely, the
behavior at infinity of $f$ determines the local regularity of
$\hat f(\pi)$. The precise description of these phenomena requires
a thorough functional analytical investigation, which we,
unfortunately, have no possibility to perform here.

It is widely known that any distribution restricted to a compact
region is of finite order. In \cite{Gelfand} the general structure
of distributions of finite order has been found for
$D(\mathbb{R}^n)$. Following a similar pattern we present here a
partial generalization of that result. By Proposition
\ref{CountNormTopEquiv} let us choose the topology $(X_i,2,l^2)$
for convenience.

\begin{proposition}\label{FinOrdDistribStruct}
Let $\mathcal{T}_K$ be an $n$-dimensional (complex)
pseudo-Riemannian vector bundle over a connected parallelizable
(pseudo-)Riemannian manifold $K$, and let $\nabla$ be a fiber
metric connection. Every $u\in D(\mathcal{T}_K)'$ of finite order
has a representation
$$
u(f)=\sum_{q\le k}(F_{\alpha,q},P_{\alpha,q}(X_i)f)_2\mbox{,
}\forall f\in D(\mathcal{T}),
$$
where $F_{\alpha,q}\in L^2(\mathcal{T}_K)$ and the smallest
possible such $k$ is the order of $u$.
\end{proposition}

\begin{proof}
By our choice
$$
\|f\|_k=\sqrt{\sum_{q\le k}\|P_{\alpha,q}(X_i)f\|_2^2}.
$$
Let $k$ be the order of $u$, i.e., $u$ is continuous in
$\|.\|_k$-norm. Define the following linear injective map
$$
\mathcal{V}:D(\mathcal{T}_K)\to\Phi=\bigoplus_{q\le
k}L^2(\mathcal{T}_K)
$$
by
$$
\mathcal{V}(f)=\bigoplus_{q\le k}P_{\alpha,q}(X_i)f.
$$
Then obviously $\|\mathcal{V}(f)\|_\Phi=\|f\|_k$. If we denote by
$\Psi=\mathcal{V}\left(D(\mathcal{T}_K)\right)\subset\Phi$, then
$u\circ\mathcal{V}^{-1}$ is a continuous functional on $\Psi$ with
the norm $\|.\|_\Phi$, and thus by Hahn-Banach theorem can be
extended to a continuous functional $F\in\Phi'$. But $\Phi$ is a
Hilbert space, thus $\Phi'=\Phi$ and $F\in\Phi$, and for any
$\phi\in\Phi$,
$$
F(\phi)=\sum_{q\le k}(F_{\alpha,q},\phi_{\alpha,q})_2\mbox{,
}F_{\alpha,q}\in L^2(\mathcal{T}).
$$
This yields our desired formula
$$
u(f)=\sum_{q\le k}(F_{\alpha,q},P_{\alpha,q}(X_i)f)_2.
$$
If such a formula held for a smaller $k$, then obviously the order
of $u$ would be smaller.
\end{proof}

Several variations of this proposition may be established by
choosing different norms. Note that the order of a distribution,
if finite, depends on the choice of the family of norms defining
the topology.

\begin{remark}\label{LocFinOrd} As already mentioned, any distribution is {\it
locally} of finite order, hence the proposition applies to the
restriction $u_K\in C^\infty_0(\mathcal{T}|_K)'$ of any $u\in
D(\mathcal{T})'$ to arbitrary compact connected region $K\subset
\Sigma$.
\end{remark}

We come back to our homogeneous bundle $\mathcal{T}$ and proceed
to the Fourier description of distributions $u\in D(\mathcal{T})'$
of finite order, which again can be applied for restrictions to
compact regions.

\begin{proposition}\label{FinOrdDistribFourierStruct}
Any distribution $u\in D(\mathcal{T})'$ of finite order is given
by
$$
u(f)=\int_{\hat G_\Sigma}d\mu(\pi)Tr\left[\hat u(\pi)^*\hat
f(\pi)\right],
$$
where $\hat u(\pi):\mathbb{C}^{mult(\pi,U_g)*n}\to\mathcal{H}_\pi$
is a $\mu$-locally integrable field of Hilbert-Schmidt operators.
(Note that the trace operator includes also the summation by fiber
indices $i=1,..,n$, which now enumerate blocks of columns.)
\end{proposition}

\begin{proof}
Let $k$ be the order of $u$. Choose $\{X_i\}$ to be the generators
of left translations on $C^\infty(\mathcal{T})$ and let by
Proposition \ref{FinOrdDistribStruct} write $u$ as
$$
u(f)=\sum_{q\le k}(F_{\alpha,q},P_{\alpha,q}(X_i)f)_2.
$$
Consider the Fourier transform
$$
\widehat{X_if}(\pi)=\int_\Sigma dx_\Sigma\left(\lim_{t\to
0}\frac{(U_{\exp(-t\xi_i)}-1)f(x_\Sigma)}{t}\right)\pi(x_\Sigma)\Pi_\pi
D_\pi
$$
where $\xi_i$ is the corresponding element of the Lie algebra of
$\Sigma$. The integral runs over a compact region, and is
therefore uniformly absolutely convergent with the Hilbert-Schmidt
norm, thus we can interchange the limit with the integral,
$$
\widehat{X_if}(\pi)=\lim_{t\to 0}\frac{1}{t}\int_\Sigma
dx_\Sigma(U_{\exp(-t\xi_i)}-1)f(x_\Sigma)\pi(x_\Sigma)\Pi_\pi
D_\pi=
$$
$$
=\lim_{t\to 0}\frac{\pi(\exp(-t\xi_i))-1}{t}\hat f(\pi).
$$
On the right hand we see nothing else but the generator of the
derived representation of $\pi$,
$$
\lim_{t\to 0}\frac{\pi(\exp(-t\xi_i))-1}{t}=-\partial_i\pi,
$$
whence we find
$$
\widehat{X_if}(\pi)=-\partial_i\pi\hat f(\pi).
$$
As a result we have
$$
\widehat{P_{\alpha,q}(X_i)f}(\pi)=P_{\alpha,q}(-\partial_i\pi)\hat
f(\pi),
$$
and thereby
$$
u(f)=\sum_{q\le k}\int_{\hat G_\Sigma}d\mu(\pi)Tr\left[\hat
F_{\alpha,q}(\pi)^*P_{\alpha,q}(-\partial_i\pi)\hat
f(\pi)\right]=\int_{\hat G_\Sigma}d\mu(\pi)Tr\left[\hat
u(\pi)^*\hat f(\pi)\right],
$$
where
$$
\hat u(\pi)=\sum_{q\le
k}\left[P_{\alpha,q}(-\partial_i\pi)\right]^*\hat
F_{\alpha,q}(\pi).
$$
This completes the proof.
\end{proof}

Such a result should not be surprising. If the measurable
functions $F_{\alpha,q}$ were $q$ times differentiable within the
space of locally integrable functions, then we could
hypothetically use integration by parts to make all the terms in
the formula of Proposition \ref{FinOrdDistribStruct} of order 0,
which would correspond to a regular distribution. The failure of
the derivatives of $F_{\alpha,q}$ to remain locally integrable is
reflected in the fact, that multiplication of
$\widehat{F}_{\alpha,q}(\pi)$ by $\partial_i\pi^*$ makes it not
square integrable any more, but possibly only locally integrable.
This reflects the local-to-global interchange made by the Fourier
transform: higher frequencies feel local irregularities.

The image $\hat{\mathcal{D}}^{\mathcal{T}}(\hat G_\Sigma)$ of
compactly supported smooth sections under the Fourier transform is
of considerable interest. In harmonic analysis it is described by
various Paley-Wiener type theorems. Although there are refined
Paley-Wiener theorems for adapted Fourier transforms for certain
classes of semisimple or solvable groups, there seems to be no
such theory for the general abstract setup. Next we present a
partial answer to the problem, namely, a criterion for smoothness
for sufficiently decaying functions, which gives hints about how
the general solution might look like.

\begin{proposition}\label{FourImgRapidDecPi}
For a function $f\in L^2(\mathcal{T})$ the following two
conditions are equivalent:

(i) for any polynomial $P(X_i)$ of generators $\{X_i\}$ with
constant coefficients, $P(X_i)f\in L^2(\mathcal{T})$

(ii) $\hat f(\pi)$ decays at infinity of $\hat G_\Sigma$ faster
than the inverse of any polynomial in the generators
$\partial_i\pi^*$
\end{proposition}

\begin{proof}
As we have seen in the proof of the previous proposition,
$$
\widehat{P(X_i)f}=P(-\partial_i\pi)^*\hat f(\pi),
$$
and the requirement that $\widehat{P(X_i)f}\in L^2(\hat G_\Sigma)$
for any $P(X_i)$ is equivalent to the assertion (ii) of the
proposition.
\end{proof}

We can go a step further and establish a weaker necessary
condition for a distribution to be given by a smooth integral
kernel. For this purpose we want to remind a few definitions on a
more abstract level.

Let $\mathcal{D}(S)$ be a test function space. We have
$\mathcal{D}(S)\subset L^\infty(S)$ and therefore
$L^\infty(S)'\subset\mathcal{D}(S)'$. Let $\{\eta_i\}$ be a finite
system of linear maps $\eta_i:S\to S$. A distribution
$u\in\mathcal{D}(S)'$ is of {\it rapid decay} in $\{\eta_i\}$ if
for any polynomial $P(\eta_i)$ of variables $\{\eta_i\}$ it holds
$u(P(\eta_i).)\in L^\infty(S)'$. We will symbolically write this
as $u=\mathfrak{o}(\{\eta_i\}^{-\infty})$. If $u$ is given by a
locally integrable kernel, and $\{\eta_i\}$ are coordinate
operators, then this definition obviously reduces to the usual
criterion for functions of rapid decay.

\begin{proposition}
For a distribution $u\in\mathcal{D}(\mathcal{T})'$ from
$\hat{\mathcal{D}}^{\mathcal{T}}(\hat G_\Sigma)'\ni\hat
u=\mathfrak{o}(\{\partial_i\pi\}^{-\infty})$ it follows that $u$
has a smooth integral kernel.
\end{proposition}

\begin{proof}
That $u$ is smooth means that all derivatives of all fiber
components $u^j$ are continuous. In other words, for any
polynomial $P$ in the generators $\{X_i\}$, the distributions
$P(X_i)u^j$ can be evaluated pointwise. A precise statement can be
given as follows. $u$ is smooth if and only if for any polynomial
$P(X_i)$, point $m\in \Sigma$ and sequence of test functions
$f_q\to\delta(x-m)$ in $C_0^\infty(\Sigma)'$, the following limit
exists for all $j=1,...,n$ and is finite,
$\lim_{q\to\infty}u^j(P(-X_i)f_q)$. The Fourier transform of the
distribution $\delta_m=\delta(x-m)$ can be easily read from the
Fourier inversion formula,
$$
\hat\delta_m(\hat f)=\int_{\hat
G_\Sigma}d\mu(\pi)Tr\left[D_\pi\Pi_\pi\pi^*(m)\hat f(\pi)\right].
$$
That means $f_q\to\delta(x-m)$ is equivalent to $\hat
f_q\to\pi(m)\Pi_\pi D_\pi$ in the weak sense. Hence
$$
\widehat{P(-X_i)f_q}\to P(\partial_i\pi)\pi(m)\Pi_\pi D_\pi
$$
in the weak topology. It follows
\begin{eqnarray}\label{uSmoothCondLim}
\lim_{q\to\infty}u^j(P(-X_i)f_q)=\hat
u^j\left(P(\partial_i\pi)\pi(m)\Pi_\pi D_\pi\right)
\end{eqnarray}
whenever one of the sides converges.

Now suppose $\hat u=\mathfrak{o}(\{\partial_i\pi\}^{-\infty})$.
Then for any $\hat f\in L^\infty(\hat G_\Sigma)$ (i.e., $\|\hat
f(\pi)\|\in L^\infty(\hat G_\Sigma)$ in the usual sense) we have
$$
\hat u^j\left(P(\partial_i\pi)\hat f^j(\pi)\right)<\infty\mbox{,
}j=1,...,n.
$$
In particular, $\pi(m)\Pi_\pi D_\pi\in L^\infty(\hat G_\Sigma)$,
whence (Eq.\ref{uSmoothCondLim}) follows.
\end{proof}

We are incline to think that this necessary condition is not far
from the desirable equivalent condition. This is, however, an open
problem in harmonic analysis, and we only hope to be able to give
a satisfactory answer in the future at least in the context we are
interested in.

\section{The adapted Fourier transform}

We start by noting that because the function $m^\star(t)$ is a
function of time only, the eigenfunctions of $D_{\Sigma_t}$ are
nothing else but the eigenfunctions of the Laplace operator
$\Delta_t$. In the first chapter we introduced the eigenfunction
decomposition associated to any self adjoint operator as the
Laplace operator $\Delta$,
$$
f\to\tilde f(\alpha)=\zeta_\alpha(f),
$$
where $\zeta_\alpha$-s are the generalized eigenfunctions of
$\Delta$. Putting additional structure related with particular
geometries one arrives at various Fourier transforms, which are
very practical in many respects. On the other hand, the abstract
harmonic analytical Fourier transform is a powerful tool for
analyzing general problems and properties, but its machinery is
functional analytically complicated for use. These two theories
are, however, related, although the exact relations have not been
sufficiently explored in the literature so far except for compact
groups. In the compact case the eigenfunctions of $\Delta$ are the
matrix elements of the irreducible representations for some choice
of the basis, and the two techniques can be unified. Each choice
of the basis results in a Fourier transform which is adapted to
it, hence such transforms are sometimes called adapted Fourier
transforms. In the non-compact case functional analytical
complications arise, though intuitively the situation remains
similar. In this section we will try to construct adapted Fourier
transforms at least on our semidirect homogeneous bundle
$\mathcal{T}$.\index{Fourier transform, adapted}

The Laplace operator $\Delta$ is invariant under $G$ and hence
commutes with $U^{\mathcal{T}}_g$. This means on each primary
component it acts as a multiplication from the right by a possibly
unbounded self-adjoint operator $\hat\Delta(\pi)$,
$$
\widehat{\Delta f}(\pi)=\hat f(\pi)\hat\Delta(\pi).
$$
For any $f\in L^2(\mathcal{T})$ we have that $\Delta f$ is a
distribution of order at most two. By Proposition
\ref{FinOrdDistribFourierStruct} it means that the multiplication
of any Hilbert-Schmidt operator $\hat f(\pi)$ by $\hat\Delta(\pi)$
from the right leaves it again Hilbert-Schmidt. Let
$\sigma(\pi)\subset\mathbb{R}$ be the spectrum of the self-adjoint
operator $\hat\Delta(\pi)$ as acting from the right (this spectrum
is non-positive, because $\Delta$ is an elliptic operator). For
each $\lambda\in\sigma(\pi)$ let $\hat\xi_{\pi,\lambda,r,s}$ be
the generalized eigenfunctions of $\hat\Delta(\pi)$, i.e.,
distributions satisfying
$\hat\xi_{\pi,\lambda,r,s}\hat\Delta(\pi)=\lambda\hat\xi_{\pi,\lambda,r,s}$
which are linearly independent and complete in $\mathcal{H}(\pi)$
for $r\in R_\pi\subset\mathbb{R}$ and $s\in
S^n_{\pi,\lambda}\subset\mathbb{R}$ (they can be constructed from
delta functions using the spectral theorem). Now consider the
following distributions in the Fourier space,
$$
\hat\zeta_{\pi,\lambda,r,s}(\pi')=\delta(\pi-\pi')\hat\xi_{\pi,\lambda,r,s}.
$$
Their preimages are distributions $\zeta_{\pi,\lambda,r,s}\in
D(\mathcal{T})'$ which are generalized eigenfunctions of $\Delta$,
and by elliptic regularity theorem, are smooth sections in
$\mathcal{T}$. Thus we have found, that the adapted Fourier
transform $\tilde f(\pi,\lambda,r,s)$ is nothing else but the
coefficients of $\hat f(\pi)$ as expended in the system
$\hat\xi_{\pi,\lambda,r,s}$. It is worth noting that $r$
parameterizes $\mathcal{H}_\pi$, and $\lambda$, $s$ parameterize
$\mathbb{C}^{mult(\pi,U_g)}*n$. Actually, $S^n_{\pi,\lambda}$
consists of $n$ copies of some set $S_{\pi,\lambda}$.

The choice of the system $\hat\xi_{\pi,\lambda,r,s}$ is rather
arbitrary and leaves room for adaptations. The first adaptation we
wish to make is the following. For any $\zeta_{\pi,\lambda,r,s}$
we want $\bar\zeta_{\pi,\lambda,r,s}=\zeta_{\pi',\lambda',r',s'}$
for some other parameters. Obviously $\lambda=\lambda'$, and it is
easy to see from the Fourier inversion formula, that this amounts
to requiring that $\bar\xi_{\pi,\lambda,r,s}$ enters the system
$\xi_{\bar\pi,\lambda,r',s'}$ for the representation $\bar\pi$
with some other parameters $r'$, $s'$. The representation
$\bar\pi$ may lie in the same equivalence class $[\pi]$ or not.

Lie groups are analytic manifolds, and all the group and algebra
structure is given by analytic functions in any analytic atlas. In
particular, the eigenfunction problem
$\Delta\zeta_{\pi,\lambda,r,s}=\lambda\zeta_{\pi,\lambda,r,s}$ is
an analytic elliptic equation, and the solutions
$\zeta_{\pi,\lambda,r,s}(x)$ are therefore analytic functions in
$x$. If $\Sigma$ is compact, then $\hat G_\Sigma$ is discrete, and
each $\sigma(\pi)$ is also discrete. Representations are finite
dimensional, hence $r$ and $s$ run over finite sets. The set
$\tilde\Sigma=\{\alpha=(\pi,\lambda,r,s)\}$ can be considered a
discrete manifold symbolically divided into $n$ components as
corresponding to each copy of $S_{\pi,\lambda}$. The space
$\hat{\mathcal{D}}^{\mathcal{T}}(\hat G_\Sigma)$ corresponds now
to the space $\tilde D(\tilde\Sigma)$ of functions on
$\tilde\Sigma$, which are of rapid decay in $\lambda$, and also
decay sufficiently fast in $\pi$ by Proposition
\ref{FourImgRapidDecPi}.

If $\Sigma$ is non-compact, suppose there exists a subset $\tilde
K\subset\hat G_\Sigma$ such that $\mu(\hat G_\Sigma\setminus\tilde
K)=0$ and $\tilde K$ can be cast into an analytic manifold. Then
we can restrict our Fourier transform from $\hat G_\Sigma$ to
$\tilde K$ without violation of the Plancherel equality. Suppose
further that the set $\tilde\Sigma=\{\alpha=(\pi,\lambda,r,s)\}$
can be made an analytic manifold consisting of $n$ disjoint
components as in the compact case. Each component itself may have
several connected components if $1<mult(\pi,U_g)<\infty$, in which
case $s$ will run over a discrete set. Then we can choose
$\zeta_{\pi,\lambda,r,s}$ to be analytic in all its parameters (if
$s$ is discrete, analyticity in $s$ is void), so that
$\hat{\mathcal{D}}^{\mathcal{T}}(\hat G_\Sigma)$ will correspond
to the space $\tilde D(\tilde\Sigma)$ of some analytic functions
on $\tilde\Sigma$ which have at least above mentioned decay
properties in $\lambda$ and $\pi$, but also are $L^2$ in $r$, and
in $s$ if the latter is continuous.

Finally let us define a symbolic involution $\alpha\to-\alpha$ on
$\tilde\Sigma$ satisfying $\zeta_{-\alpha}=\bar\zeta_\alpha$.
Clearly this involution will preserve $\lambda$. Now if the
necessary assumptions are satisfied, we arrive at a conventional
Fourier transform. In the next publication we will see that in the
majority of situations in cosmology these assumptions are valid,
and that will enable us to exploit the machinery of mode
decomposition to our cosmological models.

\subsection{Invariant bi-distributions}

In this section we will try to analyze the structure of
bi-distributions
$w\in\left(\mathcal{D}(\mathcal{T})\otimes\mathcal{D}(\mathcal{T})\right)'$
which are invariant under the left quasi-regular action
$U^{\mathcal{T}}_g$ of $G$ on $\mathcal{D}(\mathcal{T})$,
$$
w(U^{\mathcal{T}}_gf,U^{\mathcal{T}}_gh)=w(f,h)\mbox{, }\forall
f,h\in\mathcal{D}(\mathcal{T}),
$$
and compare with results obtained earlier in the
literature.\index{Invariant bi-distribution}

Decomposing each $f=\sum f^iY_i$, $f^i\in C^\infty_0(\Sigma)$, we
find for $u\in\mathcal{D}(\mathcal{T})'$ and
$w\in\left(\mathcal{D}(\mathcal{T})\otimes\mathcal{D}(\mathcal{T})\right)'$
$$
u(f)=\sum_{i=1}^nu^i(f^i)\mbox{,
}w(f,h)=\sum_{i,j=1}^nw^{ij}(f^i,h^j),
$$
$$
u^i\in C^\infty_0(\Sigma)'\mbox{,
}w^{ij}\in\left(C^\infty_0(\Sigma)\otimes
C^\infty_0(\Sigma)\right)',
$$
so that the problem reduces to that for scalar distributions.

The following proposition establishes the general form of the
$G$-invariant (or homogeneous) bi-distributions. Our approach is
greatly inspired by \cite{Gelfand_Vilenkin1964} where this
analysis is performed for $\mathbb{R}^n$.

\begin{proposition}\label{InvBiDistrib} Every $w\in
\left(C^\infty_0(\Sigma)\otimes C^\infty_0(\Sigma)\right)'$
satisfying $w(U_gf,U_gh)=w(f,h)$, $\forall f,h\in
C^\infty_0(\Sigma)$, $g\in G$, has the form
$$
w(f,h)=u_w\left(\bar f^*\ast h\right)
$$
for some $u_w\in C^\infty_0(\Sigma)'$. And conversely, any $u_w\in
C^\infty_0(\Sigma)'$ gives rise to such an invariant
bi-distribution $w$.
\end{proposition}

\begin{proof}
Recall that for scalar functions $U_gf(x_\Sigma
O)=f(g^{-1}x_\Sigma O)$. By the nuclear theorem $w$ can be
uniquely extended to $\tilde w\in C^\infty_0(\Sigma\times
\Sigma)'$ via embedding
$$
C^\infty_0(\Sigma)\otimes C^\infty_0(\Sigma)\ni f(x_\Sigma)\otimes
h(y_\Sigma)\to f(x_\Sigma)h(y_\Sigma)\in C^\infty_0(G\times G).
$$
That
$$
w(f(g^{-1}x_\Sigma O),h(g^{-1}y_\Sigma O))=w(f,h)
$$
by continuity implies that
$$
\tilde w(\phi(g^{-1}x_\Sigma O,g^{-1}y_\Sigma O))=\tilde
w(\phi(x_\Sigma,y_\Sigma))\mbox{, } \forall\phi\in
C^\infty_0(\Sigma\times \Sigma).
$$
Define the linear automorphism
$$
C^\infty_0(\Sigma\times
\Sigma)\ni\phi(x_\Sigma,y_\Sigma)\to\psi_\phi(x_\Sigma,y_\Sigma)\in
C^\infty_0(\Sigma\times \Sigma)
$$
by
$$
\psi_\phi(x_\Sigma,y_\Sigma)=\int_Odx_O\phi(x_\Sigma,x_\Sigma
x_Oy_\Sigma O)=\int_Odx_O\phi(x_\Sigma,x_\Sigma Ox_Oy_\Sigma O),
$$
and the pullback of $\tilde w$ under this automorphism by $\tilde
v$, $\tilde v(\psi_\phi)=\tilde w(\phi)$. If
$\phi_g(x_\Sigma,y_\Sigma)=\phi(g^{-1}x_\Sigma O,g^{-1}y_\Sigma
O)$ then
$$
\psi_{\phi_g}(x_\Sigma,y_\Sigma)=\int_Odx_O\phi(g^{-1}x_\Sigma
O,g^{-1}x_\Sigma Ox_Oy_\Sigma O)=\psi_\phi(g^{-1}x_\Sigma
O,y_\Sigma).
$$
Now
$$
\tilde w(\phi_g)=\tilde v(\psi_{\phi_g})=\tilde
v(\psi_\phi)=\tilde w(\phi),
$$
thus
$$
\tilde v(\psi_\phi(x_\Sigma,y_\Sigma))=\tilde
v(\psi_\phi(g^{-1}x_\Sigma O,y_\Sigma))\mbox{, }\forall g\in G.
$$
Consider the restriction $v$ of $\tilde v$ to
$C^\infty_0(\Sigma)\otimes C^\infty_0(\Sigma)$. The last equation
implies $v(f(g^{-1}x_\Sigma O),h(y_\Sigma))=v(f,h)$, $\forall
f,h\in C^\infty_0(\Sigma)$. If we fix $h$, then $v(.,h)\in
C^\infty_0(\Sigma)'$ is a distribution which is invariant under
all translations, and is thus given by a constant kernel,
$v(f,h)=u_w(h)\int_\Sigma dx_\Sigma f(x_\Sigma)$, for some
$u_w:C^\infty_0(\Sigma)\to\mathbb{C}$. On the other hand, if we
fix $f$, then continuity in $h$ implies $u_w\in
C^\infty_0(\Sigma)'$. Because the integral $\int_\Sigma dx_\Sigma
f(x_\Sigma)$ runs over a compact region, it can be transferred
into $u_w$, i.e., $v(f,h)=u_w\left(\int_\Sigma dx_\Sigma
f(x_\Sigma)h(y_\Sigma)\right)$. This in turn implies by
continuity, that $\tilde
v(\psi(x_\Sigma,y_\Sigma))=u_w\left(\int_\Sigma
dx_\Sigma\psi(x_\Sigma,y_\Sigma)\right)$. Finally we arrive at
$$
w(f,h)=\tilde w(f(x_\Sigma)h(y_\Sigma))=\tilde
v(f(x_\Sigma)\int_Odx_Oh(x_\Sigma x_oy_\Sigma O))=
$$
$$
=u_w\left(\int_\Sigma dx_\Sigma f(x_\Sigma)\int_Odx_Oh(x_\Sigma
x_Oy_\Sigma O)\right)=u_w(\bar f^*\ast h).
$$
The converse statement is obvious.
\end{proof}

For a distribution
$w\in\left(\mathcal{D}(\mathcal{T})\otimes\mathcal{D}(\mathcal{T})\right)'$
this will mean
$$
w(f,h)=\sum_{i,j=1}^nu_w^{ij}\left((\bar f^i)^*\ast h^j\right).
$$

\begin{remark}
Note that every $G$-invariant bi-distribution
$w\in\left(C^\infty_0(\Sigma)\otimes C^\infty_0(\Sigma)\right)'$
is in particular $\Sigma$-invariant. Let $f\ast h$ ($f\star h$)
and $f^*$ ($f^\star$) denote the convolution and the involution
with respect to $G$ ($\Sigma$), respectively. Then
$$
w(f,h)=u_w\left(\int_\Sigma dx_\Sigma
f(x_\Sigma)\int_Odx_Oh(x_\Sigma x_Oy_\Sigma O)\right)=
$$
$$
=u_w\left(\int_Odx_OL_{x_O^{-1}}\bar f^\star\star h(y_\Sigma
O)\right)=u'_w(\bar f^\star\star h)
$$
for some other $u'_w\in\left(C^\infty_0(\Sigma)\otimes
C^\infty_0(\Sigma)\right)'$ as expected.
\end{remark}

Let $\hat{\mathcal{D}}(\hat G_\Sigma)$ be the image of
$C^\infty_0(\Sigma)$ under the harmonic analytical Fourier
transform $f(x_\Sigma)\to\hat f(\pi)$. As an obvious corollary we
arrive at the form of an invariant bi-distribution in the Fourier
space.

\begin{corollary}\label{InvBiDistribFourier} A $G$-invariant bi-distribution $w\in\left(C^\infty_0(\Sigma)\otimes
C^\infty_0(\Sigma)\right)'$ in the Fourier space is given by
$$
w(f,h)=\hat u_w(\pi(\bar f^*)\hat h(\pi))=\hat u_w(\pi(\bar
f)^*\hat h(\pi))
$$
for some $\hat u_w\in\hat{\mathcal{D}}(\hat G_\Sigma)'$.
\end{corollary}

An immediate consequence of Proposition
\ref{FinOrdDistribFourierStruct} is the following
\begin{corollary}\label{InvFinOrdBiDistribFourier} Under the assumptions of Proposition
\ref{FinOrdDistribFourierStruct}, a $G$-invariant bi-distribution
$w_K\in\left(C^\infty_0(K)\otimes C^\infty_0(K)\right)'$ with
$K\subset \Sigma$ compact is given by
$$
w_K(f,h)=\int_{\hat G_\Sigma}d\mu(\pi)Tr\left[\hat
u_K(\pi)^*\pi(\bar f)^*\hat h(\pi)\right].
$$
\end{corollary}

\begin{proof}
It suffices to note that
$$
\mbox{supp}\{f\ast h\}\subset
O(\mbox{supp}\{f\})^{-1}\mbox{supp}\{h\}O,
$$
and to apply Proposition \ref{FinOrdDistribFourierStruct}.
\end{proof}

Finally we establish a generalization of the results by
\cite{Lueders_Roberts_1990} for FRW spacetimes.

\begin{proposition}\label{FinMultUnimodScalar}
Suppose that the group $G$ is such that all multiplicities
$mult(\pi,U_g)$ are finite. Then any $G$-invariant bi-distribution
$w\in\left(D(\mathcal{T})\otimes D(\mathcal{T})\right)'$ has the
form
$$
w(f,h)=\int_{\hat
G_\Sigma}d\mu(\pi)Tr\left[(\hat{\bar{f}}(\pi)\hat u(\pi))^*\hat
h(\pi)\right],
$$
where $\hat u(\pi)$ is a $\mu$-locally measurable field of
$\left[mult(\pi,U_g)\cdot n\right]\times\left[mult(\pi,U_g)\cdot
n\right]$ complex matrices.
\end{proposition}

\begin{proof}
Let start with the case $w\in\left(C^\infty_0(\Sigma)\otimes
C^\infty_0(\Sigma)\right)'$. The condition that the modular
function of $\Sigma$ has a nontrivial kernel ensures that the
formula (7.49) of (Folland) is valid, so that for $\mu$-almost all
$\pi$ the operator $D_\pi$ is invertible (injective). Therefore we
can write $\pi(f)=\hat f(\pi)D_\pi^{-1}$, so that $\pi(\bar
f)^*\hat h(\pi)=D_\pi^{-1}\hat{\bar{f}}(\pi)^*\hat h(\pi)$ where
$\hat{\bar{f}}(\pi)^*\hat h(\pi)$ is a $mult(\pi,U_g)\times
mult(\pi,U_g)$ complex matrix. Now for any compact $K\subset
\Sigma$ by Corollary \ref{InvFinOrdBiDistribFourier} we find that
the restriction $w_K$ of $w$ to $C^\infty_0(K)\otimes
C^\infty_0(K)$ is given by
$$
w_K(f,h)=\int_{\hat G_\Sigma}d\mu(\pi)Tr\left[\hat
u_K'(\pi)^*D_\pi^{-1}\hat{\bar{f}}(\pi)^*\hat
h(\pi)\right]=\int_{\hat G_\Sigma}d\mu(\pi)Tr\left[\hat
u_K(\pi)^*\hat{\bar{f}}(\pi)^*\hat h(\pi)\right],
$$
where $\hat u_K(\pi)$ is a $mult(\pi,U_g)\times mult(\pi,U_g)$
complex matrix. Choosing a larger compact $K\subset K'\subset
\Sigma$ we will arrive at another matrix $\hat u_{K'}(\pi)$. But
when restricted to $K$, $w_{K'}$ must coincide with $w_K$, hence
$\hat u_{K'}(\pi)=\hat u_K(\pi)$. Thus the matrix $\hat u_K(\pi)$
is the same for any $K$, and the formula holds for the entire $w$.

Now for $w\in\left(D(\mathcal{T})\otimes D(\mathcal{T})\right)'$
we have
$$
w(f,h)=\int_{\hat
G_\Sigma}d\mu(\pi)\sum_{i,j=1}^nTr\left[(\hat{\bar{f}}^i(\pi)\hat
u^{ij}(\pi))^*\hat h^j(\pi)\right]=\int_{\hat
G_\Sigma}d\mu(\pi)Tr\left[(\hat{\bar{f}}(\pi)\hat u(\pi))^*\hat
h(\pi)\right],
$$
which completes the proof.
\end{proof}

In the case of FRW spacetimes all the assumptions of the last
proposition are satisfied. In particular all $mult(\pi,U_g)=1$ and
for the scalar case we find that any $G$-invariant bi-distribution
is given by a locally measurable scalar field $\hat u(\pi)$.

\section{Acknowledgements}

The author expresses his thank to the Institute for Mathematics in
the Sciences and in particular to the International Max Planck
Research School for the organizational and material support of the
PhD project which this work is a part of. Invaluable is the
scientific patronage of the supervisor of this PhD project Prof.
Dr. Rainer Verch, whose remarks and suggestions used in this work
are too many to be listed here. Special gratitude is addressed to
Prof. Dr. Gerald Folland for extremely helpful e-mail
correspondence.

\appendix

\section{Space structures. Distributions}

Let us start with introducing symmetric metric products
$$
\langle f,h\rangle_M=\int_Md\mu_g(x)\langle
f(x),h(x)\rangle_g\mbox{, }f\in\mathcal{E}(\mathcal{T})\mbox{,
}h\in\mathcal{D}(\mathcal{T}),
$$
$$
\langle f,h\rangle_{\Sigma_t}=\int_\Sigma d\mu_h(\vec x)\langle
f(\vec x),h(\vec x)\rangle_g\mbox{,
}f\in\mathcal{E}(\mathcal{T}_t)\mbox{,
}h\in\mathcal{D}(\mathcal{T}_t).
$$
The pseudo-Riemannian metric $\langle,\rangle_\mathfrak{g}$
induces a Krein space structure on $V$, the typical fiber of
$\mathcal{T}$.\index{Krein space, involution} Whence there is a
Krein involution $\check\Gamma$, such that
$(u,v)_\mathfrak{g}=\langle\bar u,\check\Gamma
v\rangle_\mathfrak{g}$, $u,v\in V$, is a positive definite
hermitian inner product. This gives rise to positive definite
hermitian inner products
$$
(f,h)_M=\int_Md\mu_g(x)(f(x),h(x))_\mathfrak{g}\mbox{,
}f\in\mathcal{E}(\mathcal{T})\mbox{,
}h\in\mathcal{D}(\mathcal{T}),
$$
$$
(f,h)_{\Sigma_t}=\int_\Sigma d\mu_h(\vec x)(f(\vec x),h(\vec
x))_\mathfrak{g}\mbox{, }f\in\mathcal{E}(\mathcal{T}_t)\mbox{,
}h\in\mathcal{D}(\mathcal{T}_t).
$$
The completion of spaces $\mathcal{D}(\mathcal{T})$ and
$\mathcal{D}(\mathcal{T}_t)$ with respect to these products
becomes the Hilbert spaces $L^2(\mathcal{T})$ and
$L^2(\mathcal{T}_t)$, respectively. The tangent space $T_pM$ at a
point $p\in M$ with the Lorentzian metric $g$ is another example
of a Krein space. In the same spirit one defines the positive
definite inner product $(,)_g$ on $T_pM$. The metric $h$ on
$T_p\Sigma$ is Riemannian, so the construction of $(,)_h$ is
straightforward. Note that $\langle,\rangle_\mathfrak{g}$ and $g$
together give pseudo-Riemannian metrics on all product bundles
$T^*M\otimes...\otimes T^*M\otimes\mathcal{T}$ (respectively,
$\langle,\rangle_\mathfrak{g}$ and $h$ on
$T^*\Sigma\otimes...\otimes T^*\Sigma\otimes\mathcal{T}_t$). All
the resulting standard fibers are again Krein spaces, and can be
given inner products $(,)_g$ in the same fashion. These in their
turn produce products $(,)_M$ and $(,)_{\Sigma_t}$ on the
respective sections.

The perfect countably Banach topology of the test function spaces
$\mathcal{D}(\mathcal{T})$ and $\mathcal{D}(\mathcal{T}_t)$ can be
given as usual (e.g.,\cite{BarGinouxPfaffle200703}). However, as
we are going to perform a spectral analysis, we will need nuclear
countably Hilbert space structure, to which we proceed
\cite{Maurin1972}. Let $\mathcal{O}\subset M$ be a compact region.
Let
$$
\mathcal{D}_\mathcal{O}(\mathcal{T})=\left\{f\in\mathcal{D}(\mathcal{T})\mbox{:
}supp\{f\}\subset\mathcal{O}\right\}
$$
and define the family of positive definite inner products
$(,)_{\mathcal{O},p}$ on $\mathcal{D}_\mathcal{O}(\mathcal{T})$ by
$$
(f,h)_{\mathcal{O},p}=\sum_{q\le
p}((\nabla)^qf,(\nabla)^qh)_M\mbox{, }\forall
f,h\in\mathcal{D}_\mathcal{O}(\mathcal{T})\mbox{,
}p,q\in\mathbb{N},
$$
which induces a family of norms $\|.\|_{\mathcal{O},p}$. One can
show that this family of norms is growing and consistent, and
gives the same topology as the usual one. Let us give
$\mathcal{D}_\mathcal{O}(\mathcal{T})$ a countably Hilbert space
structure in the following sense,
$$
\mathcal{D}_\mathcal{O}(\mathcal{T})=\bigcap_{\mathbb{N}}\overline{\mathcal{D}_\mathcal{O}(\mathcal{T})}^{(,)_{\mathcal{O},p}}.
$$
It can be shown, that thus constructed countably Hilbert space
$\mathcal{D}_\mathcal{O}(\mathcal{T})$ is nuclear. Let now
$$
\mathcal{O}_1\subset...\subset\mathcal{O}_n\subset...\subset M
$$
be an infinite family of growing compact regions. Then give
$\mathcal{D}(\mathcal{T})$ the inductive limit topology
$$
\mathcal{D}(\mathcal{T})=\lim_{n\to\infty}\mathcal{D}_{\mathcal{O}_n}(\mathcal{T}).
$$
Here we are done. Distributions $\mathcal{D}(\mathcal{T})'$ and
operations on them can be defined as usual. The same construction
can be done for $\mathcal{D}(\mathcal{T}_t)$ with minor
modifications.

At the end let us consider the choice of the topology in detail.
In the literature one usually chooses the family of norms
$\|.\|_p$ (or sometimes a family of seminorms $|\!(.)\!|_p$; from
these seminorms one can make norms by
$\|.\|_p=\sum_{q<p}|\!(.)\!|_q$ or $\|.\|_p=\sup_{q<p}|\!(.)\!|_q$
etc.) rather arbitrarily in accordance with the setup of the
problem, and it is tacitly assumed but not everywhere proven, that
all such choices give equivalent topologies. Let us for
consistency present here a proof of this fact. The zest of the
proof (the usage of the Sobolev embedding theorem) was suggested
by G. Folland.

\begin{proposition}\label{CountNormTopEquiv}
Let $\mathcal{T}\xrightarrow{\pi}M$ be an $n$ dimensional
pseudo-Riemannian vector bundle over the $d$-dimensional
parallelizable pseudo-Riemannian manifold $M$ with positive metric
product $(,)_\mathfrak{g}$ constructed as above, so that we have
well defined $L^m$ norms $|\|.\||_m$ for $1\le m\le\infty$ on
$\mathcal{D}(\mathcal{T})$. Let $\nabla$ be a connection on
$\mathcal{T}$. Let
\begin{romanlist}
\item $X_1...X_d$ be a system of first order smooth differential
operators on $C^\infty(\mathcal{T})$ which span the tangent space
$T^*M$ everywhere

\item the seminorms be given by $|\!(
f)\!|_{\alpha,q}=|\|P_{\alpha,q}(X_i)f\||_m$, where
$P_{\alpha,q}(X_i)$ are various monomials of order $q$ in
$\{X_i\}$, $f\in\mathcal{D}(\mathcal{T})$

\item the family of norms be given as
$\|f\|_p=|\|\{|\!(f)\!|_{\alpha,q}\}_{q\le p}\||_{l^k}$, or by a
superposition of different $|\|.\||_{l^k}$, $1\le k\le\infty$.
\end{romanlist}

Then the topology of $\mathcal{D}(\mathcal{T})$ defined by this
family of norms is independent of the decisions (i) to (iii).
\end{proposition}

\begin{proof}
For convenience denote by $(X_i,m,\ast)$ the triple of choices at
points (i),(ii) and (iii). Then $(X_i,m,\ast)\sim(X_i',m',\ast')$
will mean that this two topologies are equivalent.

As the topology of $\mathcal{D}(\mathcal{T})$ is the inductive
limit of various $\mathcal{D}(\mathcal{T}_K)$ with
$\mathcal{T}_K=\pi^{-1}(K)$, $K\subset M$ compact, it suffices to
prove the assertion for an arbitrary $\mathcal{D}(\mathcal{T}_K)$.
The topologies given by two families of norms $\{\|.\|_p\}$ and
$\{\|.\|_p'\}$ are equivalent if and only if these two systems of
norms are themselves equivalent, i.e., $\forall p$, $\exists
q(p),r(p)>0$, $0<C_p,C_p'\in\mathbb{R}$ such that $\|.\|_p\le
C_p\|.\|_{q(p)}$ and $\|.\|_p'\le C_p'\|.\|_{r(p)}'$. Let us start
with the point (iii). Suppose the choices (i) and (ii) are fixed,
i.e., consider $(X_i,m,\ast)$ and $(X_i,m,\ast')$. Then all
possible choices in (iii) give equivalent systems of norms because
of the elementary inequalities
$$
|\|\{|\!(f)\!|_{\alpha,q}\}_I\||_{l^\infty}\le...\le|\|\{|\!(f)\!|_{\alpha,q}\}_I\||_{l^k}\le...\le|\|\{|\!(f)\!|_{\alpha,q}\}_I\||_{l^1}\le
N_I|\|\{|\!(f)\!|_{\alpha,q}\}_I\||_{l^\infty},
$$
where $N_I$ is the number of terms in the index set $I$. These
inequalities can be applied consecutively to estimate any
composite norm by, say, $|\|.\||_{l^\infty}$. An example of a
composite norm is $\|f\|_p=\sup_{q\le p}|\|\nabla^q f\||_\infty$.
We found that $(X_i,m,\ast)\sim(X_i,m,\ast')$.

Now let $1\le m\le\infty$ at (ii) and $k=\infty$ at (iii) be
chosen, and choose two systems of operators $\{X_i\}$ and
$\{Y_i\}$ at point (i) to construct the families of norms
$\{\|.\|_p\}$ and $\{\|.\|_p'\}$, respectively. This corresponds
to $(X_i,m,l^\infty)$ and $(Y_i,m,l^\infty)$. Because $\{X_i\}$
spans $T^*M$, there are functions $c_{ij}(x)\in C^\infty(M)$ and
smooth fields of homomorphisms $\tilde\Gamma_i\in
C^\infty(Hom(\mathcal{T},\mathcal{T}))$ with
$Y_i(x)=\sum_jc_{ij}(x)X_j(x)+\tilde\Gamma_i$. Using this for any
monomial $P_{\alpha,q}(Y_i)$ we get
$$
P_{\alpha,q}(Y_i)f=\sum_\beta
c_{\alpha,q}^\beta(x)Q_{\alpha,q}^\beta(X_i)f,
$$
where $c_{\alpha,q}^\beta(x)\in C^\infty(M)$ and
$Q_{\alpha,q}^\beta(X_i)$ are monomials of order less or equal
$q$. The number of summands is less than, say, $(4d)^q$. It
follows by Minkowsky inequality
$$
|\!(
f)\!|_{\alpha,q}'=|\|P_{\alpha,q}(Y_i)f\||_m\le\sum_\beta|\|c_{\alpha,q}^\beta(x)Q_{\alpha,q}^\beta(X_i)f\||_m,
$$
and then by H\"older inequality
$$
\sum_\beta|\|c_{\alpha,q}^\beta(x)Q_{\alpha,q}^\beta(X_i)f\||_m\le
C_{\alpha,q}\sum_\beta|\|Q_{\alpha,q}^\beta(X_i)f\||_m=C_{\alpha,q}\sum_\beta|\!(
f)\!|_{\alpha(\alpha,q,\beta),q(\alpha,q,\beta)},
$$
where $0<C_{\alpha,q}=\sup_\beta|\|c_{\alpha,q}^\beta\||_\infty$.
In other words, the seminorms of order $q$ of the second system
can be estimated by linear combinations of seminorms of the first
system of the same or lower order. Then
$$
\|f\|_p'=\sup_{q\le p}|\!(f)\!|_{\alpha,q}'\le C_p\sup_{q\le
p}\sum_\beta|\!(
f)\!|_{\alpha(\alpha,q,\beta),q(\alpha,q,\beta)}\le
$$
$$
\le C_p(4d)^p\sup_{q\le p}|\!(
f)\!|_{\alpha(\alpha,q,\beta),q(\alpha,q,\beta)}\le
C_p(4d)^p\sup_{q\le p}|\!(f)\!|_{\alpha,q}=C_p(4d)^p\|f\|_p,
$$
where $0<C_p=\sup_{q\le p}C_{\alpha,q}$. For the other direction
of the estimate we simply need to switch $\{X_i\}$ and $\{Y_i\}$.
Thus these two topologies are equivalent,
$(X_i,m,l^\infty)\sim(Y_i,m,l^\infty)$.

Finally let $X_i=\nabla_i$ (components with respect to a global
orthonromal frame in $T^*M$) be chosen at (i), and
$\|.\|=|\|\{|\!(.)\!|_{\alpha,q}\}_{q\le p}\||_{l^2}$ at (iii). We
construct two families of norms by choosing $1\le m<\infty$ and
$m'=\infty$ at (ii) for $\|.\|_p$ and $\|.\|_p'$, respectively.
This can be symbolized as $(\nabla_i,m,l^2)$ and
$(\nabla_i,\infty,l^2)$. Because $K$ is compact, by an application
of H\"older inequality we obtain
$$
|\|.\||_m\le C_m|\|.\||_\infty
$$
for some $0<C_m\in\mathbb{R}$, and hence obviously
$$
\|.\|_p\le C_m\|.\|_p'\mbox{, }p\in\mathbb{N}_0.
$$
The opposite inequality requires an application of Sobolev
embedding theorem for compact manifolds
\cite{Hebey1999},\cite{Taylor1991}. Denote the Sobolev norms
(which are equivalent to those in \cite{Hebey1999})
$$
|\|f\||_{W^{p,m}}=\sqrt{\sum_{q\le p}|\|\nabla^q f\||_m^2}.
$$
Then an application of Sobolev embedding theorem gives
$$
|\|.\||_{W^{0,\infty}}=|\|.\||_\infty\le D|\|.\||_{W^{d,1}}
$$
for some $0<D\in\mathbb{R}$. By another application of H\"older
inequality we find
$$
|\|.\||_{W^{d,1}}\le|\|.\||_{W^{d,2}},
$$
and therefore
$$
|\|.\||_\infty\le D\sqrt{\sum_{q\le d}|\|\nabla^q f\||_2^2}.
$$
Next
$$
|\|\nabla^q f\||_2^2=\sum_\alpha|\|P_{\alpha,q}(X_i)f\||_2^2,
$$
and finally
$$
\|f\|_p'=\sqrt{\sum_{q\le p}|\|P_{\alpha,q}(X_i)f\||_\infty^2}\le
D\sqrt{\sum_{q\le p}\sum_{j\le d}|\|\nabla^j
P_{\alpha,q}(X_i)f\||_2^2}=
$$
$$=D\sqrt{\sum_{q\le p}\sum_{j\le
d}|\|P_{\beta,j}(X_i)P_{\alpha,q}(X_i)f\||_2^2}\le
D\sqrt{\sum_{q\le p+d}|\|P_{\alpha,q}(X_i)f\||_2^2}\le
$$
$$
\le DE\sqrt{\sum_{q\le
p+d}|\|P_{\alpha,q}(X_i)f\||_m^2}=DE\|f\|_{p+d}',
$$
where in the last inequality again H\"olders inequality was used
with some $0<E\in\mathbb{R}$. Thus we have shown that choosing any
$1\le m<\infty$ is equivalent to choosing $m=\infty$ at point
(ii), i.e., $(\nabla_i,m,l^2)\sim(\nabla_i,\infty,l^2)$.

Write
$$
(X_i,m,\ast)\sim(X_i,m,l^\infty)\sim(\nabla_i,m,l^\infty)\sim(\nabla_i,m,l^2)\sim(\nabla_i,\infty,l^2)\sim(\nabla_i,m',l^2)\sim
$$
$$
\sim(\nabla_i,m',l^\infty)\sim(X_i',m',l^\infty)\sim(X_i',m',\ast').
$$
The proof is complete.
\end{proof}

\section{On the time dependent harmonic oscillator}

Here we will concentrate on some properties of the solutions of
the smooth complex time dependent harmonic oscillator
equation\index{Time dependent harmonic oscillator}
\begin{eqnarray}
\ddot T(s)+\Lambda(s)T(s)=0\label{THDO}
\end{eqnarray}
where $\Lambda(s)$ is a smooth complex function on the real line.
This equation is under attention since a long time, but some
results are not that easily available today (at least for us).

We start with an easy remark. Denote by
$$
W[Q,R](s)=\begin{pmatrix}
Q(s) & \dot Q(s)\\
R(s) & \dot R(s)
\end{pmatrix}
$$
the Wronski matrix of two solutions $Q$ and $R$.

\begin{remark}\label{TbyW_QR}
Let $Q,R$ be two linearly independent solutions of
(Eq.\ref{THDO}), and $T$ an arbitrary solution. Then from the
conservation of $\det W[Q,T]$ and $\det W[R,T]$ it is easy to find
$$
\begin{pmatrix}
\dot T(s)\\
-T(s)
\end{pmatrix}
=W[Q,R]^{-1}(s)\times W[Q,R](0)\times
\begin{pmatrix}
\dot T(0)\\
-T(0)
\end{pmatrix}=
$$
$$
=\det W[Q,R]^{-1}(0)
\begin{pmatrix}
\dot R(s) & -\dot Q(s)\\
- R(s) & Q(s)
\end{pmatrix}\times
\begin{pmatrix}
Q(0) & \dot Q(0)\\
R(0) & \dot R(0)
\end{pmatrix}\times
\begin{pmatrix}
\dot T(0)\\
-T(0)
\end{pmatrix}.
$$
Thus having at hand two such particular solutions $Q,R$, we have a
control over arbitrary solutions $T$ in terms of their initial
data.
\end{remark}

Our first task is to obtain a control over the magnitude of the
solution $T$ on a given compact interval $\mathcal{R}$ in terms of
its initial data $T(0)$ and $\dot T(0)$. This is done by the so
called energy estimate. Define the energy of a solution $T$ by
$$
\mathcal{W}[\breve T](s)=\frac{1}{2}|\dot
T|^2(s)+\frac{1}{2}\Re\Lambda(s)|T(s)|^2.
$$
If $\Re\Lambda>0$ on $\mathcal{R}$ then $2\mathcal{W}[\breve T]$
dominates $\Re\Lambda|T|^2$ and $|\dot T|^2$, and obtaining bounds
on $\mathcal{W}[\breve T]$ we automatically get bounds on $|T|$
and $|\dot T|$.

\begin{proposition}\label{EnergyEstCompl}
For arbitrary solution $T$ of
$$
\ddot T(s)+\Lambda(s)T(s)=0,
$$
with smooth complex valued $\Lambda(s)$ having a positive real
part (i.e., $\Re\Lambda(s)>0$) on a compact interval
$\mathcal{R}$, the energy function $\mathcal{W}[\breve T](s)$
satisfies the estimate
$$
\mathcal{W}[\breve
T](0)e^{-\int_0^sd\sigma(\frac{2|\Im\Lambda(\sigma)|}{\sqrt{\Re\Lambda(\sigma)}}+|\partial_s\ln\Re\Lambda(\sigma)|)}\le\mathcal{W}[\breve
T](s)\le\mathcal{W}[\breve
T](0)e^{\int_0^sd\sigma(\frac{2|\Im\Lambda(\sigma)|}{\sqrt{\Re\Lambda(\sigma)}}+|\partial_s\ln\Re\Lambda(\sigma)|)}
$$
for all $s\in\mathcal{R}$.
\end{proposition}

\begin{proof}
Write $T(s)=R(s)+iS(s)$, $\Lambda(s)=\Theta(s)+i\Xi(s)$, and
insert into the equation. We will get the following system of real
equations,
$$
\begin{cases}
\ddot R(s)+\Theta(s)R(s)-\Xi(s)S(s)=0,\\
\ddot S(s)+\Theta(s)S(s)+\Xi(s)R(s)=0.
\end{cases}
$$
We can cast this into a real vector equation
$$
\ddot{\breve{T}}(s)+\hat\Lambda(s)\breve T(s)=0
$$
by denoting
$$
\breve T(s)=(R(s),S(s))^\top,
$$
and
$$
\hat\Lambda(s)=
\begin{pmatrix}
\Theta(s) & -\Xi(s)\\
\Xi(s) & \Theta(s)
\end{pmatrix}=\hat\Lambda^+(s)+\hat\Lambda^-(s)=
\Theta(s)\id+\Xi(s)\begin{pmatrix}
0 & -1\\
1 & 0
\end{pmatrix},
$$
where $\hat\Lambda^\pm$ denote the symmetric and antisymmetric
parts. The energy function equals
$$
\mathcal{W}[\breve
T](s)=\frac{1}{2}\dot{\breve{T}}^2(s)+\frac{1}{2}\breve
T^\top(s)\hat\Lambda(s)\breve
T(s)=\frac{1}{2}\dot{\breve{T}}^2(s)+\frac{1}{2}\breve
T^\top(s)\hat\Lambda^+(s)\breve T(s).
$$
On the interval $\mathcal{R}$ we have $\mathcal{W}[\breve T](s)>0$
as by the assumption $\Theta(s)>0$. One can easily find that
$$
\dot{\mathcal{W}}[\breve T](s)=\breve
T^\top(s)\hat\Lambda^-(s)\dot{\breve{T}}(s)+\frac{1}{2}\breve
T^\top(s)\dot{\hat{\Lambda}}^+(s)\breve T(s),
$$
whence it follows
$$
\left|\dot{\mathcal{W}}[\breve T](s)\right|\le|\Xi(s)||\breve
T(s)||\dot{\breve{T}}(s)|+|\partial_s\ln\Theta(s)|\mathcal{W}[\breve
T](s).
$$
By definition of $\mathcal{W}[\breve T]$ and positivity of
$\Theta$ we have $|\dot{\breve{T}}(s)|\le\sqrt{2\mathcal{W}[\breve
T](s)}$ and $|\breve T(s)|\le\sqrt{2\mathcal{W}[\breve
T](s)/\Theta(s)}$ on $\mathcal{R}$. It follows then
$$
\left|\partial_s\ln\mathcal{W}[\breve
T](s)\right|\le\frac{2|\Xi(s)|}{\sqrt{\Theta(s)}}+|\partial_s\ln\Theta(s)|,
$$
and integrating this we finally arrive at
$$
\mathcal{W}[\breve
T](0)e^{-\left|\int_0^sd\sigma(\frac{2|\Xi(\sigma)|}{\sqrt{\Theta(\sigma)}}+|\partial_s\ln\Theta(\sigma)|)\right|}\le\mathcal{W}[\breve
T](s)\le\mathcal{W}[\breve
T](0)e^{\left|\int_0^sd\sigma(\frac{2|\Xi(\sigma)|}{\sqrt{\Theta(\sigma)}}+|\partial_s\ln\Theta(\sigma)|)\right|},
$$
precisely as in the statement.
\end{proof}

If however $\Lambda$ is not guaranteed to be positive, then on
those regions where it is negative the magnitude of the solutions
is expected to behave exponentially. We are able to capture that
exponential factor by the following beautiful trick.

\begin{proposition}\label{KappaTrick}
For any $0<\kappa\in\mathbb{R}$, any solution of the equation
$$
\ddot T(s)+\Lambda(s)T(s)=0
$$
can be represented as $T(s)=\tau(\frac{1}{\kappa}\Th(\kappa
s))\ch(\kappa s)$, where $\tau(z)$ is a solution of the equation
$$
\ddot\tau(z)+\Omega(z)\tau(z)=0
$$
with
$$
\Omega(z)=\frac{\kappa^2+\Lambda(\frac{1}{\kappa}\aTh(\kappa
z))}{(1-\kappa^2z^2)^2}\mbox{,
}z\in(-\frac{1}{\kappa},\frac{1}{\kappa}).
$$
\end{proposition}

\begin{proof}
The proof is elementary once we already know the clue: the
substitution of variables $\kappa z=\Th(\kappa s)$. The
substitution $T(s)=\tau(s)\ch(\kappa s)$ into the original
equation gives
$$
\ddot\tau(s)+2\kappa\Th(\kappa
s)\dot\tau(s)+(\kappa^2+\Lambda(s))\tau(s)=0,
$$
then the substitution $s\to z$ yields the final formulas.
\end{proof}

Let us say a couple of words about this. If $\Re\Lambda$ has a
minimal negative value $-c$ on some domain, then it suffices to
set $\kappa=\sqrt{c}$ to reduce the problem to an oscillatory
equation for $\rho$. The upper bound of the rate of exponential
expansion is precisely given by the square root of the minimal
negative value of $\Re\Lambda$.

Finally we combine these two statements to find an explicit
uniform bound on an arbitrary solution $T$. Let the compact
interval $\mathcal{R}$ containing 0 be fixed, and set
$$
A_\mathcal{R}=\sup_\mathcal{R}|\Im\Lambda|\mbox{,
}c_\mathcal{R}=\inf_\mathcal{R}\Re\Lambda\mbox{,
}\kappa=\sqrt{1+|\min\{0,c_\mathcal{R}\}|}\mbox{,
}B_\mathcal{R}=\sup_\mathcal{R}\left|\partial_s\ln\left(\kappa^2+\Re\Lambda\right)\right|,
$$
$$
D_\mathcal{R}=\sup_\mathcal{R}(\kappa^2+\Re\Lambda)\mbox{,
}e_\mathcal{R}=\inf_\mathcal{R}(\kappa^2+\Re\Lambda)=1+\max\{0,c_\mathcal{R}\},
$$
$$
L_\mathcal{R}=\left(2A_\mathcal{R}+\ch^2(\kappa|\mathcal{R}|)B_\mathcal{R}+2\kappa\sh(2\kappa|\mathcal{R}|)\right)
$$
(we suppressed the index $\mathcal{R}$ of $\kappa$ for
convenience).

\begin{corollary}\label{TEstAbstract}
For an arbitrary solution $T$ it holds
$$
|T(s)|\le|T(0)|\sqrt{\frac{D_\mathcal{R}}{e_\mathcal{R}}}e^{L_\mathcal{R}}\ch(\kappa|\mathcal{R}|)+|\dot
T(0)|\frac{1}{\sqrt{e_\mathcal{R}}}e^{L_\mathcal{R}}\ch(\kappa|\mathcal{R}|)
$$
for all $s\in\mathcal{R}$.
\end{corollary}

\begin{proof}
Consider the linearly independent solutions $Q$ and $R$ given by
initial data
$$
Q(0)=1\mbox{, }\dot Q(0)=0\mbox{, }R(0)=0\mbox{, }\dot R(0)=1.
$$
Using Proposition \ref{KappaTrick} represent them as
$Q(s)=\xi(z(s))\ch(\kappa s)$ and $R(s)=\rho(z(s))\ch(\kappa s)$,
where $\xi(z)$ and $\rho(z)$ are solutions of the equation
$$
\ddot\tau(z)+\Omega(z)\tau(z)=0
$$
with
$$
\Omega(z)=\frac{\kappa^2+\Lambda(\frac{1}{\kappa}\aTh(\kappa
z))}{(1-\kappa^2z^2)^2}.
$$
Using $z(0)=0$ and
\begin{eqnarray}
\frac{d}{ds}\left[\tau(z(s))\ch(\kappa
s)\right]=\frac{\dot\tau(z(s))}{\ch(\kappa
s)}+\tau(z(s))\sh(\kappa s)\kappa,\label{dTbydtau}
\end{eqnarray}
we find
$$
\xi(0)=1\mbox{, }\dot \xi(0)=0\mbox{, }\rho(0)=0\mbox{, }\dot
\rho(0)=1.
$$
Note that $\Re\Omega(s)=\kappa^2+\Re\Lambda\ge1$, thus Proposition
\ref{EnergyEstCompl} is applicable for $\xi$ and $\rho$. We have
$\mathcal{W}[\xi](0)=\frac{1}{2}\Re\Omega(0)$ and
$\mathcal{W}[\rho](0)=\frac{1}{2}$. Now
$$
\frac{d}{dz}\ln\Re\Omega(z)=\frac{ds}{dz}(s)\frac{d}{ds}\ln\Re\Omega(z(s))=\ch^2(\kappa
s)\frac{d}{ds}\ln\left((\kappa^2+\Re\Lambda(s))\ch^4(\kappa
s)\right)=
$$
$$
=\ch^2(\kappa
s)\frac{d}{ds}\ln\left(\kappa^2+\Re\Lambda(s)\right)+2\kappa\sh(2\kappa
s).
$$
Then it follows
$$
\left|\int_0^zd\sigma(\frac{2|\Im\Omega(\sigma)|}{\sqrt{\Re\Omega(\sigma)}}+|\partial_z\ln\Re\Omega(\sigma)|)\right|\le\frac{2}{\kappa}\Th(\kappa|\mathcal{R}|)\left(2A_\mathcal{R}+\ch^2(\kappa|\mathcal{R}|)B_\mathcal{R}+2\kappa\sh(2\kappa|\mathcal{R}|)\right)\le
$$
$$
\le2\left(2A_\mathcal{R}+\ch^2(\kappa|\mathcal{R}|)B_\mathcal{R}+2\kappa\sh(2\kappa|\mathcal{R}|)\right)=2L_\mathcal{R}.
$$
By Proposition \ref{EnergyEstCompl} we have
$$
\mathcal{W}[\xi](z)\le\mathcal{W}[\xi](0)e^{2L_\mathcal{R}}\mbox{,
}\mathcal{W}[\rho](z)\le\mathcal{W}[\rho](0)e^{2L_\mathcal{R}},
$$
which entails
$$
|\xi(z)|\le\sqrt{\frac{\kappa^2+\Re\Lambda(0)}{\kappa^2+\Re\Lambda(s(z))}}e^{L_\mathcal{R}}\mbox{,
}|\dot\xi(z)|\le\sqrt{\kappa^2+\Re\Lambda(0)}e^{L_\mathcal{R}},
$$
$$
|\rho(z)|\le\frac{1}{\sqrt{\kappa^2+\Re\Lambda(s(z))}}e^{L_\mathcal{R}}\mbox{,
}|\dot\rho(z)|\le e^{L_\mathcal{R}}.
$$
For $Q$ and $R$ we get
$$
|Q(s)|\le\sqrt{\frac{D_\mathcal{R}}{e_\mathcal{R}}}e^{L_\mathcal{R}}\ch(\kappa|\mathcal{R}|)\mbox{,
}|R(s)|\le\frac{1}{\sqrt{e_\mathcal{R}}}e^{L_\mathcal{R}}\ch(\kappa|\mathcal{R}|),
$$
and using (Eq.\ref{dTbydtau})
$$
|\dot
Q(s)|\le\sqrt{D_\mathcal{R}}e^{L_\mathcal{R}}\left(1+\frac{\kappa\sh(\kappa|\mathcal{R}|)}{\sqrt{e_\mathcal{R}}}\right)\mbox{,
}|\dot R(s)|\le
e^{L_\mathcal{R}}\left(1+\frac{\kappa\sh(\kappa|\mathcal{R}|)}{\sqrt{e_\mathcal{R}}}\right).
$$
Finally let $T$ be an arbitrary solution of the original equation.
Applying Remark \ref{TbyW_QR} for $Q$,$R$ and $T$ we find
$$
T(s)=T(0)Q(s)+\dot T(0)R(s),
$$
and hence
$$
|T(s)|\le|T(0)|\sqrt{\frac{D_\mathcal{R}}{e_\mathcal{R}}}e^{L_\mathcal{R}}\ch(\kappa|\mathcal{R}|)+|\dot
T(0)|\frac{1}{\sqrt{e_\mathcal{R}}}e^{L_\mathcal{R}}\ch(\kappa|\mathcal{R}|),
$$
as asserted.
\end{proof}

\section{A result from functional calculus}

In this section we will obtain a result using the theory of
holomorphic functional calculus of strip type operators. We are
grateful to M. Haase for very useful comments on this theory, and
refer to his book \cite{Haase2006} for all the information
necessary in this section.

Let $\mathbb{H}_a=\{z\in\mathbb{C}: |\Im z|<a\}$ denote the
symmetric strip of height $a>0$. If for an (unbounded) operator
$A$ on the Banach space $\mathcal{X}$ we have $A\in\Strip(a)$,
then we can apply the holomorphic functional
calculus\index{Functional calculus} of $A$ given by
$$
F(A)=\frac{1}{2\pi i}\int_{\gamma_a}dzf(z)\mathfrak{R}(z,A)\mbox{,
}\forall F\in\mathcal{M}[\mathbb{H}_a],
$$
where $\gamma_a=\partial\mathbb{H}_a$ oriented positively
(counterclockwise), and $\mathfrak{R}(z,A)$ is the resolvent of
$A$ for $z\in\mathbb{C}$. Define
$$
\mathcal{A}(\mathbb{H}_a)=\{F\in\Hol(\mathbb{H}_a): \exists
N\in\mathbb{N}\mbox{ s.t. }F=O(|\Re z|^N)\},
$$
and
$$
\mathcal{A}[\mathbb{H}_a]=\bigcup_{b>a}\mathcal{A}(\mathbb{H}_b).
$$

Now let $D_{\Sigma_t}=-\Delta+m^\star(t,\vec x)$ be the known real
lower semi-bounded operator acting on the vector bundle
$\mathcal{T}_t$ over a Riemannian manifold $\Sigma_t$, and let
$K\subset\Sigma_t$ be a compact region. Denote
$$
\mathcal{D}(K)=\{f\in\mathcal{D}(\mathcal{T}_t): \supp f\subset
K\}.
$$
Then we have the following result.

\begin{proposition}\label{FuncCalcProp}
For any $F\in\mathcal{A}[\mathbb{H}_0]$ and $f\in\mathcal{D}(K)$
it follows
$$
F(D_{\Sigma_t})f\in\mathcal{D}(K).
$$
\end{proposition}

\begin{proof}
Let the nuclear topology be given by $(X_i,2,l^2)$, i.e., for any
$p\in\mathbb{N}_0$ we set
$$
(f,h)_p=\sum_{q\le p}(Q_{\alpha,q}(X_i)f,Q_{\alpha,q}(X_i)h)_{L^2}
$$
and consider the induced norms $\|.\|_p$. Define the Hilbert
spaces
$$
\mathcal{H}_p=\overline{\mathcal{D}(K)}^{(,)_p},
$$
then by the property of the countably normed spaces we have
$$
\mathcal{H}_p\subset\mathcal{H}_q\mbox{, }q<p,
$$
$$
\mathcal{D}(K)=\bigcap_{p=0}^\infty\mathcal{H}_p.
$$
Fix $p$, and define the operator $D_p$ on $\mathcal{H}_p$ by
setting $D_p f=D_{\Sigma_t} f$ for all
$f\in\Dom(D_{\Sigma_t})\cap\mathcal{H}_p$, then
$\Dom(D_p)\supset\mathcal{H}_{p+2}$ is a dense subspace of
$\mathcal{H}_p$. Then $D_p$ is a real symmetric operator, and
hence by von Neumann's theorem possesses a self-adjoint extension
$A_p$ which needs not be lower semi-bounded. The self-adjoint
operator $A_p$ has a purely real spectrum, thus $A_p\in\Strip(0)$.
Let $\mathcal{A}(\mathbb{H}_a)\ni F(z)=O(|\Re z|^N)$, then for a
sufficiently large $a<\lambda\in\mathbb{R}$, the function
$e(z)=(z-i\lambda)^{-(N+2)}$ will regularize $F$ on
$\mathbb{H}_a$. In particular, we will have
$[eF](A_p)\in\mathcal{B}(\mathcal{H}_p)$. Then
$F(A_p)=(A_p-i\lambda)^{N+2}[eF](A_p)=[eF](A_p)(A_p-i\lambda)^{N+2}$,
from where it follows that $\Dom(A_p^{N+2})\subset\Dom(F(A_p))$.
From the definition of $A_p$ it is clear that
$\mathcal{H}_{p+2(N+2)}\subset\Dom(A_p^{N+2})$, whence
$\mathcal{H}_{p+2(N+2)}\subset\Dom(F(A_p))$. Thus we have
established, that whenever $f\in\mathcal{H}_{p+2(N+2)}$, then
necessarily $F(A_p)f\in\mathcal{H}_p$. Now if
$f\in\mathcal{D}(K)$, then for any $p\ge0$ we have
$f\in\mathcal{H}_{p+2(N+2)}$, and hence $F(A_p)f\in\mathcal{H}_p$.
Meanwhile for any $p\ge0$, the self-adjoint operator
$D_{\Sigma_t}$ agrees with $A_p$ on $\mathcal{D}(K)$. Therefore
also their functional calculi agree,
$F(D_{\Sigma_t})f=F(A_p)f\in\mathcal{H}_p$. Thus
$$
F(D_{\Sigma_t})f\in\bigcap_{p=0}^\infty\mathcal{H}_p=\mathcal{D}(K),
$$
which completes the proof.
\end{proof}

\bibliographystyle{ieeetr}
\bibliography{lib}

\end{document}